\documentclass[11pt]{article}
\usepackage[left=3cm,right=3cm,top=3cm,bottom=3cm]{geometry} 
\usepackage{amsmath,amssymb} 
\usepackage{amsmath, amsthm,amssymb,amsfonts,latexsym}
\usepackage{amsthm}
\usepackage{comment}
\usepackage{tikz}
\usepackage{lipsum}
\usepackage{color}
\usepackage{hyperref}
\usepackage[capitalise]{cleveref}
\usetikzlibrary{patterns}
\usetikzlibrary{intersections}
\usetikzlibrary{arrows,positioning}
\usepackage{graphicx}
\usepackage{pgfplots}
\pgfplotsset{compat=1.14}
\usetikzlibrary{plotmarks}
  \usepgfplotslibrary{patchplots}
  \usepackage{grffile}
\usepgfplotslibrary{fillbetween}

\usepackage{float}
\usepackage{caption}
\usepackage{subcaption}

\newcommand*{\rom}[1]{\expandafter\@slowromancap\romannumeral #1@}

\newcommand{\gettikzxy}[3]{%
  \tikz@scan@one@point\pgfutil@firstofone#1\relax
  \edef#2{\the\pgf@x}%
  \edef#3{\the\pgf@y}%
}

\pgfmathdeclarefunction{gauss}{3}{%
  \pgfmathparse{1/(#3*sqrt(2*pi))*exp(-((#1-#2)^2)/(2*#3^2))}%
}
\pgfmathdeclarefunction{gaussX}{3}{%
  \pgfmathparse{1/(#3*sqrt(2*pi))*exp(-((#1-#2)^2)/(2*#3^2))*1}%
}

\def\Nn{50}

\newcommand{\sA}{\mathcal A}
\newcommand{\sB}{\mathcal B}
\newcommand{\sC}{\mathcal C}
\newcommand{\sD}{\mathcal D}

\newcommand{\sF}{\mathcal F}

\newcommand{\sH}{\mathcal H}
\newcommand{\sI}{\mathcal I}

\newcommand{\sL}{\mathcal L}
\newcommand{\sM}{\mathcal M}

\newcommand{\sP}{\mathcal P}

\newcommand{\sS}{\mathcal S}
\newcommand{\sT}{\mathcal T}

\newcommand{\R}{\mathbb R}
\newcommand{\E}{\mathbb E}
\newcommand{\NN}{\mathbb N}
\newcommand{\bF}{\mathbb F}
\newcommand{\bT}{\mathbb T}
\newcommand{\QQ}{\mathbb Q}

\newcommand{\PP}{\mathbb P}

\newtheorem{theorem}{Theorem}[section]
\newtheorem{prop}[theorem]{Proposition}

\newtheorem{lem}[theorem]{Lemma}
\newtheorem{cor}[theorem]{Corollary}
\newtheorem{remark}[theorem]{Remark}
\newtheorem{rem}[theorem]{Remark}
\newtheorem{sass}{Standing Assumption}
\newtheorem{ass}[theorem]{Assumption}
\newtheorem{defn}[theorem]{Definition}
\newtheorem{prob}{Problem}

\newcommand{\cgreen}{\color{green}}
\newcommand{\cblue}{\color{blue}}

\newcommand{\cpur}{\color{purple}}

\begin{document}
\title{Model-independent upper bounds for the prices of Bermudan options with convex payoffs}
\author{
David Hobson\thanks{Department of Statistics, University of Warwick \textit{d.hobson@warwick.ac.uk}} \hspace{10mm}
Dominykas Norgilas\thanks{Department of Mathematics, NC State University \textit{dnorgil@ncsu.edu}}
}
\date{\today}
\maketitle

\begin{abstract}
Suppose $\mu$ and $\nu$ are probability measures on $\R$ satisfying $\mu \leq_{cx} \nu$. Let $a$ and $b$ be convex functions on $\R$ with $a \geq b \geq 0$. We are interested in finding
\[ \sup_{\sM} \sup_{\tau} \E^{\sM} \left[ a(X) I_{ \{ \tau = 1 \} } + b(Y) I_{ \{ \tau = 2 \} } \right] \]
where the first supremum is taken over consistent models $\sM$ (i.e., filtered probability spaces $(\Omega, \sF, \bF, \PP)$ such that $Z=(z,Z_1,Z_2)=(\int_{\R} x \mu(dx) = \int_{\R} y \nu(dy), X, Y)$ is a $(\bF,\PP)$ martingale, where $X$ has law $\mu$ and $Y$ has law $\nu$ under $\PP$) and $\tau$ in the second supremum is a $(\bF,\PP)$-stopping time taking values in $\{1,2\}$.

Our contributions are first to characterise and simplify the dual problem, and second to completely solve the problem {under some structural assumptions on the measures $\mu$ and $\nu$ (namely that $\mu$ and $\nu$ are absolutely continuous probability measures that satisfy the Dispersion Assumption).}
A key finding is that the canonical set-up in which the filtration is that generated by $Z$ is not rich enough to define an optimal model and additional randomisation is required. This holds even though the  marginal laws $\mu$ and $\nu$ are atom-free.

The problem has an interpretation of finding the robust, or model-free, no-arbitrage bound on the price of a Bermudan option with two possible exercise dates, given the prices of co-maturing European options. 

\end{abstract}
\tableofcontents

\section{The model-free approach to derivative pricing: problem motivation}

Suppose $S = (S_t)_{t \geq 0}$ is the price process of a risky asset in a financial market with riskless bank account paying deterministic rate of interest $r = (r_t)_{t \geq 0}$. According to standard no-arbitrage theory the price of a call option with strike $k$ and maturity $T$ (i.e., a payoff of $(S_T-k)^+$ at time $T$) is given by $\E^{\QQ}[e^{- \int_0^T r_t dt} (S_T - k)^+]$ where $\QQ$ is a risk neutral measure and $\E^{\QQ}$ denotes expectations with respect to $\QQ$. In the classical approach we work on a filtered probability space $(\Omega,\sF,\bF = (\sF_t)_{0 \leq t \leq T},\PP)$  and assume that there exists an equivalent martingale measure $\QQ$ such that the discounted price process $Z = (Z_t)_{0 \leq t \leq T}$, defined by $Z_t = e^{- \int_0^t r_s ds} S_t$, is a $(\bF,\QQ)$-martingale (and $\QQ$ is equivalent to $\PP$ on $\sF=\sF_T$). Then $\E^{\QQ}[e^{- \int_0^T r_t dt} (S_T - k)^+] = \E^{\QQ}[(Z_T - {K})^+]$, where ${K} = ke^{- \int_0^T r_t dt}$ is the discounted strike. In a complete market, the price of the call can be justified as the lowest price with which it is possible to replicate the call option. As a simple example, $(\Omega,\sF,\bF = (\sF_t)_{0 \leq t \leq T},\PP)$ may support a Brownian motion $W$ and if $S$ is given by $S_t = S_0 e^{\sigma W_t + \mu t}$ and $r$ is constant then call prices are given by the Black-Scholes option pricing formula.

In well-functioning markets vanilla option prices are not given by a model, but rather are fixed by supply and demand. Then we may still have $C(K,T) = \E^\QQ[e^{- \int_0^T r_t dt} (S_T - K e^{\int_0^T r_t dt})^+] = \E^\QQ[(Z_T - K)^+]$, but now it is the call prices which are given (as traded prices on the financial market) and the probabilistic model as represented by $\QQ$ which is unknown---typically we care about the risk-neutral probabilities $\QQ$ rather than the physical measure $\PP$. Nonetheless, if the set of call prices is sufficiently rich, then we can infer quantities such as $\QQ(Z_T>K)$. As Breeden and Litzenberger~\cite{BreedenLitzenberger:78} conclude, this means that we do not need a model to price an option with payoff $a(S_T)$ at time $T$: instead we can write it as a combination of call and put payoffs whose prices are known.

What can we say about the prices of exotic or path-dependent options? Assume that we are given the prices of a class of derivatives (these become our vanilla, liquidly traded derivatives whose prices can be observed in the financial market), and that there exists a stochastic model such that in the model the discounted price process $Z = (Z_t)_{t \geq 0}$ is a martingale under the risk-neutral measure $\QQ$ and the prices of vanilla derivatives are given by expectations under $\QQ$. Then the expected payoff under $\QQ$ is a candidate price for the exotic option. But, there may be many models which are consistent with the given prices of the vanilla derivatives. Then the robust derivative pricing problem becomes to find the supremum (and infimum) of the possible prices given by expectation, where the supremum (resp., infimum) 
is taken over all models (for which $Z$ is a martingale) which agree with the quoted prices of the vanilla options in the sense that the expected discounted payoff under the model agrees with the traded price for each vanilla derivative. See Hobson~\cite{Hobson:11} for a survey of this approach.

Suppose that the time index set is $\bT= \{ 0,1,2 \}$, suppose that the initial price of the risky asset is known, and that we know the prices of call and put options of all strikes with maturities $T=1$ and $T=2$. This is a reasonable class to take as the class of vanilla options. Then, with $Z_1 = X$ and $Z_2=Y$, we know $C(K,1) = \E^{\QQ}[(X-K)^+]$ and $C(K,2) = \E^{\QQ}[(Y-K)^+]$ for all $K>0$. It follows that we know the laws of both $X$ and $Y$ (but note that we have no information about the joint law beyond the marginals). We denote these laws by $\mu$ and $\nu$, respectively. We also know that $(Z_0=z,Z_1,Z_2)$
is a martingale. It then follows that $Z_0 = \int x \mu(dx) = \int y \nu(dy)$ and that $\mu$ and $\nu$ are in convex order\footnote{Two integrable (Borel) measures $\eta,\chi$ on $\R$, with $\eta(\R)=\chi(\R)<\infty$, are in convex order ($\eta\leq_{cx}\chi$) if $\int fd\eta\leq\int f d\chi$ for all convex $f:\R\to\R$.}, denoted by $\mu \leq_{cx} \nu$. For a given (Borel) $c:\R^2\to\R$ and a path-dependent random payoff $c = c(X,Y)$, the problem is to find 
$\sup \E[ c(X,Y) ]$, where the supremum is taken over possible joint laws of $(X,Y)$ which respect the marginals $(X \sim \mu,Y \sim \nu)$ and the martingale property $\E[Y|X] = X$. The case $c(x,y)=\pm|y-x|$ corresponding to a forward start straddle was studied by Hobson and Neuberger~\cite{HobsonNeuberger:12} and Hobson and Klimmek~\cite{HobsonKlimmek:15} (see also Beiglb\"ock and Juillet~\cite{BeiglbockJuillet:16} and Henry-
Labord\`ere and Touzi~\cite{HenryLabordereTouzi:16}, where the authors construct a model that is optimal for a certain (but large) class of cost functions $c$). More generally, this is the martingale optimal transport problem, as introduced by Beiglb\"ock et al~\cite{BHLP:13} and Galichon et al~\cite{GHLT:14}.

One fruitful approach to the martingale optimal transport problem is via the dual. In the context of the previous paragraph, the dual approach involves searching for
univarite functions $\phi, \psi$ and $\theta$ such that
\begin{equation}
    \label{eq:mot}
    c(x,y) \leq \phi(x) + \psi(y) + \theta(x) (y-x),\quad  x,y \in \R.
\end{equation}
If \eqref{eq:mot} holds and $\E[Y|X]=X$, then, since $\E[\theta(X)(Y-X)|X] = \theta(X)(\E[Y|X]-X) = 0$, we have $\E[c(X,Y)] \leq \E[\phi(X) + \psi(Y)]$. The primal problem of finding $\sP = \sup \E[c(X,Y)]$, where the supremum is taken 
over joint laws with $X \sim \mu$ and $Y \sim \nu$ which respect the martingale property, is thus related to the dual problem of finding $\sD = \inf \left( \int \phi(x) \mu(dx) + \int \psi(y) \nu(dy) \right)$, where the infimum is taken over all trios 
$(\phi, \psi, \theta)$ for which \eqref{eq:mot} holds. The martingale optimal transport literature 
is concerned with formalising the above set-up, with deriving sufficient conditions for strong duality $\sP = \sD$ (rather than the weak duality $\sP \leq \sD$, which follows very easily) and with (explicitly constructing or) characterising the form of the primal and dual optimisers (where they exist) for particular choices of objective function $c$.

In this article we are concerned with Bermudan-style payoffs in a two-period model. In the setting of the previous paragraph, given laws $\mu$ and $\nu$ in convex order, (Borel) functions $a,b:\R\to\R$, and setting $c(\cdot,1)=a$, $c(\cdot,2)=b$, the primal problem is to find
\begin{equation}\label{eq:primal} \sP = \sP(\mu,\nu;a,b) = \sup_{\sM \in M(\mu,\nu)} \sup_{\tau \in \sT_{1,2}} \E^{\sM}[ c(Z_\tau,\tau)], 
\end{equation}
where $M=M(\mu,\nu)$ is the set of  $(\mu,\nu)$-consistent models (recall, a model is a filtered probability space $(\Omega, \sF, \bF, \QQ)$ supporting a stochastic process $Z=(Z_0,Z_1,Z_2)$ such that $Z$ is a $(\bF,\QQ)$-martingale with given marginals $X \equiv Z_1 \sim \mu$ and $Y \equiv Z_2 \sim \nu$) and $\sT=\sT_{1,2}$ is the set of $\bF$-stopping times taking values in $\{1,2\}$.
As introduced in Neuberger~\cite{Neuberger:07} and Hobson and Neuberger~\cite{HobsonNeuberger:17}, the dual problem is to find 
\begin{equation}
\label{eq:dual}
\sD = \sD(\mu,\nu;a,b)  =  \inf_{\phi,\psi,\theta_1,\theta_2}  
\left(  \int \phi(x) \mu(dx) + \int \psi(y) \nu(dy) \right),
\end{equation}
where the infimum is taken over quadruples $(\phi,\psi,\theta_1,\theta_2): \R \to \R$ such that, for all $x,y\in\R$,
\begin{eqnarray}
    a(x) & \leq & \phi(x) + \psi(y) + \theta_1(x)(y-x), \label{eq:adual}\\
    b(y) & \leq & \phi(x) + \psi(y) + \theta_2(x)(y-x).
    \label{eq:bdual}
\end{eqnarray}
Then, if \eqref{eq:adual} and \eqref{eq:bdual} hold, for any $\sigma\in\sT$ we have that (almost surely)
\[ c(Z_\sigma,\sigma) = a(Z_\sigma)I_{ \{\sigma=1\} } + b(Z_\sigma) I_{ \{ \sigma= 2 \} } \leq \phi(Z_1) + \psi(Z_2) + \theta_\sigma(Z_1) (Z_2 - Z_1) .\]
In particular, whatever the stopping strategy of the Bermudan option holder, a hedging strategy of 
\begin{enumerate}
    \item holding a portfolio of call and put options with maturity $T=1$ and payoff $\phi$, 
    \item holding a portfolio of call and put options with maturity $T=2$ and payoff $\psi$,
    \item if the Bermudan option is exercised at $t=1$, holding $\theta_1 =\theta_1(Z_1)$ units of the risky asset between times one and two,
    \item otherwise, if the Bermudan option is not exercised at $t=1$ holding $\theta_2=\theta_2(Z_1)$ units of the risky asset between times one and two 
\end{enumerate}
is a superreplicating strategy. Furthermore, under each $(\mu,\nu)$-consistent model $\sM\in M(\mu,\nu)$ we have that $\E^\sM[c(Z_\sigma,\sigma)]\leq\int \phi(x) \mu(dx) + \int \psi(y) \nu(dy)$, from which the weak duality $\sP\leq\sD$ follows.

Neuberger~\cite{Neuberger:07} and Hobson and Neuberger~\cite{HobsonNeuberger:17} studied the Bermudan option pricing problem for assets taking values on a lattice and showed (using linear programming methods) that  there is no duality gap $\sP=\sD$. One of the key insights was that the filtration matters and it is not enough to simply consider the primal problem as one of finding the optimal martingale transport in the canonical filtration for the price process. The results in \cite{HobsonNeuberger:17} were re-proved and extended (e.g., to the non-lattice case) by Aksamit et al.~\cite{ADOT:19}, where the authors, instead of focusing on the filtration, considered the impact of enlarging the set of traded assets. They show that in a wide set of circumstances strong duality holds. Bayraktar et al.~\cite{BHZ:15} also consider the robust hedging of Bermudan-style options in a discrete-time framework. They consider both upper and lower bounds, but, since they restrict attention to a setting where the filtration is the canonical filtration, they find a duality gap---subsequently Bayraktar and Zhou~\cite{BZ:17} show that this gap can be removed if the set-up is extended to allow for randomized stopping times.

Hobson and Norgilas~\cite{HobsonNorgilas:19} studied the Bermudan option pricing problem for the case of put
options. (In the case when the risk-free interest rate is positive, a Bermudan call is trivial since the optimal strategy is to wait until maturity to exercise the call.) 
The authors showed that there is no duality gap, the model which achieves the highest price for the put is associated to the left-curtain martingale coupling of Beiglbock and Juillet~\cite{BeiglbockJuillet:16}, and it is possible to write down the cheapest superhedging portfolio. For a given strike for the Bermudan put, the optimal portfolio involves vanilla puts and calls with a finite number of strikes. The results in \cite{HobsonNorgilas:19} are obtained under the assumption that the initial law $\mu$ is atom-free, and, {in the setting of \cite{HobsonNorgilas:19},} 
it is enough to look for models that are equipped with the canonical filtration of the price process. Later, Hobson and Norgilas \cite{HobsonNorgilas:19x} extended the results of \cite{HobsonNorgilas:19} {to} the case of a general initial law $\mu$: in this setting, the optimal model is (still) associated to the \textit{lifted} left-curtain coupling, but the information generated by the price process {alone} is no longer sufficient and additional randomization is {required.}

In this paper we extend the results of \cite{HobsonNeuberger:17} and \cite{HobsonNorgilas:19} to general convex payoffs. 
Our results are in two directions. First, we show that the set of superreplicating strategies over which we search in the dual problem can be greatly simplified. Second, {under some structural assumptions on the laws $\mu$ and $\nu$ (which are trivially satisfied in the case where $\mu$ and $\nu$ are (in convex order and) both normal, log-normal or shifted exponential)
we characterize the optimal model and the optimal superhedge.} Based on \cite{HobsonNorgilas:19}, one could conjecture that, in the case the marginals $(\mu,\nu)$ are atom-free, it is enough to restrict the search {to} the set of models that are equipped with the canonical filtration of the price process. From our results, however, it follows that even in this regular case {there is no maximiser within this class.} 
Instead a richer class of models is required.

{The structural conditions that we require are that $\mu$ and $\nu$ satisfy the Dispersion Assumption  (see Definition~\ref{def:dispersion})
introduced in
Hobson and Klimmek \cite{HobsonKlimmek:15}.
This assumption states that there exist points $e^{L}=e^L(\mu,\nu)< e^R=e^R(\mu,\nu)$ such that if $\rho$ (resp., $\eta$) is the density of $\mu$ (resp., $\nu$) and $(\alpha_\mu,\beta_\mu)$ (resp.,   $(\alpha_\nu,\beta_\nu)$) is the support of $\mu$ (resp., $\nu$), then $\rho < \eta$ on $(\alpha_\nu,e^L) \cup (e^R,\beta_\nu)$ whereas $\eta < \rho$ on $(e^L,e^R)$.}


\textit{Notation}: for a measurable function $h$ we write $h^+$ for its positive part, and define $h^c$ to be the convex hull of $h$, so that $h^c$ is the largest convex function $H$ satisfying $H \leq h$; for a convex function $g$ we write $g'$ for its right-derivative -- we could in fact use any subdifferential. Sometimes we abbreviate $\int a(x) \mu(dx)$ to $\int a d\mu$.

For any $w,z\in \R$ with $w < z$, and two functions $g,h:\R\to \R \cup\{\infty\}$ with $g(w),h(z)\in\R$, the line that goes through $(w,g(w))$ and $(z,h(z))$ is denoted by $T^{w,z}_{g,h}:\R\to\R$. 

We denote the point mass at $x$ by $\delta_x$. If $f,g\in\R$ are such that $f<x < g$, then 
$\pi_x^{f,g}$, given by $\pi_x^{f,g} = \frac{x-f}{g-f} \delta_g + \frac{g-x}{g-f} \delta_f$, is the two-point distribution with support $\{ f,g \}$ and mean $x$.
If $f \leq x=g$ or $f=x < g$ then $\pi_x^{f,g} = \delta_x$.

For a (Borel) measure $\xi$ on $\R$, the barycenter of $\xi$ is defined by $\bar\xi=\int_\R x\xi(dx)/\xi(\R)$.

\section{Simplifying the dual problem}

In this section we want to study 
the cheapest superhedge for a Bermudan-style option which pays $a(Z_1)$ if exercised at time-1 and $b(Z_2)$ if exercised at time-2, where we assume that the prices of European options imply that $X \equiv Z_1$ has law $\mu$ and $Y \equiv Z_2$ has law $\nu$. Necessarily, we must have $\mu \leq_{cx} \nu$, so that both $\mu$ and $\nu$ are integrable (i.e., elements of $L^1$) and $\int x \mu(dx) = \int y \nu(dy)$. We also assume that both $a,b :\R\to\R_+$ are non-negative.

\begin{defn}[Superhedge, \protect{Hobson and Neuberger~\cite[Definition 2.7, 2.8]{HobsonNeuberger:17}, Hobson and Norgilas~\cite[Definition 2]{HobsonNorgilas:19}}] \label{def:superhedge}
$(\phi,\psi, \theta=\{ \theta_i \}_{i=1,2} )$ is a superhedge for the Bermudan option with payoff $(a,b)$ if
\eqref{eq:adual} and \eqref{eq:bdual} hold for all $x,y\in\R$.
\end{defn}
The terminology is explained by the fact that if $(\phi,\psi, \theta)$ is a superhedge then $a(Z_1) I_{ \{\tau = 1 \} } + b(Z_2)I_{ \{\tau > 1 \} } \leq  \phi(Z_1) + \psi(Z_2) + \theta(Z_1) (Z_2-Z_1)$ holds almost surely, where
$\theta = I_{ \{\tau = 1 \} } \theta_1 + I_{ \{\tau = 2 \} } \theta_2$.
We write $\sS = \sS(a,b)$ for the set of superhedges, in the sense of Definition \ref{def:superhedge}, for the Bermudan option with payoff $(a,b)$.

\begin{defn}[Hedging cost]
Suppose $(\phi,\psi, \theta=\{ \theta_i \}_{i=1,2} )$ is a superhedge for the Bermudan option. The hedging cost (HC) associated to $(\phi,\psi, \theta=\{ \theta_i \}_{i=1,2} )$ is defined as
$HC(\phi,\psi,\theta) = \int \phi(x) \mu(dx) + \int \psi(y) \nu(dy)$.
\end{defn}
{Note that the component $\theta(Z_1)(Z_2-Z_1)$ does not contribute to the cost of the hedge. This corresponds to the notion that 
$Z$ is the discounted price of a traded asset, and is related to the notion that $Z$ is a martingale under any consistent martingale model.}

\begin{remark}
At this stage we do not assume that the integrals in the definition of the hedging cost are finite. However, we use the convention that $(-\infty) + (+\infty) = +\infty$. Thus, if $\int \phi(x) I_{  \{ \phi(x) < 0 \} } \mu(dx) = - \infty$ and $\int \phi(x) I_{  \{ \phi(x) > 0 \} } \mu(dx) = \infty$ then we define $\int \phi(x) \mu(dx) =  \infty$ (similarly for integrals of $\psi$ against $\nu$) and if
either $\int \phi(x) \mu(dx)=\infty$ or $\int \psi(x) \nu(dx)=\infty$ then we define $\int \phi(x) \mu(dx) + \int \psi(x) \nu(dx)=\infty$.
\end{remark}

Since $H(\phi,\psi,\theta)$ does not depend on $\theta$ we write $HC(\phi,\psi)$ instead of $HC(\phi,\psi,\theta)$.

The problem of finding the cheapest superhedging strategy is the dual problem:
\begin{prob}[Dual (superhedging) problem]
\label{prob:dual}
Find
\[ \sD = {\sD}(\mu, \nu;a,b) = \inf_{(\phi,\psi, \theta) \in \sS(a,b)} \left\{ \int \phi(x) \mu(dx) + \int \psi(y) \nu(dy) \right\} =\inf_{(\phi,\psi, \theta) \in \sS(a,b)} \sH(\phi,\psi). \]
\end{prob}

It follows from Hobson and Norgilas~\cite{HobsonNorgilas:19} that any function $\psi \geq b$, with $\psi$ convex, can be used to generate a superhedge:

\begin{lem}[\protect{Hobson and Norgilas~\cite[Lemma 2]{HobsonNorgilas:19}}]
\label{lem:HN}
Suppose $\psi \geq b$ with $\psi$ convex. Define $\phi = (a-\psi)^+$ and set $\theta_2=0$ and $\theta_1= - \psi'$. Then
$(\phi,\psi, \{\theta_i \}_{i = 1,2})$ is a superhedge.
\label{lem:phifrompsi}
\end{lem}
\begin{proof}
We have, for all $x,y\in\R$,
\[ b(y) \leq \psi(y) \leq \phi(x)+\psi(y) = \phi(x) +\psi(y) + \theta_2(x)(y-x) \]
and \eqref{eq:bdual} follows. Also, by the convexity of $\psi$,
\[ \psi(x) \leq \psi(y) - \psi'(x)(y-x)=\psi(y) + \theta_1(x)(y-x) \]
and we have
\[ a(x) \leq (a(x) -\psi(x))^+ + \psi(x) \leq \phi(x) + \psi(y) + \theta_1(x)(y-x) \]
and \eqref{eq:adual} follows.
\end{proof}

\begin{defn}[Superhedge generated by $\psi$]
If $\psi \geq b$ with $\psi$ convex, we say $((a-\psi)^+,\psi, - \psi',0)$ is the superhedge generated by $\psi$.
\end{defn}

Let $ 
\tilde{\sS}(b) = \{ \mbox{$\psi \geq b$ with $\psi$ convex.} \}$. 
Let $(\phi,\psi,\theta)$ be given by $(\phi,\psi,(\theta_{i})_{i=1,2}) = ((a-\psi)^+,\psi, - \psi',0)$.
If follows from Lemma~\ref{lem:phifrompsi} that each $\psi \in \tilde{\sS}(b)$ generates an element $(\phi,\psi,\theta) \in \sS(a,b)$.
Then, for $\psi \in {\tilde{\sS}}(b)$ we can define $\widetilde{HC}(\psi) = HC((a-\psi)^+,\psi)$, which is the hedging cost associated with the superhedge $(\phi,\psi, \{\theta_i \}_{i = 1,2}) = (\phi,\psi, -\psi',0)$.

\begin{prob}[Restricted Dual (superhedging) problem]\label{prob:restrictedDual}
Find
\[ \tilde{\sD} = \tilde{\sD}(\mu, \nu; a,b)  = \inf_{ \psi \in \tilde{S}(b) } \left\{ \int (a(x)-\psi(x))^+ \mu(dx) + \int \psi(y) \nu(dy) \right\} =\inf_{ \psi \in \tilde{S}(b) }\widetilde{HC}(\psi). \]
\end{prob}

Clearly $\sD \leq \tilde{\sD}$.
Now suppose $b$ is convex. The main result of this section is that the cheapest superreplicating strategy is of a form generated by $\psi \in \tilde{\sS}(b)$.
\begin{theorem}
\label{thm:simpledual}
Suppose $b$ is convex. Then $\sD = \tilde{\sD}$
\end{theorem}

The idea behind the proof is to take a general superhedging strategy $(\phi,\psi,\theta) \in \sS(a,b)$ and to show that it can be modified to give another superhedging strategy which is generated by an element of $\hat{\psi} \in \tilde{\sS}(b)$, and which has a lower hedging cost. We do this in three stages. 
First we show that given $(\phi,\psi,\theta) \in \sS(a,b)$ we can replace $\psi$ with $\psi^c$, so that $(\phi,\psi^c,\theta) \in \sS(a,b)$ is still a superreplicating strategy. Clearly, this can only reduce the hedging cost. Hence, without loss of generality, we may restrict attention to superhedges for which $\psi$ is convex.
Second, we show that, given $(\phi,\psi,\theta) \in \sS(a,b)$ with $\psi$ convex, we can take a particular choice of $\phi$ (namely $\phi=\max \{ (a-\psi), (-(\psi-b)^c) \}$) and we still have a superhedge. Again, we will show that this can only lower the hedging cost. Hence we may restrict attention to superhedges for which $\psi$ is convex and $\phi$ takes this particular form.
Finally, we show that, given $\psi$ convex and $\phi$ of the particular form, we can introduce $\hat{\psi}$ with $\hat{\psi} = \psi - (\psi-b)^c \geq b$, and such that the hedging cost associated with the superhedge generated by $\hat{\psi}$ is no larger than the hedging cost associated with the superhedge $(\phi,\psi,\theta)$.

We begin with some preliminaries from Beiglb\"{o}ck et al~\cite{BeiglbockHobsonNorgilas:22}.

\begin{lem}[\protect{\cite[Lemma 2.3]{BeiglbockHobsonNorgilas:22}}]
\label{lem:BHN2}
Suppose $f$ and $g$ are convex. Set $G = g - (g-f)^c$. Then $G$ is convex.
\end{lem}

\begin{lem}[\protect{\cite[Lemma 2.4]{BeiglbockHobsonNorgilas:22}}]
\label{lem:BHN1}
Suppose $g$ is convex and $h$ is measurable. Then $(h-g)^c = (h^c-g)^c$.
\end{lem}

Also, we have the following `obvious' result.
\begin{lem}
\label{lem:addline}
Suppose $L$ is a straight line. Then $(g + L)^c = g^c + L$.
\end{lem}

\begin{prop}
Suppose $b$ is convex and $(\phi,\psi, \{ \theta_i \}_{i=1,2} )$ is a superhedge. Then so is $(\phi,\psi^c, \{ \theta_i \}_{i=1,2} )$. Moreover, $HC(\phi,\psi^c) \leq HC(\phi,\psi)$.
\label{prop:convexhull}
\end{prop}

\begin{proof}
The inequality $HC(\phi,\psi^c) \leq HC(\phi,\psi)$ is trivial and thus we focus on showing that $(\phi,\psi^c, \{ \theta_i \}_{i=1,2} )\in\sS(a,b)$. From \eqref{eq:adual} we have $a(x) \leq  \phi(x) + \psi(y) + \theta_1(x)(y-x)$ for all $x,y\in\R$. Take $x$ as fixed and consider taking the convex hull on both sides with respect to $y$. Then, using Lemma~\ref{lem:addline}, we have that $a(x) \leq \phi(x) + \psi^c(y) + \theta_1(x)(y-x)$.

Similarly, from \eqref{eq:bdual} we have $b(y)  \leq  \phi(x) + \psi(y) + \theta_2(x)(y-x)$. Fix $x$ and let $L(y) = \phi(x) + \theta_2(x)(y-x)$.
Then $b \leq \psi+L$. Taking the convex hull on both sides (with respect to $y$) and using Lemma~\ref{lem:addline} together with the convexity of $b$, we have that
$b = b^c \leq (\psi + L)^c = \psi^c + L$, i.e., for all $x,y\in\R$,
$$
b(y)\leq  \phi(x) + \psi^c(y) + \theta_2(x)(y-x).
$$
\end{proof}

From now on we may and do assume that $\psi$ is convex.

\begin{prop}
Suppose $b$ is convex and $(\phi,\psi, \{ \theta_i \}_{i=1,2} )$ is a superhedge with $\psi$ convex. Then there exists $\{ \tilde{\theta}_i \}_{i=1,2}$ such that $((a-\psi) \vee (-(\psi - b)^c),\psi, \{ \tilde{\theta}_i \}_{i=1,2} )$ is a superhedge.  Moreover, $HC((a-\psi) \vee (-(\psi - b)^c),\psi) \leq HC(\phi,\psi)$.
\label{prop:betterphi}
\end{prop}

\begin{proof}
If $(\phi,\psi, \{ \theta_i \}_{i=1,2} )$ is a superhedge, then taking $y=x$ in \eqref{eq:adual} gives
\[ a(x) \leq \phi(x) + \psi(x) \]
so that $\phi \geq (a - \psi)$.

Also, from \eqref{eq:bdual} we have that
\[ 0 \leq \phi(x) + \psi(y) - b(y) + \theta_2(x)(y-x), \]
and thus, fixing $x$ and with $L(y) = \phi(x) + \theta_2(x)(y-x)$, we have $0 \leq \psi - b + L$. Using Lemma~\ref{lem:addline},
\[ 0 \leq (\psi - b + L)^c = (\psi-b)^c + L. \]
In particular, $0 \leq (\psi - b)^c(y) + \phi(x) +  \theta_2(x)(y-x)$, and at $y=x$, $0 \leq \phi(x)+ (\psi - b)^c(x)$ so that $\phi \geq (-(\psi - b)^c)$.

We find that necessarily $\phi \geq (a - \psi) \vee (-(\psi - b)^c)$ so that $H(\phi,\psi) \geq HC((a-\psi) \vee (-(\psi - b)^c),\psi)$, provided that $((a-\psi) \vee (-(\psi - b)^c),\psi)$ generates a superhedge. Hence, it remains to show that we can find $\{ \tilde{\theta}_i \}_{i=1,2}$ such that $((a-\psi) \vee (-(\psi - b)^c),\psi, \{ \tilde{\theta}_i \}_{i=1,2} )\in\sS(a,b)$. 

Set $\tilde{\phi}=(a-\psi) \vee (-(\psi - b)^c)$ and $h = (\psi-b)^c$. Let $\tilde{\theta}_1 = - \psi'$ and $\tilde{\theta}_2 = - h'$.

By the convexity of $\psi$ we have $\psi(y) \geq \psi(x) + \psi'_+(x)(y-x)$ so that $\psi(x) \leq \psi(y) + \tilde{\theta}_1(x)(y-x)$. Then
\begin{equation}
\label{eq:worksfora}
a(x) = (a(x) - \psi(x)) + \psi(x) \leq \tilde{\phi}(x) + \psi(x) \leq \tilde{\phi}(x) +  \psi(y) + \tilde{\theta}_1(x)(y-x).
\end{equation}
Also, $\tilde{\phi} \geq -h$ and by the convexity of $h$, $h(x) \leq h(y) - h'(x)(y-x) = h(y) + \tilde{\theta}_2(x)(y-x)$. Then
\begin{eqnarray}
b(y) & \leq & b(y)+ \tilde{\phi}(x) + h(x) \nonumber \\
& \leq & b(y) + \tilde{\phi}(x) + (\psi - b)^c(y) + \tilde{\theta}_2(x)(y-x) \nonumber \\
& \leq & b(y) + \tilde{\phi}(x) + (\psi - b)(y) + \tilde{\theta}_2(x)(y-x) \nonumber \\
& = & \tilde{\phi}(x) + \psi(y) + \tilde{\theta}_2(x)(y-x).\label{eq:worksforb}
\end{eqnarray}
\eqref{eq:worksfora} and \eqref{eq:worksforb} combine to show that $(\tilde{\phi},\psi, \{ \tilde{\theta}_i \}_{i=1,2} )$ is a superhedge.
\end{proof}

From now on we may assume that $\psi$ is convex and $\phi = (a - \psi) \vee (-(\psi-b)^c)$.

\begin{prop}
Suppose $(\phi = (a - \psi) \vee (-(\psi-b)^c),\psi,\{ {\theta}_i \}_{i=1,2} )$ is a superhedge.
Define $\hat{\psi} = \psi - (\psi-b)^c$. Then $\hat{\psi} \geq b$, $\hat{\psi}$ is convex, $(\hat{\psi}-b)^c\equiv0$ and $((a-\hat{\psi})^+, \hat{\psi}, \{ - \hat{\psi}', 0 \})$ is a superhedge. Moreover, $HC((a-\hat{\psi})^+, \hat{\psi}) \leq HC(\phi, \psi)$.
\label{prop:besthedge}
\end{prop}

\begin{proof}
Clearly, $\hat{\psi} - b = \psi - b - (\psi-b)^c \geq 0$. Moreover, since $\psi$ and $b$ are convex, $\hat{\psi}$ is convex by Lemma~\ref{lem:BHN2}.
Since $\hat{\psi}$ is convex and $\hat{\psi} \geq b$, by Lemma~\ref{lem:HN} and Proposition \ref{prop:betterphi} we have that it generates a superhedge with hedging cost $HC((a-\hat{\psi})^+, \hat{\psi})=\int (a - \hat{\psi})^+ d \mu + \int \hat{\psi} d\nu$.

Taking $h=(\psi-b)$ and $g=(\psi-b)^c$, Lemma \ref{lem:BHN1} implies that $(\hat{\psi}-b)^c = (h-g)^c=(h^c-g)\equiv0$.

It only remains to check that $HC((a-\hat{\psi})^+, \hat{\psi}) \leq HC(\phi, \psi)$. But, with $g=(\psi-b)^c$
\begin{eqnarray}
HC(\phi,\psi) &=& \int \{ (a - {\psi}) \vee (-g) \} d \mu + \int \psi d\nu  \nonumber\\
& = & \int \{ (a - \hat{\psi} - g) \vee (-g) \} d \mu + \int \{ \hat{\psi} + g \} d\nu \nonumber \\
& = & \int \{ (a - \hat{\psi})^+ - g \} d \mu  + \int \{ \hat{\psi} + g \} d\nu \nonumber \\ 
& = & \int  (a - \hat{\psi})^+  d \mu  + \int \hat{\psi}  d\nu + \int g d\nu - \int g d\mu \label{eq:4integrals} \\
& \geq & \int  (a - \hat{\psi})^+  d \mu  + \int \hat{\psi}  d\nu \nonumber
\end{eqnarray}
with the last inequality following since $g$ is convex and $\mu \leq_{cx} \nu$.
\end{proof}

\begin{remark}
    Note that at no stage did we assume that the hedging cost is finite. The comparisons in Propositions~\ref{prop:convexhull} and \ref{prop:betterphi} rely on the monotonicity of integration and do not need finiteness.



    In Proposition~\ref{prop:besthedge}, if $\int \psi d \nu = \infty$ then $HC(\phi,\psi)=\infty$ and there is nothing to prove. So suppose $\int \psi d \nu<\infty$.
   {Note that if $\eta \in L^1$ and $f$ is convex then necessarily $\int f(y) I_{ \{ f(y)<0 \} } \eta(dy) > - \infty$. Then,
    since $\int \psi d \nu = \int \{ \hat{\psi} + g \} d \nu \geq \int b d\nu + \int g d \nu$, we conclude that $\int g d \nu < \infty$. Then also $\int |g| d \nu < \infty$ (and because of the convex order $\int |g| d \mu < \infty$). It follows that all the integrals in \eqref{eq:4integrals}  are well defined (the first two in $[0,\infty]$ and the last two in $(-\infty,\infty)$).}

    Putting this all together, we do not claim that $HC((a-\hat{\psi})^+, \hat{\psi}) < \infty$ in Proposition~\ref{prop:besthedge}, but nonetheless we always have $HC((a-\hat{\psi})^+, \hat{\psi}) \leq HC(\phi, \psi)$.
\end{remark} 

\begin{rem}\label{rem:D0}
Actually we have shown that $\sD = \tilde{\sD} = \tilde{\sD^0}$ where
\[ \tilde{\sD}^0 =  \tilde{\sD}^0(\mu, \nu;a,b) = \inf_{ \psi \in \tilde{\sS}^0(b) } \left\{ \int (a(x)-\psi(x))^+ \mu(dx) + \int \psi(y) \nu(dy) \right\}  \]
where $\tilde{\sS}^0(b) = \{ \psi: \mbox{$\psi$ convex,  $\psi \geq b$, $(\psi-b)^c\equiv0$} \} = \{ \psi \in \tilde{\sS}(b): (\psi-b)^c\equiv0 \}$.
\end{rem}

{The main result of this section is Theorem~\ref{thm:simpledual}, which says that when looking for an optimiser in the dual problem it is sufficient to search over convex functions which dominate $b$, and then to use these convex functions to generate hedging strategies. This greatly simplifies the analysis of the dual problem. In the next few sections we exploit this simplification to find the solution of the primal problem and to show that there is no duality gap. We make a related point in Remark~\ref{rem:hedgingStrategies}.}

\section{Standing assumptions}\label{sec:primal}
In the second part of the paper we aim to find explicit solutions to the pricing and hedging problems under some simplifying structural assumptions on $\mu$ and $\nu$. The goal is to find the model $\sM^*$ and associated stopping time $\tau^*$ such that the highest model-based price is attained, along with the cheapest superhedege $(\phi^*,\psi^*,\theta^*)$. We find candidates for each and proceed to show that $\sP^* := \E^{\sM^*}[c(Z_{\tau^*},\tau^*)] = \int \phi^* d \mu + \int \psi^* d \nu =: \sD^*$ where $\E^{\sM}$ denotes expectations in the model $\sM=(\Omega, \sF, \bF, \PP)$. Then $\sP^* \leq \sP \leq \sD \leq \sD^*$ (the two outer inequalities are by definition, and the middle one follows by weak duality). It follows that we must have equality throughout and that we have found a model under which the Bermudan option has the highest price and a superhedge with the lowest hedging cost. Moreover, there is no duality gap.
We begin by stating standing assumptions on the laws $\mu$ and $\nu$ and the payoff functions.

\begin{sass}[Standing assumptions on the distributions]\label{sass:densities}
$\mu \leq_{cx} \nu$ with $\mu \neq \nu$; $\mu$ and $\nu$ are absolutely continuous (with respect to Lebesgue measure) with densities $\rho$ and $\eta$ respectively; the smallest open interval $(\alpha,\beta)\subseteq \R$ such that $\mu(\alpha,\beta)=1$ is the interval $(\alpha_\mu, \beta_\mu)$ and the smallest open interval $(\alpha,\beta)\subseteq \R$ such that $\nu(\alpha,\beta)=1$ is the interval $(\alpha_\nu, \beta_\nu)$;
then necessarily $-\infty \leq \alpha_\nu \leq \alpha_\mu < \beta_\mu \leq \beta_\nu \leq \infty$.
\end{sass}


Consider a Bermudan option with convex payoff functions $a,b:\R\to[0,\infty]$ which satisfy $a\geq b$ on $\R$ and $a(x_0) <\infty$ for some $x_0\in\R$. We begin by ruling out a couple of cases where the result is trivial.

First, if $a=b$ on $(\alpha_\mu,\beta_\mu)$, then by (the conditional) Jensen's inequality and the martingale property of $(Z_0,Z_1,Z_2)$ (under any model $\sM\in M(\mu,\nu)$), we have that $(C_1=a(Z_1),C_2=b(Z_2))$ forms a $\sM$-submartingale, and thus it is optimal to stop at time-2. It follows that every model $\sM$ is optimal and $\sP=\int_\R b(y)\nu(dy)$. 

Second, dropping the assumption that $a=b$, due to the convexity (of $a$ and $b$) we have that there exist intervals $I_a,I_b\subseteq \R$ such that $a<\infty$ on $I_a$ (resp., $b<\infty$ on $I_b$) and $a=\infty$ on $I_a^\infty:=\R\setminus I_a$ (resp., $b=\infty$ on $I_b^\infty:=\R\setminus I_b$). If $\mu(I^\infty_a)=\mu((\alpha_\mu,\beta_\mu)\cap I^\infty_a)>0$ (resp.,   $\nu(I^\infty_b)=\nu((\alpha_\nu,\beta_\nu)\cap I^\infty_b)>0$), then for any model $\sM$, by taking $\tau^*=1$ (resp., $\tau^*=2$) we have that $\E^{\sM}[a(Z_1)I_{ \{ \tau^*=1 \} } + b(Z_2)I_{ \{ \tau^*=2 \} } ]=\E^\sM[a(Z_1)]=\int_\R a(x)\mu(dx)=\infty$ (resp., $\E^{\sM}[a(Z_1)I_{ \{ \tau^*=1 \} } + b(Z_2)I_{ \{ \tau^*=2 \} } ]=\E^\sM[b(Z_2)]=\int_\R b(y)\nu(dy)=\infty$). Again, it follows that all models $\sM$ are optimal and $\sP=\infty$. 

These two observations motivate the following standing assumption (the second part is a mild simplifying assumption on the payoff functions which helps in the proof of Proposition~\ref{prop:Lambday}). 
\begin{sass}[Standing assumptions on the payoff functions]
\label{sass:payoffs}
The payoff functions $a:(\alpha_\mu,\beta_\mu)\to\R_+$, $b:(\alpha_\nu,\beta_\nu) \to \R_+$  are such that $a\neq b$ and $a \geq b$ on $(\alpha_\mu,\beta_\mu)$, and that both $a$ and $b$ are convex. 
 
Moreover, there are no intervals $I \subseteq (\alpha_\mu,\beta_\mu)$ on which $a=b$ and $a$ is linear.

\end{sass}

We extend the definitions of $a$ and $b$ to $\R$ by setting $a = \infty$ on $(-\infty,\alpha_\mu) \cup (\beta_\mu,\infty)$, $b = \infty$ on $(-\infty,\alpha_\nu) \cup (\beta_\nu,\infty)$, $a(\alpha_\mu)=\lim_{x\downarrow\alpha_\mu}a(x)$, $a(\beta_\mu)=\lim_{x\uparrow\beta_\mu}a(x)$, $b(\alpha_\nu)=\lim_{x\downarrow\alpha_\nu}b(x)$ and $b(\beta_\nu)=\lim_{x\uparrow\beta_\nu}b(x)$. Note that the choice to extend $a$ and $b$ in this way does not affect the expected payoff of the Bermudan option since $\mu((\alpha_\mu,\beta_\mu)) = 1 = \nu((\alpha_\nu,\beta_\nu))$.


Let $L^b$ be the smallest straight line such that $L^b \geq b$ on $(\alpha_\nu,\beta_\nu)$. Note that if both $\alpha_\nu$ and $\beta_\nu$ are finite and $b(\alpha_\nu)\vee b(\beta_\nu)<\infty$ then $L^b=T_{b,b}^{\alpha_\nu,\beta_\nu}$ is the straight line passing through $(\alpha_\nu,b(\alpha_\nu))$ and $(\beta_\nu,b(\beta_\nu))$; if $b(\alpha_\nu)\vee b(\beta_\nu)=\infty$ then there does not exist a finite straight line with $L^b \geq b$ in which case we set $L^b \equiv \infty$. If $\alpha_\nu$ and $\beta_\nu$ are both infinite then either $b$ is constant or again there is no finite  straight line with $L^b \geq b$ and then we set $L^b \equiv \infty$.
If exactly one of the endpoints $\alpha_\nu,\beta_\nu$ is finite then 
there may or may not be a finite straight line with $L^b \geq b$, if not then $L^b \equiv \infty$. 

The case considered in the following lemma only occurs when 
$L^b$ is finite.

\begin{lem}
    {Suppose that 
    $a \geq L^b$.} 

    Let $\tau^*=1$, and let $\psi^{*}=L^b$.
    Then, for every model $\sM$ we have 
    \[ \E^{\sM}[a(Z_1)] = \E^{\sM}[a(Z_1)I_{ \{ \tau^*=1 \} } + b(Z_2)I_{ \{ \tau^*=2 \} } ] =\int a d \mu = \int (a-L^b)^+ d \mu + \int L^b d \nu. \] 
    It follows that for every model $\sM$, $(\sM,\tau^*=1)$ gives the highest model-based price for the Bermudan option, and $L^b$ generates the cheapest superhedge. There is no duality gap.
    \label{lem:tau=1}
\end{lem}

\begin{proof}
   Only the last equality is not immediate, but this follows since $a \geq L^b$ and $\int L^b d \mu = \int L^b d \nu$, the latter equality following from the fact that $\mu$ and $\nu$ are in convex order.    
\end{proof}

{From now on we will exclude this case. Note that if $b$ is constant and $a \geq b$ then $a \geq L^b =b$, so that in particular we are excluding the case where $b$ is constant. This case is covered in Lemma~\ref{lem:tau=1}.}

\begin{sass}
\label{sass:ageqLb}
It is not the case that $a \geq L^b$. 
\end{sass}


\begin{lem}
\label{lem:SA3}
Under Standing Assumption~\ref{sass:ageqLb} there exists $\hat{w},\hat{u},\hat{z}$ with $\alpha_\nu < \hat{w} < \hat{u} < \hat{z}< \hat{\beta}_\nu$ such that
$T_{b,b}^{\hat{w},\hat{z}} (\hat{u}) > a(\hat{u})$.
\end{lem}

\begin{proof}
    This follows from the properties of convex functions.
\end{proof}

\section{Special subsets of ${\mathbb R}^n$, special convex functions and the main result}
In this section we introduce three classes of models defined via four, five or six points and three classes of hedges, again defined by four, five or six points. The plan is to show that if we can find a model and a hedge defined by the same $n$ points (where $n \in \{4,5,6\}$) then the model and hedge are both optimal and there is no duality gap.
Later we show that under some reasonable assumptions it is always possible to find such a set of points.

For $n=4,5,6$ and $w^{(n)} = (w_1, \ldots, w_n) \in \R^n$ define { $\Sigma_n= \Sigma_n^{\mu,\nu}$ by
\begin{eqnarray*}
\Sigma^{\mu,\nu}_4 & = & \{ w^{(4)} \in \R^4 : \alpha_\nu < w_1 \leq w_2 < w_3 \leq w_4 <  \beta_\nu, \alpha_\mu \leq w_2 < w_3  \leq \beta_\mu, \mbox{and} \\ 
&& \hspace{80mm} (\beta_\mu - w_3)+(w_2 - \alpha_\mu)>0  \},    \\
\Sigma^{\mu,\nu}_5 & = & {\{ w^{(5)} \in \R^5 : \alpha_\nu < w_1 \leq  w_2 \leq w_3 \leq w_4  \leq w_5 < \beta_\nu, \alpha_\mu < w_2 < w_4 <\beta_\mu, \mbox{and} } \\
&& \hspace{10mm} {\mbox{(either $w_2>w_1$ or $w_1=w_2=w_3$) and (either $w_4<w_5$ or $w_3=w_4=w_5$)}  \},  }  \\
\Sigma^{\mu,\nu}_6 & = & \{ w^{(6)} \in \R^6 : \alpha_\nu < w_1 \leq w_2 \leq w_3 < w_4 \leq w_5 \leq w_6 < \beta_\nu, \alpha_\mu < w_2 < w_5 <\beta_\mu, \mbox{ and}   \\
&& \hspace{10mm} { \mbox{(either $w_2>w_1$ or $w_1=w_2=w_3$) and (either $w_5<w_6$ or $w_4=w_5=w_6$)}  \}.  }  
\end{eqnarray*}

Note the different conventions over weak and strict inequalities in each case. Essentially, for $n=4$ we need to allow for one of the points $w_2,w_3$ to lie at the ends of the support of $\mu$, but this case does not occur when $n=5$ or $n=6$.
}


For $n=4$ define $\Gamma^{\mu,\nu}_4$ by
\begin{eqnarray}
\label{eq:Gamma4def}
\Gamma^{\mu,\nu}_4 & = & 
\{ (y_-,x_-,x_+,y_+) \in \Sigma_4 :  \mu|_{(\alpha_\mu, y_-) \cup (y_+,\beta_\mu)} \leq \nu|_{(\alpha_\mu, y_-) \cup (y_+,\beta_\mu)}, \\
&& \nonumber \hspace{55mm} \mu|_{(y_-,x_-) \cup (x_+,y_+)} \leq_{cx} \nu|_{(y_-,y_+)} \} .
\end{eqnarray}

\begin{lem}
\label{lem:musthave4}
If $(y_-,x_-,x_+,y_+) \in \Gamma^{\mu,\nu}_4$ then
\begin{equation}
    \label{eq:musthave4}
\mu|_{(x_-,x_+)} \leq_{cx} (\nu - \mu)_{(\alpha_\nu,y_-) \cup (y_+,\beta_\nu)} . 
\end{equation}
\end{lem}
\begin{proof}
Since $\mu|_{(y_-,x_-) \cup (x_+,y_+)} \leq_{cx} \nu|_{(y_-,y_+)}$ we have that $\mu|_{(x_-,x_+)} = \mu - \mu|_{(\alpha_\nu,y_-) \cup (y_+,\beta_\nu)} -
\mu|_{(y_-,x_-) \cup (x_+,y_+)}$ and $(\nu - \mu)|_{(\alpha_\nu,y_-) \cup (y_+,\beta_\nu)} = \nu  - \mu|_{(\alpha_\nu,y_-) \cup (y_+,\beta_\nu)} - \nu|_{(y_-,y_+)}$ have the same mean and mass.

Then the fact that $\mu|_{(x_-,x_+)} \leq_{cx} (\nu - \mu)_{(\alpha_\nu,y_-) \cup (y_+,\beta_\nu)}$ follows from the fact that 
$\mu|_{(x_-,x_+)}$ has support in $(x_-,x_+)$
and $(\nu - \mu)_{(\alpha_\nu,y_-) \cup (y_+,\beta_\nu)}$ has support outside this set. Alternatively, this result follows from Proposition~\ref{prop:shadow_assoc} below. 
\end{proof}

For $n=5$ define $\Gamma^{\mu,\nu}_5$ by
\begin{eqnarray}
\label{eq:Gamma5def}
\Gamma^{\mu,\nu}_5 & = & 
\{ (y_-,x_-,z,x_+,y_+) \in \Sigma_5 :  \mu|_{(\alpha_\mu, y_-) \cup (y_+,\beta_\mu)} \leq \nu|_{(\alpha_\mu, y_-) \cup (y_+,\beta_\mu)}, \\
&& \nonumber \hspace{55mm} \mu|_{(y_-,x_-)} \leq_{cx} \nu|_{(y_-,z)}, \mu|_{(x_+,y_+)}  \leq_{cx} \nu|_{(z,y_+)} \} . 
\end{eqnarray}

As in the case of $n=4$ we have an identical result with a similar proof.

\begin{lem}
\label{lem:musthave5}
If $(y_-,x_-,z,x_+,y_+) \in \Gamma^{\mu,\nu}_5$ then
$ 
\mu|_{(x_-,x_+)} \leq_{cx} (\nu - \mu)_{(\alpha_\nu,y_-) \cup (y_+,\beta_\nu)} . 
$ 
\end{lem}

For $n=6$ define $\Gamma^{\mu,\nu}_6$ by
\begin{eqnarray}
\label{eq:Gamma6def}
\Gamma^{\mu,\nu}_6 & = & 
\{ (y_-,x_-,z_-,z_+,x_+,y_+) \in \Sigma_6 :  \mu|_{(\alpha_\mu, y_-) \cup (y_+,\beta_\mu)} \leq \nu|_{(\alpha_\mu, y_-) \cup (y_+,\beta_\mu)}, \\
&& \nonumber \hspace{15mm} \mu|_{(y_-,x_-)} \leq_{cx} \nu|_{(y_-,z_-)}, \mu|_{(x_+,y_+)}  \leq_{cx} \nu|_{(z_+,y_+)},
\nu|_{(z_-,z_+)} \leq \mu|_{(z_-,z_+)} 
\} . 
\end{eqnarray}
This time we find that
\begin{lem}
\label{lem:musthave6} 
Suppose $(y_-,x_-,z_-,z_+,x_+,y_+) \in \Gamma^{\mu,\nu}_6$. Then
$\mu\lvert_{(x_-,z_-)\cup(z_+,x_+)} +(\mu-\nu)|_{(z_-,z_+)} \leq_{cx} (\nu - \mu)_{(\alpha_\nu,y_-) \cup (y_+,\beta_\nu)}$. 
\end{lem}


We let $C(\alpha_\nu,\beta_\nu)$ be the set of continuous functions on $(\alpha_\nu,\beta_\nu)$ and we let $C_{cx}(\alpha_\nu,\beta_\nu) \subset C(\alpha_\nu,\beta_\nu)$ be the set of convex functions on $(\alpha_\nu,\beta_\nu)$.

{For $(y_-,x_-,x_+,y_+) \in \Sigma_4$ define $\Psi_4 = \Psi^{y_-,x_-,x_+,y_+}_{4,a,b}:(\alpha_\nu,\beta_\nu): \to \R$ by
\begin{equation}
    \label{eq:psi4points}
 \Psi_4 
 (u)  =  
\begin{cases} 
     b(u), & u < y_-; \\
    T^{y_-,y_+}_{b,b}(u), \quad \quad
         &  y_- \leq u \leq y_+; \\
     b(u), & y_+ <  u.    
           \end{cases} 
\end{equation}
Our convention that $\alpha_\nu < y_-<y_+ < \beta_\nu$ for $(y_-,x_-,x_+,y_+) \in \Sigma_4$ ensures that $\Psi_4$ is well defined.
By construction, $\Psi_4$ is convex. 
Now we define a subset $\sL^{a,b}_4$ of $\Sigma_4$ with the property that $T_{b,b}^{y_-,y_+}$ takes the same value as $a$ at $x_{\pm}$:
\begin{equation}
\label{eq:L4def}
\sL^{a,b}_4 = \{ (y_-,x_-,x_+,y_+) \in \Sigma_4 : a(x_--) \geq T_{b,b}^{y_-,y_+}(x_-) \geq a(x_-+), a(x_+-) \leq T_{b,b}^{y_-,y_+}(x_+) \leq a(x_++)  \} .
\end{equation}
Observe that if $x_- > \alpha_\mu$ then $a(x_--)= a(x_-+)$, and similarly if $x_+<\beta_\mu$ then $a(x_+ -)= a(x_+ +)$, so that in the case when $\alpha_\mu<x_-<x_+<\beta_\mu$ we can more simply write the condition $(y_-,x_-,x_+,y_+) \in \sL^{a,b}_4$ as the condition that the points
$(y_-,b(y_-))$, $(x_-,a(x_-))$, $(x_+,a(x_+))$ and $(y_+,b(y_+))$ all lie on the same straight line.}

For $(y_-,x_-,z,x_+,y_+) \in \Sigma_5$ define $\Psi_5=\Psi^{y_-,x_-,z,x_+,y_+}_{5,a,b}:(\alpha_\nu,\beta_\nu): \to \R$ by
\begin{equation}
    \label{eq:psi5points}
 \Psi_5 
 (w)  =  
\begin{cases} 
     b(w), & w \leq y_-; \\
    T^{y_-,x_-}_{b,a}(w), \quad \quad
         &  y_- < w \leq z; \\
     T^{x_+,y_+}_{a,b}(w), 
        &  z < w \leq y_+;  \\
     b(w), & y_+<w.    
           \end{cases} 
\end{equation}
Note that $\Psi_5 
$ need not be continuous at $z$. Furthermore, if $y_-=z$ (resp., $z=y_+$) then the second (resp., third) line in the definition of $\Psi_5$ is redundant; {otherwise, by the definition of $\Sigma_5$ we must have that $y_-<x_-$ (resp., $x_+=y_+$) so that $T^{y_-,x_-}_{b,a}$ (resp.,  $T^{x_+,y_+}_{a,b}$) is well defined.} Let $\sL^{a,b}_5$ be given by
\[ \sL^{a,b}_5 = \{ (y_-,x_-,z,x_+,y_+) \in \Sigma_5 : \Psi_{5,a,b}^{y_-,x_-,z,x_+,y_+} \in C_{cx}(\alpha_\nu,\beta_\nu), \Psi_{5,a,b}^{y_-,x_-,z,x_+,y_+}(z) \geq a(z) \} .\]
In particular, $\Psi_{5} \in C_{cx}(\alpha_\nu,\beta_\nu)$ guarantees that $\Psi_{5}$ is continuous. 


For $(y_-,x_-,z_-,z_+,x_+,y_+) \in {\Sigma}_6$ define $\Psi_6 = \Psi^{y_-,x_-,z_-,z_+,x_+,y_+}_{6,a,b}:(\alpha_\nu,\beta_\nu): \to \R$ by
\begin{equation}
    \label{eq:psi6points}
\Psi_6 
(w)  =  
\begin{cases} 
     b(w), & w \leq y_-; \\
    T^{y_-,x_-}_{b,a}(w), \quad \quad
         &  y_- < w \leq z_-; \\
         a(w), & z_- < w\leq z_+;\\
     T^{x_+,y_+}_{a,b}(w), 
        &  z_+ < w \leq y_+;  \\
     b(w), &  y_+<w.    
           \end{cases} 
\end{equation}
Note that $\Psi^{y_-,x_-,z_-,z_+,x_+,y_+}_{6,a,b}$ need not be continuous at $w$ for $w \in \{ z_\pm \}$.  Furthermore, if $y_-=z_-$ (resp.,  $z_+=y_+$) then the second (resp., fourth) line in the definition of $\Psi_6$ is redundant; {otherwise the definition of $\Sigma_6$ ensures that $T^{y_-,x_-}_{b,a}$ (resp., $T^{x_+,y_+}_{a,b}$) is well defined.} Let $\sL^{a,b}_6$ be given by
\[ \sL^{a,b}_6 = \{ (y_-,x_-,z_-,z_+,x_+,y_+) \in \Sigma_6 : \Psi_{6,a,b}^{y_-,x_-,z_-,z_+,x_+,y_+} \in C_{cx}(\alpha_\nu,\beta_\nu) \} .\]

\begin{remark}\label{rem:hedgingStrategies} For $n\in\{4,5,6\}$, the functions $\Psi_n$ (recall \eqref{eq:psi4points}, \eqref{eq:psi5points} and \eqref{eq:psi6points}) can be used to construct candidate superhedging strategies. Here we are motivated by Theorem \ref{thm:simpledual} (see also Remark \ref{rem:D0}): when searching for the cheapest superhedging strategy, one can restrict the attention to strategies that are generated by convex $\psi$ with $\psi\geq b$ (i.e., $\psi\in\tilde{\mathcal{S}}(b)$). Note that, for $n=4$ and an arbitrary $w=(w_1,...,w_4)\in\R^4$, $\Psi_4^w\in \tilde{\mathcal{S}}(b)$ (see \eqref{eq:psi4points}). However, for $n\in\{5,6\}$ and an arbitrary $w=(w_1,...,w_n)\in\R^n$, we could have that $\Psi_n^w\notin \tilde{\mathcal{S}}(b)$ (in fact, $\Psi_n^w$ may not be even continuous). On the other hand, if $w\in\sL_n^{a,b}$, then  due to our definitions (see \eqref{eq:psi5points} and \eqref{eq:psi6points}), $\Psi_n^w$ is convex and $\Psi^w_n\geq b$, so that $\Psi_n^w\in \tilde{\mathcal{S}}(b)$ and then we can use it to form a hedge.
\end{remark}

The first main result says that if we can find $w = (w_1, \ldots w_n) \in \Gamma_n^{\mu,\nu} \cap \sL_{n}^{a,b}$ for some $n \in \{4,5,6\}$ then there is no duality gap and $w$ characterises both the optimal model and the optimal hedge.

\begin{theorem}
\label{thm:main1}
Suppose $\cup_{n=4,5,6} (\Gamma^{\mu,\nu}_n \cap \sL^{a,b}_n) \neq \emptyset$. 
Then $\sP=\sD$. In particular, if $w=(w_1, \ldots w_n) \in \Gamma^{\mu,\nu}_n \cap \sL^{a,b}_n$ then $w$ defines a model $\sM^w \in M(\mu,\nu)$, a stopping time $\tau^w$ and a hedge $\Psi^w \in \tilde\sS(b)$ such that
\[ \E^{\sM^w}[a(Z_1)I_{ \{ \tau^w=1 \}} + b(Z_2) I_{ \{ \tau^w=2 \} }] = \sP = \sD = \int (a - \Psi^w)^+ d \mu + \int \Psi^w d \nu \]
(i.e., $(M^w,\tau^w)$ and $\Psi^w$ are optimal).
\end{theorem}



The second main result says that under an assumption on the pair of laws $(\mu,\nu)$ which will be introduced below, but which is satisfied for a wide class of examples, the non-emptyness hypothesis of Theorem~\ref{thm:main1} is satisfied. 
It follows that under this assumption 
the problem is completely solved.
\begin{theorem}
\label{thm:main3}
Suppose Assumption~\ref{ass:simple} below holds. Then $\cup_{n=4,5,6} (\Gamma^{\mu,\nu}_n \cap \sL^{a,b}_n)$ is non-empty. Furthermore, $\sP=\sD$ and we can find an optimal model, stopping time and hedge.
\end{theorem}

\section{Martingale couplings and models}

\subsection{Canonical models and lifted models}

In order to explicitly solve the primal problem \eqref{eq:primal} (by finding an optimal model and an associated optimal stopping time), we introduce two useful (and natural) subsets of the set of $(\mu,\nu)$-consistent models $M(\mu,\nu)$. For $S = (\alpha,\beta) \subseteq \R$  let $\sB(S)$ denote the set of Borel subsets of $(\alpha,\beta)$ and more generally for $S^n=(\alpha_1,\beta_1) \times \ldots \times (\alpha_n,\beta_n) \subseteq \R^n$ let $\sB(S^n)$ denote the set of Borel subsets of $S^n$.

Let $\Pi_M(\mu,\nu)$ be the set of martingale couplings of $\mu\leq_{cx}\nu$. In particular, we write $\pi\in\Pi_M(\mu,\nu)$ if $\pi$ is a (Borel) probability measure on $(\alpha_\mu,\beta_\mu)\times(\alpha_\nu,\beta_\nu)$, if the first and second marginals of $\pi$ are $\mu$ and $\nu$, respectively (i.e., $\pi(A\times(\alpha_\nu,\beta_\nu))=\mu(A)$ and $\pi((\alpha_\mu,\beta_\mu)\times B)=\nu(B)$ for all $A\subseteq \sB((\alpha_\mu,\beta_\mu)),B\subseteq\sB((\alpha_\nu,\beta_\nu))$), and if $\pi(dx,dy)=\mu(dx)\pi_x(dy)$ is such that $\int |y|\pi_x(dy)<\infty$ and $\int y\pi_x(dy)=x$ for $\mu$-a.e. $x\in(\alpha_\mu,\beta_\mu)$. Here $(\pi_x)_{x \in (\alpha_\mu,\beta_\mu)}$ is a disintegration of $\pi$ with respect to $x$. The following theorem is due to Strassen~\cite{Strassen:65}

\begin{theorem}[Strassen's Theorem~\cite{Strassen:65}]
\label{thm:strassen}
$\Pi_M(\mu,\nu)\neq\emptyset$ if and only if $\mu\leq_{cx}\nu$. 
\end{theorem}

It follows that $\Pi_M(\mu,\nu)$ naturally induces the set of canonical models. In particular, let $\Omega^{can}= \{(x,y):\alpha_\mu < x < \beta_\mu,~ \alpha_\nu < y < \beta_\nu \}$, $\sF^{can} = \{ A \times B: A \in \sB((\alpha_\mu,\beta_\mu)),~ B \in \sB((\alpha_\nu,\beta_\nu))\}$, $\sF^{can}_0 = \{\Omega,\emptyset\}$, $\sF^{can}_1 = \{ A \times (\alpha_\nu,\beta_\nu) : A \in \sB((\alpha_\mu,\beta_\mu)) \}$,~ $\sF^{can}_2 = \sF^{can}$ and define
\begin{equation}\label{eq:canonicalModels}
M^{can}(\mu,\nu):=\{(\Omega^{can},\sF^{can},\mathbb F^{can}=(\sF^{can}_0,\sF^{can}_1,\sF^{can}_2),\PP=\pi):\pi\in\Pi_M(\mu,\nu)\}.
\end{equation}
Then $M^{can}(\mu,\nu)\subset M(\mu,\nu)$, and for each $\pi\in\Pi_M(\mu,\nu)$ there are random variables $(X,Y)=(Z_1,Z_2)$ such that $\PP(X \in dx, Y \in dy) = \pi(dx,dy)$. The corresponding model (associated to a coupling $\pi\in\Pi_M(\mu,\nu)$) is denoted by $\sM_\pi\in M^{can}(\mu,\nu)$. Denoting by $\sT^{can}_{1,2}$ the set of $\mathbb F^{can}$-stopping times (taking values in $\{1,2\}$) set
\begin{align}\label{eq:primalCanonical}
\sP^{can} = \sP^{can}(\mu,\nu;a,b) &= \sup_{\sM \in M^{can}(\mu,\nu)} \sup_{\tau \in \sT^{can}_{1,2}} \E^{\sM}[ c(Z_\tau,\tau)],
\end{align}
and note that
$$
\sP\geq\sP^{can}=\sup_{\pi\in\Pi_M(\mu,\nu)}\sup_{B\in \sB((\alpha_\mu,\beta_\mu))}\left\{\int_Ba(x)\mu(dx)+\int_{((\alpha_\mu,\beta_\mu)\setminus B)\times\R}b(y)\pi_x(dy)\mu(dx)\right\}.
$$

The second useful subset of $M(\mu,\nu)$ requires a uniform random variable on $[0,1]$ that is independent of $X\sim\mu$.

Let $\lambda=\lambda_{[0,1]}$ denote the Lebesgue measure on $[0,1]$. We write $\hat\pi\in\hat\Pi_M(\mu,\nu)$ if $\hat\pi$ is a (Borel) probability measure on $(\alpha_\mu,\beta_\mu)\times(0,1)\times(\alpha_\nu,\beta_\nu)$, if $\hat\pi(A\times V\times\R)=(\mu\otimes\lambda)(A\times V)$ for all  $A\in \sB((\alpha_\mu,\beta_\mu)), V \in \sB((0,1))$ (i.e., $\int_y\hat\pi(dx,du,dy)=(\mu\otimes\lambda)(dx,du)=\mu(dx)du$ is the product measure of $\mu$ and $\lambda$, so that trivially the first and second marginals of $\hat\pi$ are $\mu$ and $\lambda$, respectively), if the third marginal of $\hat\pi$ is $\nu$, and if $\hat\pi(dx,du,dy)=\mu(dx)du\hat\pi_{x,u}(dy)$ is such that $\int |y|\pi_{x,u}(dy)<\infty$ and $\int y\pi_{x,u}(dy)=x$ for $(\mu\otimes\lambda)$-a.e. $(x,u)\in(\alpha_\mu,\beta_\mu)\times(0,1)$. We again have that $\hat\Pi_M(\mu,\nu)$ is non-empty if and only if $\mu\leq_{cx}\nu$. Indeed, take $\pi\in\Pi_M(\mu,\nu)$ and set $\tilde\pi(dx,du,dy)=\mu(dx)du\pi_x(dy)$. Then $\tilde\pi\in\hat\Pi_M(\mu,\nu)$. Also note that, if $\hat\pi\in\hat\Pi_M(\mu,\nu)$, then  $\int_{u \in (0,1)} d\hat\pi\in\Pi_M(\mu,\nu)$.

Next, we introduce the set of models that are induced by $\hat\Pi_M(\mu,\nu)$. Let $\Omega^{rand}= \{(x,u,y): \alpha_\mu < x < \beta_\mu,~0<u<1,~ \alpha_\nu < y < \beta_\nu \}$, $\sF^{rand}= \sigma(\{ A\times V \times B: A \in \sB((\alpha_\mu,\beta_\mu)),~ V \in \sB((0,1)),~B \in \sB((\alpha_\nu,\beta_\nu))\})$, $\sF^{rand}_0 = \{\Omega,\emptyset\}$, $\sF^{rand}_1 = \sigma(\{ A\times V \times (\alpha_\nu,\beta_\nu) :  A  \in \sB((\alpha_\mu,\beta_\mu)),~V \in \sB((0,1))\}$), $\sF^{rand}_2 = \sF^{rand}$. Define
\begin{equation}\label{eq:UcanonicalModels}
{M}^{rand}(\mu,\nu):=\{(\Omega^{rand},\sF^{rand},{\mathbb F}^{rand}=(\sF^{rand}_0,\sF^{rand}_1,\sF^{rand}_2),\PP^{rand}=\hat\pi):\hat\pi\in\hat\Pi_M(\mu,\nu)\}.
\end{equation}
Then $M^{rand}(\mu,\nu)\subset M(\mu,\nu)$, and for each $\hat\pi\in\hat\Pi_M(\mu,\nu)$ there are random variables $(X,U,Y)=(Z_1,U,Z_2)$ such that $\PP( X \in dx,U\in du, Y \in dy) = \hat \pi(dx,du,dy)$ (the corresponding model is denoted by $\sM_{\hat\pi}\in M^{rand}(\mu,\nu)$). The set of ${\mathbb F}^{rand}$-stopping times (taking values in $\{1,2\}$) is denoted by $\sT^{rand}_{1,2}$. Set
\begin{align}\label{eq:UprimalCanonical}
\sP^{rand} = \sP^{rand}(\mu,\nu;a,b) = \sup_{\sM \in M^{rand}(\mu,\nu)} \sup_{\tau \in \sT^{rand}_{1,2}} \E^{\sM}[ c(Z_\tau,\tau)],
\end{align}
and note that
\begin{align*}
\sP\geq \sP^{rand}=\sup_{\hat\pi\in\hat\Pi_M(\mu,\nu)}\sup_{\hat B\in \sigma( \sB((\alpha_\mu,\beta_\mu))\times\sB((0,1)))}\Bigg\{&\int_{\hat B}a(x)\mu(dx)du\\&+\int_{\left((\alpha_\mu,\beta_\mu)\times(0,1)\right)\setminus\hat B}\left(\int_\R b(y)\hat\pi_{x,u}(dy)\right)\mu(dx)du\Bigg\}.
\end{align*}
\begin{lem}\label{lem:primalBounds} Let $\sP,\sP^{can}$ and $\sP^{rand}$ be as in \eqref{eq:primal}, \eqref{eq:primalCanonical} and \eqref{eq:UprimalCanonical}, respectively. Then $\sP^{can}\leq\sP^{rand}\leq\sP$.
\end{lem}
\begin{proof}
This follows immediately by observing that each $\pi\in\Pi_M(\mu,\nu)$ induces $\tilde\pi\in\hat\Pi_M(\mu,\nu)$ via $\tilde\pi(dx,du,dy)=\mu(dx)du\pi_x(dy)$ and that any stopping time $\tau$ in the model $\sM_\pi$ induces a stopping time $\tilde{\tau} $ for the model $\sM_{\tilde{\pi}}$ using that if $\{(x,y) : \tau = 1 \}=A \times (\alpha_\nu,\beta_\nu) \in \sF^{can}_1$ then $\{(x,u,y) : \tilde{\tau} = 1 \}=A \times (0,1) \times (\alpha_\nu,\beta_\nu) \in \sF^{rand}_1$. 

\end{proof}

Under an assumption that there exists special $n$-tuples ($n=4,5,6$), in Section \ref{subs:main} we will show that either $\sP^{can}=\sP^{rand}=\sP$ (see 4-point and 5-point constructions) or $\sP^{rand}=\sP$ (see 6-point construction); this is summarized in Theorem \ref{thm1}. In the first case we have that the supremum in $\sP^{can}$ is attained. The supremum in $\sP^{rand}$ is always attained.
We will prove our results by explicitly constructing optimal models that belong to $M^{can}(\mu,\nu)$ or $M^{rand}(\mu,\nu)$, respectively. Later, in Section \ref{sec:optimalModel}, we will further show that, under an additional (but natural) assumption regarding the marginals $(\mu,\nu)$ (see Definition \ref{def:dispersion}), the assumptions of Theorem \ref{thm1} are satisfied (proving the strong duality for such pairs of marginals).

\subsection{Useful classes of martingale models and martingale couplings}\label{subs:main}
In this section we identify elements of $\Gamma^{\mu,\nu}_n$ with martingale couplings $\pi \in \Pi_M(\mu,\nu)$. {\em Since the marginal distributions $\mu$ and $\nu$ are assumed to be continuous, in all the constructions below it is sufficient to describe the couplings in terms of behaviour on open inetrvals.}

\paragraph{4-point construction.} Fix $x_-,x_+\in(\alpha_\mu,\beta_\mu)$ and $y_-,y_+\in(\alpha_\nu,\beta_\nu)$ with $y_- \leq x_-<x_+ \leq y_+$. Define
\begin{equation}\label{eq:4points}
    \Pi^{y_-,x_-,x_+,y_+}_{M,4}(\mu,\nu):=\{\pi\in\Pi_M(\mu,\nu):\eqref{eq:4points1},~\eqref{eq:4points2},~\eqref{eq:4points3}\textrm{ hold}\},
\end{equation}
where
\begin{align}\label{eq:4points1}
\pi(A\times\R)=\pi(A\times A)\quad&  \forall A\in \sB((\alpha_\mu,y_-)\cup(y_+,\beta_\mu)),\\ 
\pi(A\times\R)=\pi(A\times(y_-,y_+))\quad&\forall A\in \sB((y_-,x_-)\cup(x_+,y_+))\label{eq:4points2},\\
\pi(A\times\R)=\pi(A\times(\alpha_\nu,y_-)\cup(y_+,\beta_\nu))\quad&\forall A\in \sB((x_-,x_+)).\label{eq:4points3}
\end{align}
Note that \eqref{eq:4points1} can be restated as mass in $(\alpha_\mu,y_-)\cup(y_+,\beta_\mu)$ stays where it is under $\pi$; \eqref{eq:4points2} can be restated as mass in $(y_-,x_-)\cup(x_+,y_+)$ is mapped to the set $(y_-,y_+)$ under $\pi$; \eqref{eq:4points3} can be restated as mass in $(x_-,x_+)$ is mapped to the set $(\alpha_\nu,y_-)\cup(y_+,\beta_\nu)$ under $\pi$. 

\begin{lem}
\label{lem:GammaPsi4}
  $ \Pi^{y_-,x_-,x_+,y_+}_{M,4}(\mu,\nu)\neq\emptyset$ if and only if 
$(y_-,x_-,x_+,y_+) \in \Gamma_4^{\mu,\nu}$. 
\end{lem}

\begin{proof}
Suppose $\pi \in  \Pi^{y_-,x_-,x_+,y_+}_{M,4}(\mu,\nu)$. Then from \eqref{eq:4points1} we can infer that $\mu \leq \nu$ on 
$(\alpha_\mu,y_-)\cup(y_+,\beta_\mu)$.
From \eqref{eq:4points2} we can infer that mass in $(y_-,x_-)\cup(x_+,y_+)$ is mapped to 
$(y_-,y_+)$ and from \eqref{eq:4points1} and \eqref{eq:4points3} we know that no other mass is mapped to $(y_-,y_+)$; hence $\mu|_{(y_-,x_-)\cup(x_+,y_+)} \leq_{cx} \nu|_{(y_-,y_+)}$. It follows that 
$(y_-,x_-,x_+,y_+) \in \Gamma_4^{\mu,\nu}$.

Conversely, suppose $(y_-,x_-,x_+,y_+) \in \Gamma_4^{\mu,\nu}$.
Let the martingale couplings $(\pi^k_4)_{k=1,2,3}$ be such that $\pi^1_4 \in \Pi_M(\mu|_{(\alpha_\nu,y_-) \cup (y_+,\beta_\nu)}, \mu|_{(\alpha_\nu,y_-) \cup (y_+,\beta_\nu)})$, $\pi^2_4 \in \Pi_M(\mu|_{(y_-,x_-)\cup(x_+,y_+)},\nu|_{(y_-,y_+)})$ and 
$\pi^3_4 \in \Pi_M(\mu|_{(x_-,x_+)}, (\nu - \mu)_{(\alpha_\nu,y_-) \cup (y_+,\beta_\nu)})$. In each case we know that such a martingale coupling must exist by Lemma~\ref{lem:musthave4} and Theorem~\ref{thm:strassen}. 
Then if we define $\pi_4 := \sum_{k=1}^3 \pi_4^k$ then it is clear that $\pi_4$ is a martingale coupling with marginals $\mu$ and $\nu$ and hence that $\pi_4 \in \Pi^{y_-,x_-,x_+,y_+}_{M,4}(\mu,\nu)$. 
\end{proof}



Define a candidate stopping time $\tau^{x_-,x_+}_4$ by $\tau^{x_-,x_+}_4= \tau^{x_-,x_+}$ where
\begin{equation}\label{eq:tau4point}
\tau^{x_-,x_+}=\begin{cases}1,\quad \textrm{if } X=Z_1\in(\alpha_\nu,x_-]\cup[x_+,\beta_\mu),\\
2,\quad \textrm{otherwise.}
\end{cases}\end{equation}
Note that $\tau^{x_-,x_+}_4\in\sT_{1,2}^{can}$.


\paragraph{5-point construction.} Fix $x_-,z,x_+\in(\alpha_\mu,\beta_\mu)$ and $y_-,y_+\in(\alpha_\nu,\beta_\nu)$ with $y_-<x_- \leq z \leq x_+<y_+$. Define
\begin{equation}\label{eq:5points}
    \Pi^{y_-,x_-,z,x_+,y_+}_{M,5}(\mu,\nu):=\{\pi\in\Pi_M(\mu,\nu):\eqref{eq:5points1},~\eqref{eq:5points2},~\eqref{eq:5points3},~\eqref{eq:5points4}\textrm{ hold}\},
\end{equation}
where
\begin{align}\label{eq:5points1}
\pi(A\times\R)=\pi(A\times A)\quad& \forall A\in \sB((\alpha_\mu,y_-)\cup(y_+,\beta_\mu)),\\ 
\pi(A\times\R)=\pi(A\times(y_-,z))\quad&\forall A\in \sB((y_-,x_-))\label{eq:5points2},\\
\pi(A\times\R)=\pi(A\times(z,y_+))\quad&\forall A \in \sB((x_+,y_+))\label{eq:5points3},\\
\pi(A\times\R)=\pi(A\times(\alpha_\nu,y_-)\cup(y_+,\beta_\nu))\quad&\forall A \in \sB((x_-,x_+)).\label{eq:5points4}
\end{align}
{Note that \eqref{eq:5points1} can be restated as mass in $(\alpha_\mu,y_-)\cup(y_+,\beta_\mu)$ stays where it is under $\pi$; \eqref{eq:5points2} (resp., \eqref{eq:5points3}) can be restated as mass in $(y_-,x_-)$ (resp., $(x_+,y_+)$) is mapped to the set $(y_-,z)$ (resp.,  $(z,y_+)$) under $\pi$; \eqref{eq:5points4} can be restated as mass in $(x_-,x_+)$ is mapped to the set $(\alpha_\nu,y_-)\cup(y_+,\beta_\nu)$ under $\pi$. }

The proof of the following result is similar to the proof of Lemma~\ref{lem:GammaPsi4} {(except that in this case we use Lemma \ref{lem:musthave5} in place of Lemma~\ref{lem:musthave4}, together with Theorem \ref{thm:strassen})}.
\begin{lem}
\label{lem:GammaPsi5}
  $ \Pi^{y_-,x_-,z,x_+,y_+}_{M,5}(\mu,\nu)\neq\emptyset$ if and only if 
$(y_-,x_-,z,x_+,y_+) \in \Gamma_5^{\mu,\nu}$. 
\end{lem}

In this case we take the same candidate stopping time as in the 4-point construction. In particular, set $\tau^{x_-,x_+}_5 = \tau^{x_-,x_+}$ where $\tau^{x_-,x_+}$ is as defined in \eqref{eq:tau4point}

\paragraph{6-point construction.} Fix $x_-,z_-,z_+,x_+\in(\alpha_\mu,\beta_\mu)$ and $y_-,y_+\in(\alpha_\nu,\beta_\nu)$ with $y_-<x_-\leq z_- < z_+ \leq x_+<y_+$. Define
\begin{equation}\label{eq:6points}
    \hat\Pi^{y_-,x_-,z_-,z_+,x_+,y_+}_{M,6}(\mu,\nu):=\{\hat\pi\in\hat\Pi_M(\mu,\nu):\eqref{eq:6points1},~\eqref{eq:6points2},~\eqref{eq:6points3},~\eqref{eq:6points4},~\eqref{eq:6points5},~\eqref{eq:6points6}\textrm{ hold}\},
\end{equation}
where, for each $\hat\pi\in\hat\Pi(\mu,\nu)$ we write $\hat\pi(dx,du,dy)=\mu(dx)\hat\pi_x(du,dy)$, and
\begin{eqnarray}
\hat\pi(A\times V\times\R) & = & \hat\pi(A\times V\times A),  \label{eq:6points1} \\ & &  \forall A\times V \in \sB (((\alpha_\mu,y_-)\cup(y_+,\beta_\mu)) \times(0,1)); \nonumber \\
    \hat\pi(A\times V\times\R) &= & \hat\pi(A\times V \times (y_-,z_-)), \label{eq:6points2} \\
&& \forall A\times V \in \sB {((y_-,x_-))} \times(0,1)); \nonumber \\
\label{eq:6points3}
    \hat\pi(A\times V\times\R) &=&\hat\pi(A\times V\times(z_+,y_+)), \\  
&& \forall A\times V \in \sB ((x_+,y_+) \times(0,1)); \nonumber \\
\label{eq:6points4}
    \hat\pi(A\times V\times\R)&=&\hat\pi(A\times V\times(\alpha_\nu,y_-)\cup(y_+,\beta_\nu)), \\
&& \forall A\times V \in \sB(((x_-,z_-)\cup(z_+,x_+)) \times (0,1)); \nonumber  \\    
\label{eq:6points5}
\hat\pi_x(V\times\R)& = &\hat\pi_x(V\times\{x\}), \\
&& \forall V \in \sB \left(\left(0,\frac{\nu(dx)}{\mu(dx)}=\frac{\eta(x)}{\rho(x)}\wedge1\right)\right) \mbox{for $\mu$-a.e. $x\in(z_-,z_+)$;} \nonumber \\
\label{eq:6points6}   \hat\pi_x(V\times\R) &=&\hat\pi_x(V\times(\alpha_\nu,y_-)\cup(y_+,\beta_\nu)), \\ 
&& \forall V \in \sB \left(\left(\frac{\nu(dx)}{\mu(dx)}=\frac{\eta(x)}{\rho(x)}\wedge1,1\right)\right) \mbox{for $\mu$-a.e. $x\in(z_-,z_+)$.} \nonumber
\end{eqnarray}
Here {\eqref{eq:6points1}--\eqref{eq:6points4} have a similar interpretation as in the 4- and 5-point construction. For  \eqref{eq:6points5} and \eqref{eq:6points6} mass at $x$ in $(z_-,z_+)$ either stays where it is with probability $\frac{\eta(x)}{\rho(x)}$ or moves to the tails of $\nu$ (i.e., outside $(y_-,y_+)$) with probability $1-\frac{\eta(x)}{\rho(x)}$.}
\begin{lem}
\label{lem:GammaPsi6}
  $\hat\Pi^{y_-,x_-,z_-,z_+,x_+,y_+}_{M,6}(\mu,\nu)\neq\emptyset$ and $\rho\geq\eta$ on $(z_-,z_+)$ if and only if   $(y_-,x_-,z_-,z_+,x_+,y_+) \in \Gamma_6^{\mu,\nu}$.
\end{lem}

\begin{proof}
 Suppose $\hat\pi\in\hat\Pi_{M,6}^{y_-,x_-,z_-,z_+,x_+,y_+}$ and $\rho\geq\eta$ on $(z_-,z_+)$. From \eqref{eq:6points1} we infer that $\mu\leq\nu$ on $(\alpha_\mu,y_-)\cup(y_+,\beta_\mu)$. From \eqref{eq:6points2} (resp., \eqref{eq:6points3}) we have that $(y_-,x_-)$ (resp., $(x_+,y_+)$) is mapped to $(y_-,z_-)$ (resp., $(z_+,y_+)$), while from \eqref{eq:6points4}, \eqref{eq:6points5} and \eqref{eq:6points6} we see that no other mass is mapped to $(y_-,z_-)$ (resp., $(z_+,y_+)$), and thus $\mu\lvert_{(y_-,x_-)}\leq_{cx}\nu\lvert_{(y_-,z_-)}$ (resp.,  $\mu\lvert_{(x_+,y_+)}\leq_{cx}\nu\lvert_{(z_+,y_+)}$). Also, from the assumption that $\rho\geq\eta$ on $(z_-,z_+)$, we immediately have that $\nu\lvert_{(z_-,z_+)}\leq \mu\lvert_{(z_-,z_+)}$ (note that in \eqref{eq:6points5} and \eqref{eq:6points6}  we have that $(\eta(x)/\rho(x))\wedge1=\eta(x)/\rho(x)$ for all $x\in(z_-,z_+)$). It follows that $(y_-,x_-,z_-,z_+,x_+,y_+)\in\Gamma^{\mu,\nu}_6$.

Now suppose $(y_-,x_-,z_-,z_+,x_+,y_+)\in\Gamma^{\mu,\nu}_6$. Then by the definition of $\Gamma^{\mu,\nu}_6$ we immediately have that $\nu\leq\mu$ (and thus $\eta\leq\rho$) on $(z_-,z_+)$. Take $\pi^1\in\Pi_M(\mu\lvert_{(\alpha_\mu,y_-)\cup(y_+,\beta_\mu)},\mu\lvert_{(\alpha_\mu,y_-)\cup(y_+,\beta_\mu)})$ (so that $\pi^1_x=\delta_x$), $\pi^2\in\Pi_M(\mu\lvert_{(y_-,x_-)},\nu\lvert_{(y_-,z_-)})$, $\pi^3\in\Pi_M(\mu\lvert_{(x_+,y_+)},\nu\lvert_{(z_+,y_+)})$, $\pi^4\in\Pi_M(\nu\lvert_{(z_-,z_+)},\nu\lvert_{(z_-,z_+)})$ (so that again $\pi^4_x=\delta_x$), $\pi^5\in\Pi_M(\mu\lvert_{(x_-,z_-)\cup(z_+,x_+)} +(\mu-\nu)|_{(z_-,z_+)} ,(\nu - \mu)_{(\alpha_\nu,y_-) \cup (y_+,\beta_\nu)})$. Note that all five couplings exist by the definition of $\Gamma^{\mu,\nu}_6$, Lemma \ref{lem:musthave6} and Theorem \ref{thm:strassen}. Now define $\hat\pi(dx,du,dy)=\mu(dx)du\hat\pi_{x,u}(dy)$ by setting 
$$
\hat\pi_{x,u}=\begin{cases}
    \pi^1_x,\quad (x,u)\in ((\alpha_\mu,y_-)\cup(y_+,\beta_\mu))\times(0,1)\\
    \pi^2_x,\quad (x,u)\in(y_-,x_-)\times(0,1)\\
    \pi^3_x,\quad (x,u)\in(x_+,y_+)\times(0,1)\\
    \pi^4_x,\quad (x,u)\in(z_-,z_+)\times \left(0,\frac{\eta(x)}{\rho(x)}\right)\\
    \pi^5_x,\quad(x,u)\in\left(((x_-,z_-)\cup(z_+,x_+))\times(0,1)\right)\cup\left((z_-,z_+)\times\left(\frac{\eta(x)}{\rho(x)},1\right)\right).
\end{cases}
$$
Then it is easy to see that $\hat\pi\in\hat\Pi_{M,6}^{y_-,x_-,z_-,z_+,x_+,y_+}$.
\end{proof}

Define a candidate stopping time $\hat\tau^{x_-,z_-,z_+,x_+}$ by
\begin{equation}\label{eq:tau6point}
\hat\tau^{x_-,z_-,z_+,x_+}=\begin{cases}1,\quad \textrm{if } X=Z_1\in(\alpha_\nu,x_-)\cup(x_+,\beta_\mu)\textrm{ or }X=Z_1\in(z_-,z_+)\textrm{ and } U\in\left(0,\frac{\eta(X)}{\rho(X)}{ \wedge1}\right),\\
2,\quad \textrm{otherwise.}
\end{cases}\end{equation}

We are now in a position to prove Theorem~\ref{thm:main1} which will follow as an immediate corollory of the following result and Lemmas~\ref{lem:GammaPsi4}, \ref{lem:GammaPsi5}, and \ref{lem:GammaPsi6}. 
%
%
Recall that $c(\cdot,1)=a$ and $c(\cdot,2)=b$. Also recall Problem \ref{prob:restrictedDual} and the definition of the restricted dual value $\tilde\sD=\inf_{ \psi \in \tilde{S}(b) }\widetilde{HC}(\psi)$.

{
\begin{theorem}\label{thm1}
\begin{itemize}
\item Case 1: suppose that $(y_-,x_-,x_+,y_+)\in \Sigma^{\mu,\nu}_4$ is such that
$$
\Pi^{y_-,x_-,x_+,y_+}_{M,4}(\mu,\nu)\neq\emptyset\quad\textrm{and}
\quad (y_-,x_-,x_+,y_+)\in  \sL^4_{a,b}.$$
Then $\Psi^{y_-,x_-,x_+,y_+}_{4,a,b}\in\tilde\sS(b)$. Moreover, the 4-point martingale coupling construction with stopping rule $\tau^{x_-,x_+}$ given in \eqref{eq:tau4point} and the hedge $\Psi^{y_-,x_-,x_+,y_+}_{4,a,b}$ are  optimal:
$$
\E^{\sM_\pi}[c(Z_{\tau^{x_-,x_+}},\tau^{x_-,x_+})]=\sP^{can}=\sP=\sD=\tilde\sD=\widetilde{HC}(\Psi^{y_-,x_-,x_+,y_+}_{4,a,b})
$$
for all $\pi\in\Pi^{y_-,x_-,x_+,y_+}_{M,4}$.
\item Case 2: 
Suppose that $(y_-,x_-,z,x_+,y_+)\in \Sigma^{\mu,\nu}_5$ is such that
$$
\Pi^{y_-,x_-,z,x_+,y_+}_{M,5}(\mu,\nu)\neq\emptyset\quad\textrm{and}\quad (y_-,x_-,z,x_+,y_+)\in  \sL^5_{a,b}
$$
Then $\Psi_{5,a,b}^{y_-,x_-,z,x_+,y_+}\in\tilde\sS(b)$. Moreover, the 5-point martingale coupling construction with stopping rule $\tau^{x_-,x_+}$ given in \eqref{eq:tau4point} and the hedge $\Psi_{5,a,b}^{y_-,x_-,z,x_+,y_+}$ are optimal:
$$
\E^{\sM_\pi}[c(Z_{\tau^{x_-,x_+}},\tau^{x_-,x_+})]=\sP^{can}=\sP=\sD=\tilde\sD=\widetilde{HC}(\Psi_{5,a,b}^{y_-,x_-,z,x_+,y_+})
$$
for all $\pi\in\Pi^{y_-,x_-,z,x_+,y_+}_{M,5}$.
\item Case 3: 
suppose that $(y_-,x_-,z_-,z_+,x_+,y_+)\in \Sigma^{\mu,\nu}_6$ is such that
$\nu\leq\mu$ on $(z_-,z_+)$, and 
$$
\hat\Pi^{y_-,x_-,z_-,z_+,x_+,y_+}_{M,6}(\mu,\nu)\neq\emptyset\quad\textrm{and}\quad (y_-,x_-,z_-,z_+,x_+,y_+)\in  \sL^6_{a,b}
$$
Then $\Psi_{6,a,b}^{y_-,x_-,z_-,z_+,x_+,y_+}\in\tilde\sS(b)$. Moreover, the 6-point martingale coupling construction with stopping rule $\tau^{x_-,x_+}$ given in \eqref{eq:tau6point} is optimal:
$$
\E^{\sM_{\hat\pi}}[c(Z_{\hat\tau^{x_-,z_-,z_+,x_+}},\hat\tau^{x_-,z_-,z_+,x_+})]=\sP^{rand}=\sP=\sD=\tilde\sD=\widetilde{HC}(\Psi_{6,a,b}^{y_-,x_-,z_-,z_+,x_+,y_+})
$$
for all $\hat\pi\in\hat\Pi^{y_-,x_-,z_-,z_+,x_+,y_+}_{M,6}$.
\end{itemize}
\end{theorem}
}

\begin{proof} 

We prove Case 3 here; the other cases are proved in Appendix \ref{appendix}. Cases 1 and 2 use similar arguments but are simpler since they do not require randomisation.

Since, by assumption, $\Psi_{6}:=\Psi_{6,a,b}^{y_-,x_-,z_-,z_+,x_+,y_+}$ is convex, it is continuous at each $w\in\{y_-,x_-,z_-,z_+,x_+,y_+\}$, and therefore $\Psi_{6}(y_-)=b(y_-)=T^{y_-,x_-}_{b,a}(y_-)$, $\Psi_{6}(z_-)=T^{y_-,x_-}_{b,a}(z_-)=a(z_-)$, $\Psi_{6}(z_+)=T^{x_+,y_+}_{a,b}(z_+)=a(z_+)$ and $\Psi_{6}(y_+)=b(y_+)=T^{x_+,y_+}_{a,b}(y_+)$. Then, since $a$ is convex, we have that $\Psi_{6}\geq a\geq b$ on $[x_-,x_+]$ and $\Psi_{6}\leq a$ on $\R\setminus(x_-,x_+)$. On the other hand, by the convexity of $b$, we also have that $\Psi_{6}\geq b$ on $[y_-,y_+]$ (recall that $\Psi_{6}=b$ on $\R\setminus[y_-,y_+]$). It follows that $\Psi_{6}\in\tilde\sS(b)$. Set $\psi=\Psi_{6}^{y_-,x_-,z_-,z_+,x_+,y_+}$, $\phi=(a-\psi)^+$, $T_-=T^{y_-,x_-}_{b,a}$ and $T_+=T^{x_+,y_+}_{a,b}$. Then
\begin{eqnarray*}
 \lefteqn{\int_\R\phi(x)\mu(dx)+\int_\R\psi(y)\nu(dy) }\\
&=&\int_{(\alpha_\mu,y_-) \cup (y_+,\beta_\mu)}(a(x)-b(x))\mu(dx) + \int_{y_-}^{x_-}(a(x)- T_-(x))\mu(dx) + \int^{y_+}_{x_+}(a(x)- T_+(x))\mu(dx)\\
&& \quad +\int_{(\alpha_\nu,y_-) \cup (y_+,\beta_\nu)}b(y)\nu(dy)+\int_{y_-}^{z_-} T_-(y)\nu(dy) 
+\int_{z_+}^{y_+} T_+(y)\nu(dy) 
+\int_{z_-}^{z_+}a(y)\nu(dy)\\
&=&\int_{(\alpha_\mu,x_-) \cup (x_+,\beta_\nu)}a(x)\mu(dx)+\int_{(\alpha_\nu,y_-) \cup (y_+,\beta_\nu)}b(y)(\nu-\mu)(dy)+\int_{z_-}^{z_+}a(y)\nu(dy) \\
&& \quad + \left\{\int^{y_+}_{z_+} T_+(y)\nu(dy)-\int^{y_+}_{x_+} T_+(x)\mu(dx)\right\} + \left\{\int_{y_-}^{z_-} T_-(y)\nu(dy)-\int_{y_-}^{x_-} T_-(x)\mu(dx)\right\}\\
&=&\int_{(\alpha_\mu,x_-) \cup (x_+,\beta_\nu)}a(x)\mu(dx)+\int_{(\alpha_\nu,y_-) \cup (y_+,\beta_\nu)}b(y)(\nu-\mu)(dy)+\int_{z_-}^{z_+}a(y)\nu(dy)
\end{eqnarray*}
where we use that $\hat\Pi^{y_-,x_-,z_-,z_+,x_+,y_+}_{M}(\mu,\nu)\neq\emptyset$ (and thus $\mu|_{(y_-,x_-)} \leq_{cx} \nu|_{(y_-,z_-)}$ and $\mu|_{(x_+,y_+)} \leq_{cx} \nu|_{(z_+,y_+)}$, see Lemma \ref{lem:GammaPsi6} and \eqref{eq:Gamma6def}), to show that both the bracketed terms are zero. Also (again see \eqref{eq:Gamma6def}), $\mu\leq\nu$ on $(\alpha_\nu,y_-)\cup(y_+,\beta_\nu)$, and thus $(\nu-\mu)$ is a well-defined (non-negative) measure on $(\alpha_\nu,y_-)\cup(y_+,\beta_\nu)$.

 Finally, fix $\hat\pi\in\hat\Pi^{y_-,x_-,z_-,z_+,x_+,y_+}_{M,6}(\mu,\nu)$ and let $\sM_{\hat\pi}\in M^{rand}(\mu,\nu)$ be the corresponding model. By hypothesis, $\nu\leq\mu$ (and thus $\eta\leq\rho$) on $(z_-,z_+)$. Also, recall the definition of the candidate stopping time $\hat\tau=\hat\tau^{x_-,z_-,z_+,x_+}\in\sT^{rand}_{1,2}$, see \eqref{eq:tau6point}. Then
\begin{eqnarray*}
\lefteqn{\left( \int_\R\phi(x)\mu(dx)+\int_\R\psi(y)\nu(dy) \right) }\\
& = & \int I_{ \{ x \in (\alpha_\mu,x_-)\cup(x_+,\beta_\mu) \} } a(x) \mu(dx)  +  \int I_{ \{ x \in (z_-,z_+) \} }a(x)\mu(dx)  \frac{\eta(x)}{\rho(x)} \\
&& \hspace{10mm} +  \int I_{ \{ y \in (\alpha_\nu, y_-) \cup (y_+,\beta_\nu) \} } b(y) (\nu(dy)-\mu(dy))  \\
& = & \E^{\sM_{\hat\pi}}[ a(Z_1) I_{ \{ \hat\tau=1 \} } +  b(Z_2) I_{ \{ \hat\tau=2 \} } ],
\end{eqnarray*}
where we use \eqref{eq:6points4}, \eqref{eq:6points5} and \eqref{eq:6points6} for the last equality. This completes the proof.
\end{proof}

\begin{remark} 
{(i) To get strong duality $\sP=\sD$ in Case 3 we construct a model $\sM \in M^{rand}(\mu,\nu)$ and show $\sP^{rand} = \sP = \sD$.
In Section~\ref{sec:can=rand} we argue that {(under the Dispersion Assumption; see Definition \ref{def:dispersion})} $\sP^{can}=\sP^{rand}$, however, the supremum for $\sP^{can}$ is not attained.}

Note that in the financial context there is no reason to expect the price process to be the only source of information in the financial market. There may be multiple scenarios which lead to the price process arriving at the same price point at time one, and these different scenarios may lead to different dynamics over future time intervals. As described in Hobson and Neuberger~\cite{HobsonNeuberger:17}, it is this extra information which gives the full value of Bermudan (and American) options over their European counterparts.

    In the case of Bermudan put options studied in Hobson and Norgilas~\cite{HobsonNorgilas:19} 
    this richer structure is not required. This paper shows that the put case is rather special.

(ii) Whether the 4-point, 5-point or 6-point constructions are optimal depends, in general, 
on the payoff functions $a,b$ and the marginal distributions $\mu,\nu$. The goal of the next two sections is to show that for a large class of marginals $(\mu,\nu)$ (and for arbitrary convex $a\geq b$) we can always find a solution of one of the three types (as in Theorem \ref{thm1}). 
\end{remark}

\section{Shadows and the left- and right-curtain martingale couplings}\label{sec:shadows}

\subsection{Shadow couplings}
We recall the notion of the shadow measure (see Beiglb\"ock and Juillet \cite{BeiglbockJuillet:16} and Beiglb\"ock et al. \cite{BeiglbockHobsonNorgilas:22}).

Two (Borel) measures $\xi,\chi$ on $\R$ are said to be in \textit{extended convex order}, denoted by $\xi\leq_E\chi$, if $\int f d\xi\leq\int fd\chi$ for all non-negative and convex $f:\R\to\R_+$. (It is easy to see that if $\mu\leq_{cx}\nu$, then any $\tilde\mu\leq\mu$ satisfies $\tilde\mu\leq_E\nu$.) If $\xi\leq_E\chi$, then $\{\theta:\xi\leq_{cx}\theta\leq\chi\}$ is non-empty and admits the minimal element with respect to $\leq_{cx}$.

For any (Borel) measure $\xi$ on $\R$, we define the (put) potential $P_\xi:\R\to\R$ by $P_\xi(k)=\int_\R(k-x)^+\xi(dx)$, $k\in\R$. It is well-known that the potenial $P_\xi$ uniquely identifies the underlying measure (via the second (distributional) derivative of $P_\xi''$). Furthermore, $\xi\leq_{cx}\chi$ if and only if $P_\xi\leq P_\chi$ on $\R$ and $\xi,\chi$ have the same total mass and mean. See, for example, Chacon \cite{Chacon}, Chacon and Walsh \cite{ChaconW} or Hirsch and Roynette \cite{HR} for the relevant properties of potential functions.

Recall that given $f:\R\to\R$ we denote the convex hull of $f$ by $f^c$. 
\begin{defn}\label{def:shadow}
    Fix $\xi\leq_E\chi$. The shadow measure of $\xi$ in $\chi$, denoted by $S^\chi(\xi)$, is the unique measure that satisfies
    \begin{enumerate}
        \item $\xi\leq_{cx}S^\chi(\xi)\leq\chi$,
        \item $S^\chi(\xi)\leq_{cx}\theta$ for all measures $\theta$ such that $\xi\leq_{cx}\theta\leq\chi$.
\end{enumerate}
\end{defn}

\begin{prop}[Beiglb\"{o}ck et al~\cite{BeiglbockHobsonNorgilas:22}]
\label{prop:potential}
The potential of $S^\chi(\xi)$ is given by
$$
P_{S^\chi(\xi)}=P_\chi-(P_\chi-P_\xi)^c.
$$
\end{prop}

One of the main structural properties of the shadow measure (see Beiglb\"ock and Juillet \cite{BeiglbockJuillet:16,BJ:21}, Beiglb\"ock et al. \cite{BeiglbockHobsonNorgilas:22}), 
is the associativity property.

\begin{prop}[Associativity property of the shadow measure]
\label{prop:shadow_assoc}
    If $\xi_1+\xi_2$ is such that $\xi_1+\xi_2\leq_{E}\chi$, then
    $$
    \xi_2\leq_E(\chi-S^\chi(\xi_1))\quad\textrm{and}\quad S^\chi(\xi_1+\xi_2)=S^\chi(\xi_1)+S^{\chi-S^\chi(\xi_1)}(\xi_2).
    $$
\end{prop}

The following lemmas will be useful in the proofs of our main results.

\begin{lem}\label{lem:specialMeasures1}
    Suppose $\xi\leq_E\chi$ and $S^\chi(\xi)\leq\eta$ for some $\eta\leq\chi$.  Then $\xi\leq_E\eta$ and $S^\chi(\xi)=S^\eta(\xi)$.
\end{lem}
\begin{proof}
    Since $\xi\leq_{cx}S^\chi(\xi)\leq\eta$, for any non-negative and convex $f:\R\to\R$ we have that $\int f d\xi\leq\int fd S^\chi(\xi)\leq\int fd\eta$, and thus $\chi\leq_E\eta$. It follows that $S^\eta(\xi)$ is well-defined, and we are left to prove that $S^\chi(\xi)=S^\eta(\xi)$. First, since $\xi\leq_{cx}S^\eta(\xi)\leq\eta \leq \chi$, by the definition of the shadow measure (Definition \ref{def:shadow}) we have that $S^\chi(\xi)\leq_{cx}S^\eta(\xi)$. On the other hand, since $\xi\leq_{cx}S^\chi(\xi)\leq\eta$, {we have that} $S^\eta(\xi)\leq_{cx}S^\chi(\xi)$.  Hence, $S^\chi(\xi) = S^\eta(\xi)$.
\end{proof}

\begin{lem}\label{lem:specialMeasures2}
    Suppose $\xi\leq_E\chi$, are both continuous and $\chi((\alpha_\xi,\beta_\xi))=0$. Then $S^\chi(\xi)=\chi\lvert_{(x_-,\alpha_\xi)\cup(\beta_\xi,x_+)}$ for some $x_-<\alpha_\xi<\beta_\xi<x_+$.
\end{lem}
\begin{proof}
By Proposition \ref{prop:potential} we have that $P_{\chi-S^\chi(\xi)}=P_\chi-P_{S^\chi(\xi)}=(P_\chi-P_\xi)^c$. Now observe that, due to our assumptions on $\xi$ and $\chi$, $(P_\chi-P_\xi)$ is convex (but not linear) and strictly increasing on $(\alpha_\chi,\alpha_\xi)$, concave (but not linear) on $(\alpha_\xi,\beta_\xi)$ and again convex (but not linear) on $(\beta_\xi,\beta_\chi)$; moreover $(P_\chi-P_\xi)$ is strictly positive on $(\alpha_\chi,\beta_\chi)$. It follows that $(P_\chi-P_\xi)>(P_\chi-P_\xi)^c$ on $(\alpha_\xi,\beta_\xi)$, and thus $(P_\chi-P_\xi)^c$ must be linear on some interval $I:=(x_-,x_+)\supset(\alpha_\xi,\beta_\xi)$ and $(P_\chi-P_\xi)=(P_\chi-P_\xi)^c$ on $\mathbb R\setminus I$. On one hand this gives that $(\chi-S^\chi(\xi))$ does not charge $I$. On the other hand, on $(\alpha_\chi,x_-)\cup(x_+,\beta_\chi)$,  the second (distributional) derivatives of $(P_\chi-P_\xi)^c$ and $P_\chi$ coincide (since $\xi$ does not charge this region and thus $P_\xi$ is linear there), and therefore $(\chi-S^\chi(\xi))=\chi\lvert_{(\alpha_\chi,x_-)\cup(x_+,\beta_\chi)}$. It follows that $S^\chi(\xi)=\chi\lvert_{(x_-,\alpha_\xi)\cup(\beta_\xi,x_+)}$ as claimed.
\end{proof}

\subsection{Curtain couplings}

The shadow measure allows us to construct martingale couplings $\pi \in \Pi_M(\mu,\nu)$ which later can be used to construct models $\sM_\pi \in M^{can}(\mu,\nu)$. 
Recall the definition of $\pi^{f,g}_v$
from the introduction.

\begin{theorem}\label{thm:curtain}
[Beiglb\"{o}ck and Juillet~\cite{BeiglbockJuillet:16}]
Fix $\mu\leq_{cx}\nu$.


(i) There exists the unique martingale coupling of $\mu$ and $\nu$, denoted by $\pi^{L}$ and called the left-curtain coupling of $\mu$ and $\nu$, such that, for each $x\in(\alpha_\mu,\beta_\mu)$,
$$
\pi^L((\alpha_\mu,x]\times B)=S^\nu(\mu\lvert_{(\alpha_\mu,x]})(B)\quad for~all~Borel~B\subseteq(\alpha_\nu,\beta_\nu).
$$

Furthermore, under Standing Assumption \ref{sass:densities}, $\pi^{L}(dx,dy) = \rho(x) dx \pi^{f^L(x),g^L(x)}_x(dy)$, where $f^L,g^L : (\alpha_\mu,\beta_\mu) \to (\alpha_\nu,\beta_\nu)$ are such that $f^L(x) \leq x \leq g^L(x)$, $g^L$ is increasing and {left-continuous}, $f^L(x)=g^L(x)$ on $\{ x:g^L(x)=x \}$ and if $\alpha_\mu < x<x'<\beta_\mu$ then either $f^L(x') < f^L(x)$ or $f^L(x') > g^L(x)$.

(ii) There exists the unique martingale coupling of $\mu$ and $\nu$, denoted by $\pi^{R}$ and called the right-curtain coupling of $\mu$ and $\nu$, such that, for each $x\in(\alpha_\mu,\beta_\mu)$,
$$
\pi^L([x,\beta_\mu)\times B)=S^\nu(\mu\lvert_{[x,\beta_\mu)})(B)\quad for~all~Borel~B\subseteq(\alpha_\nu,\beta_\nu).
$$

Furthermore, under Standing Assumption \ref{sass:densities}, $\pi^{R}(dx,dy) = \rho(x) dx \pi^{f^R(x),g^R(x)}_x(dy)$, where $f^R,g^R : (\alpha_\mu,\beta_\mu) \to (\alpha_\nu,\beta_\nu)$ are such that $f^R(x) \leq x \leq g^R(x)$, $f^R$ is increasing and {right-continuous}, $g^R(x)=f^R(x)$ on $\{x : f^R(x)=x \}$ and if $\alpha_\mu < x'<x <\beta_\mu$ then either $g^R(x') > g^R(x)$ or $g^R(x') < f^R(x)$.
\end{theorem}

\begin{figure}[H]
\centering
\begin{tikzpicture}[scale=1]
\begin{axis}[axis lines=middle,
            ytick=\empty,
            xtick=\empty,
             xmin=-1,xmax=20,
             ymin=-4,ymax=21,
             yticklabels={},
            xticklabels={},
            axis line style={draw=none}]

\addplot[name path=diag,black,domain={0:4}, thick] {x} node[pos=1, below]{};
\addplot[name path=diag2,black,domain={4:20}, dashed] {x} node[pos=1, below]{};

\draw[name path=tu1,thick] (4, 4) to[out=85,in=225] (15.5,17.5);
\draw[name path=tu2, thick] (15.5, 17.5) to[out=45,in=225] (20,20.5 );
\draw[name path=td, thick] (4, 4) to[out=335,in=180] (20, 0);
\draw[gray,dotted] (1,-1) -- (1,1) -- (10,1) -- (10,12.8) -- (12.8,12.8)--(12.8,-1);
\draw[gray,dotted] (10,1) -- (10,-1);
\draw[gray,dotted] (4,4) -- (4,-1);

\node (e)[scale=0.7] at (2.5,5) {$(e^L,e^L)$};
\node[circle,fill=black,inner sep=0pt,minimum size=3pt] at (4,4) {};
\node (f)[scale=1] at (16,1.5) {$f^L$};
\node (g)[scale=1] at (16,19.5) {$g^L$};
\node [below] at (4,-1) {$e^L$};
\node [below] at (1,-1) {$f^L(x)$};
\node [below] at (10,-1) {$x$};
\node [below] at (12.8,-1) {$g^L(x)$};
\end{axis}
\end{tikzpicture}

\caption{The left-curtain coupling in a special case: sketch of functions $f^L$ and $g^L$ in the case in which there exists a point $e^L$ such that $\{x:f^L(x)=x=g^L(x)\}=(\alpha_\mu,e^L]$, and $g^L>id$ and $f^L$ is decreasing on $(e^L,\beta_\mu)$.}
\label{fig:fgdispersion}
\end{figure}
{In many natural cases (for example a pair of centred Gaussian distributions, or a pair of lognormal distributions with the same mean) there exists $e^L \in [\alpha_\mu,\beta_\mu)$ such that $\rho < \eta$ on an interval $(\alpha_\nu, e^L)$ and $\rho>\eta$ on a small interval $(e^L,e^L+\epsilon)$ to the right of $e^L$. Then 
we can define $f^L=g^L=id$ on $(\alpha_\mu,e^L]$ and extend the functions rightward from $e^L$ using the relationship
\begin{equation}
    \label{eq:deffg}
\int_{f^L(x)}^x z^i \rho(z) dz = \int_{f^L(x)}^{g^L(x)} z^i \eta(z) dz; \hspace{10mm} i=0,1, 
\end{equation}
which encapsulates the idea that mass in $(f^L(x),x)$ according to the law $\mu$ can be mapped to $(f^L(x),g^L(x))$ according to the law $\nu$ in a way which respects the martingale property (and then at the margins, mass at $x$ is mapped to either $f^L(x)$ or $g^L(x)$ as in the definition of the left-curtain coupling). Under Standing Assumption~\ref{sass:densities} the solutions to \eqref{eq:deffg} can be expressed via a coupled pair of differential equations
\begin{equation}
    \label{eq:fgode}
\frac{d}{dx} g^L (x) = \frac{x - f^L(x)}{g^L(x)-f^L(x)} \frac{\rho(x)}{\eta(g(x))};
\hspace{2mm} \frac{d}{dx} f^L (x) = - \frac{g^{L}(x) - x}{g^L(x)-f^L(x)} \frac{\rho(x)}{\eta(f^L(x))- \rho(f^L(x))}  
\end{equation}
subject to $f^L(e^L)=g^L(e^L)=e^L$.
Note that $f^L$ is decreasing (at least locally, so that by hypothesis the denominators in \eqref{eq:fgode} are positive for $x>e^L$. In many (regular) cases the solutions $f^L, g^L$ to \eqref{eq:fgode} satisfy $g^L > id$ on $(e^L,\beta_\mu)$. In this case the left-curtain coupling has a simple structure, see Figure~\ref{fig:fgdispersion}.

Similar considerations apply to the right-curtain coupling. In regular cases the densities $\rho,\eta$ are such that there exists $e_R>e_L$ such that $\rho<\eta$ on $(e^R, \beta_\mu)$ and $\rho> \eta$ on an interval $(e^R-\epsilon,e^R)$ and then 
$f^R,g^R$ can be defined to the left of $e^R$ via
\[ \int_x^{g^R(x)} z^i \rho(z) dz = \int_{f^R(x)}^{g^R(x)} z^i \eta(z) dz; \hspace{10mm} i=0,1, \]
which again have a representation via a pair of coupled differential equations. In many cases $f^R$ defined in this way satisfies $f^R < id$ on $(\alpha_\mu,e^R)$ and the right-curtain coupling has a simple structure.}

The following assumption is a small modification of one introduced by Hobson
and Klimmek \cite{HobsonKlimmek:15}, see also Henry-Labord\`ere and Touzi \cite{HenryLabordereTouzi:16}. See Figure \ref{fig:DensitiesDispersion}.
{It is often satisfied when $\mu$ and $\nu$ are from the same family of centred, continuous distributions with a unimodal density such as centred normals, lognormals or uniforms. }

\begin{defn}[Dispersion Assumption]\label{def:dispersion}
{ $\mu$ and $\nu$ satisfy the dispersion assumption if $(\mu,\nu)$ satisfy Standing Assumption~\ref{sass:densities}; if
    the densities $\rho$ and $\eta$ (of $\mu$ and $\nu$) are strictly positive on $(\alpha_\mu,\beta_\mu)$ and $(\alpha_\nu,\beta_\nu)$, respectively;
    and if there exists $e^L,e^R\in [\alpha_\mu,\beta_\mu]$ with $e^L<e^R$, such that $\rho>\eta$ on $(e^L,e^R)$ and $\rho<\eta$ on ${ (\alpha_\nu,\beta_\nu)}\setminus[e^L,e^R]$.}
\end{defn}

For convenience, we introduce the following notation: for a continuous measure $\xi$ (on $\R$) and for $u \leq v$ let $\xi_u^v = \xi|_{(u,v)}$ and $\bar{\xi}_u^v = \xi - \xi_u^v = \xi|_{(-\infty,u)\cup(v,\infty)}$.

Under the Dispersion Assumption, the left-curtain and right-curtain martingale couplings have a simple structure; the proof of the following lemma can be found in Appendix \ref{appendix}.

\begin{figure}[H]
\centering
\begin{tikzpicture}[
declare function={	
    	phi(\x,\m,\s)=(1/sqrt(2*pi*\s))*exp(-pow(\m-\x,2)/(2*\s));
	s1=0.7;
	s2=2;
	m=0;
	f=-2;
	z=0.15;
	g=0.8;
	e1=sqrt(-(ln(sqrt(s2))-ln(sqrt(s1)))/(pow(2*s2,-1)-pow(2*s1,-1)));
	}]
\begin{axis}[axis lines=middle,
            ytick=\empty,
            xtick=\empty,
             xmin=-4, xmax=4,
             ymin=-0.15, ymax=0.5,
             yticklabels={},
            xticklabels={},
            axis line style={draw=none}]

\addplot[name path=base,black,domain={-5:5}] {0} node[pos=1, below]{};
\addplot[name path=rho,black,domain={-5:5},samples=100] {phi(\x,m,s1)} node at (0.55,0.45) {$\rho$};
\addplot[name path=eta,black,domain={-5:5},samples=100] {phi(\x,m,s2)} node at (1.7,0.17) {$\eta$};

\path [name path=lineA](f,0)--(f,0.5);
\draw [name intersections={of=lineA and eta}, black] (f,0) -- (intersection-1);
\draw[dashed] (f,0) -- (f,-0.06) node[below]  {$f^L(x)$};

\path [name path=lineB](-e1,0)--(-e1,0.5);
\draw [name intersections={of=lineB and eta}, black, dashed] (-e1,0) -- (intersection-1) node at (-e1+0.1,-0.03) {$e^L$};

\path [name path=lineC](z,0)--(z,0.5);
\draw [name intersections={of=lineC and rho}, black] (z,0) -- (intersection-1);
\draw[dashed] (z,0) -- (z,-0.06) node[below]  {$x$};

\path [name path=lineD](g,0)--(g,0.5);
\draw [name intersections={of=lineD and eta}, black] (g,0) -- (intersection-1);
\draw[dashed] (g,0) -- (g,-0.06) node[below]  {$g^L(x)$};

\path [name path=lineE](e1,0)--(e1,0.5);
\draw [name intersections={of=lineE and eta}, black,dashed] (e1,0) -- (intersection-1) node at (e1+0.1,-0.03) {$e^R$};

\addplot[pattern=crosshatch, pattern color=red!50] fill between[of=eta and rho, soft clip={domain=f:-e1}];
\addplot[pattern=crosshatch, pattern color=red!50] fill between[of=eta and rho, soft clip={domain=-e1:z}];
\addplot[pattern=crosshatch, pattern color=red!50] fill between[of=eta and base, soft clip={domain=z:g}];
\addplot[pattern=north west lines, pattern color=blue!50] fill between[of=rho and base, soft clip={domain=f:-e1}];
\addplot[pattern=north west lines, pattern color=blue!50] fill between[of=eta and base, soft clip={domain=-e1:z}];

	\draw[-latex', thick] (-0.2, .35) to[out=0,in=90] (0.5, 0.15);
	\draw[-latex', thick] (-0.3, .35) to[out=180,in=90] (-1.7, 0.1);

\end{axis}
\end{tikzpicture}

\caption{Sketch of the densities $\rho$ and $\eta$ under the Dispersion Assumption, and the locations of $f^L(x),g^L(x)$ for given $x>e^L$. Mass in the interval $(f^L(x),x)$
is mapped to $(f^L(x),e^L)$ or $(x,g^L(x))$ in a way which respects the martingale property.}
\label{fig:DensitiesDispersion}
\end{figure}

\begin{lem}\label{lem:couplingsDispersionAssumption}
Suppose that $\mu\leq_{cx}\nu$ satisfy the Dispersion Assumption (see Definition \ref{def:dispersion}). Let $(f^L,g^L)$ (resp., $(f^R,g^R)$) be the functions that support the left-curtain (resp., right-curtain) coupling (see Theorem \ref{thm:curtain}). Then
\begin{enumerate}
        \item $f^L,g^L,f^R,g^R:(\alpha_\mu,\beta_\mu)\to(\alpha_\nu,\beta_\nu)$ are continuous;
        \item $\{x:f^L(x)=x=g^L(x)\}=(\alpha_\mu,e^L]$ and $\{x:f^R(x)=x=g^R(x)\}=[e^R,\beta_\mu)$;
        \item $g^L$ (resp., $f^R$) is strictly increasing, while $f^L$ (resp., $g^R$) is strictly decreasing on $(e^L,\beta_\mu)$ (resp.,  $(\alpha_\mu,e^R)$).
        \item { $g^L>id$ on $(e^L,\beta_\mu)$ (resp., $f^R<id$ on $(\alpha_\mu,e^R)$.}
    \end{enumerate}
    Further $\lim_{x\uparrow\beta_\mu}f^L(x)=\alpha_\nu=\lim_{x\downarrow\alpha_\mu}f^R(x)$ and $\lim_{x\uparrow\beta_\mu}g^L(x)=\beta_\nu=\lim_{x\downarrow\alpha_\mu}g^R(x)$.
\end{lem}

\begin{cor}\label{cor:LCandRC}
Suppose that $\mu\leq_{cx}\nu$ satisfy the Dispersion Assumption (Definition \ref{def:dispersion}).

For each $x\in(\alpha_\mu,\beta_\mu)$ and $B\in \sB ((\alpha_\nu,\beta_\nu))$,
\begin{enumerate}
\item[(i)] $\pi^L((\alpha_\mu,x]\times B)=S^\nu(\mu_{\alpha_\mu}^x)(B)=\mu_{\alpha_\mu}^{f^L(x)}(B)+\nu_{f^L(x)}^{g^L(x)}(B)$
and ${\mu}_x^{\beta_\mu} \leq_{cx} \nu-S^\nu(\mu_{\alpha_\mu}^{x})$; 
\item[(ii)] $\pi^R([x,\beta_\mu)\times B)=S^\nu(\mu_{x}^{\beta_\mu})(B)=\mu^{\beta_\mu}_{g^R(x)}(B)+\nu_{f^R(x)}^{g^R(x)}(B)$
and ${\mu}^x_{\alpha_\mu} \leq_{cx} \nu-S^\nu(\mu^{\beta_\mu}_{x})$.
\end{enumerate}
\end{cor}
\begin{proof}
    This follows immediately from the definitions and properties of the shadow measure and the left- and right-curtain martingale couplings; see Definition \ref{def:shadow}, Proposition~\ref{prop:shadow_assoc} and Theorem \ref{thm:curtain}. 
    \end{proof}

We finish this section with a lemma that will be useful in proving our main results.

Recall the definitions and the properties (under the Dispersion Assumption, Definition \ref{def:dispersion}) of $g^L$ and $f^R$; see Theorem \ref{thm:curtain} and Lemma \ref{lem:couplingsDispersionAssumption}. Note that if $\beta_\mu<\beta_\nu$, then $(g^L)^{-1}(\beta_\mu)<\beta_\mu$, while in the case $\beta_\mu=\beta_\nu$ we set $(g^L)^{-1}(\beta_\mu)=\beta_\mu$. Similarly, if $\alpha_\nu<\alpha_\mu$, then $\alpha_\mu<(f^R)^{-1}(\alpha_\mu)$, and in the case $\alpha_\nu=\alpha_\mu$ we set $(f^R)^{-1}(\alpha_\mu)=\alpha_\mu$.

{The motivation behind the next definition is that we want to separate the cases where $S^\nu(\mu_{\alpha_\mu}^{x}) + S^\nu(\mu_y^{\beta_\mu})= S^\nu(\mu_{\alpha_\mu}^{x} + \mu_y^{\beta_\mu})$  and $S^\nu(\mu_{\alpha_\mu}^{x}) + S^\nu(\mu_y^{\beta_\mu}) \neq S^\nu(\mu_{\alpha_\mu}^{x} + \mu_y^{\beta_\mu})$. In the former case, we have $S^\nu(\mu_y^{\beta_\mu})=S^{\nu- S^\nu(\mu_{\alpha_\mu}^{x})}(\mu_y^{\beta_\mu})$.}

Set $\bar x=(g^L)^{-1}(e^R)$. Define $B:(\alpha_\mu,(g^L)^{-1}(\beta_\mu))\to(\alpha_\nu,\beta_\nu)$ by 
\begin{equation}\label{eq:B}
B(x):=(f^R)^{-1}(g^L(x)),\quad x\in(\alpha_\mu,(g^L)^{-1}(\beta_\mu)).
\end{equation}
In the case $\beta_\mu<\beta_\nu$, we extend the definition of $B(\cdot)$ to $(\alpha_\mu,\beta_\mu)$ by setting $B(x)=\beta_\mu$ for $x\in[(g^L)^{-1}(\beta_\mu),\beta_\mu)$. Then, using that $g^L(x)=x$ for $x\leq e^L$ and $f^R(y)=y$ for $y\geq e^R$, we equivalently have that
\begin{equation}
\label{eq:BoundaryB}    
B(x) = \begin{cases} (f^R)^{-1}(x)  & \alpha_\mu < x \leq e^L  \\
 (f^R)^{-1}(g^L(x))  &  e^L < x \leq \bar{x}  \\
g^L(x)  & \bar{x} < x \leq (g^L)^{-1}(\beta_\mu) \\
\beta_\mu & (g^L)^{-1}(\beta_\mu) < x < \beta_\mu. \end{cases}
\end{equation} 

\begin{lem}
\label{lem:new:FormOfShadow}
Suppose that $\mu\leq_{cx}\nu$ satisfy the Dispersion Assumption (Definition \ref{def:dispersion}).

Suppose that $x,y\in(\alpha_\mu,\beta_\mu)$ with $x<y$.
\begin{enumerate}
    \item[(i)]If $y\geq B(x)$, then 
    $$
    S^\nu(\bar{\mu}_x^y) = S^\nu(\mu_{\alpha_\mu}^x)+S^{\nu}(\mu_y^{\beta_\mu})=\mu_{\alpha_\mu}^{f^L(x)} + \nu_{f^L(x)}^{g^L(x)} + \nu_{f^R(y)}^{g^R(y)} + \mu_{g^R(y)}^{\beta_\mu}
    $$
    and $\mu^y_x\leq_{cx}(\nu-\mu)^{f^L(x)}_{\alpha_\nu}+\nu_{g^L(x)}^{f^R(y)}+(\nu-\mu)_{g^R(y)}^{\beta_\nu}$;
    
    \item[(ii)]If $y\leq B(x)$, then
    $$
    S^\nu(\bar{\mu}_x^y) = \mu_{\alpha_\mu}^w + \nu_w^z + \mu_{z}^{\beta_\mu}\quad\textrm{and}\quad \mu^y_x\leq_{cx}(\nu-\mu)^{w}_{\alpha_\nu}+(\nu-\mu)_{z}^{\beta_\nu},
    $$
    for some (unique) $w \leq e^L\wedge x < e^R\vee y \leq z$.
\end{enumerate}
\end{lem}

{The proof of Lemma~\ref{lem:new:FormOfShadow} is deferred to Appendix \ref{appendix}.}

\section{Existence of an optimal model}\label{sec:optimalModel}

The goal of this section is to show that, under the Dispersion Assumption on the pair $(\mu,\nu)$, the hypotheses of Theorem~\ref{thm1} are satisfied and we have found the optimal solution. {In particular, we prove that if $\mu$ and $\nu$ satisfy the dispersion assumption then $\cup_{n=4,5,6} (\Gamma^{\mu,\nu}_n \cap \sL^{a,b}_n) \neq \emptyset$
and hence, by Theorem~\ref{thm1}, we can construct optimal solutions to both the primal and dual problems and there is no duality gap.}

\begin{theorem}\label{thm2}

Suppose that $\mu\leq_{cx}\nu$ satisfy 
Definition \ref{def:dispersion}.
Then $\cup_{n=4,5,6} (\Gamma^{\mu,\nu}_n \cap \sL^{a,b}_n) \neq \emptyset$.

\end{theorem}

{The focus of this section is on proving Theorem~\ref{thm2}. Many of the proofs of intermediate results are deferred to Appendix \ref{appendix}.}

Given the structure of the left- and right-curtain couplings under the Dispersion Assumption (recall Lemma \ref{lem:couplingsDispersionAssumption}), it is rather easy to see that the sets $ \Gamma^{\mu,\nu}_4$, $\Gamma^{\mu,\nu}_5$ and $ \Gamma^{\mu,\nu}_6$ (as in \eqref{eq:Gamma4def}, \eqref{eq:Gamma5def} and \eqref{eq:Gamma6def}, respectively) are non-empty. Indeed, for $e^L<x_-<x_+<e^R$ with $g^L(x_-)= f^R(x_+)$, we have that $f^L(x_-)<e^L<x_-<g^L(x_-)=f^R(x_+)<x_+<e^R<g^R(x_+)$, and then, by setting $y_-=f^L(x_-), z=g^L(x_-)=f^R(x_+),g^R(x_+)=y_+$, we have that $\Gamma^{\mu,\nu}_4\neq\emptyset$ and $\Gamma^{\mu,\nu}_5\neq\emptyset$. (This is an immediate consequence of the definitions of curtain couplings, Corollary \ref{cor:LCandRC} and Lemma \ref{lem:new:FormOfShadow}.) Similarly, if $e^L<x_-<x_+<e^R$ with $g^L(x_-)< f^R(x_+)$, by setting $y_-=f^L(x_-), z_-=g^L(x_-), z_+=f^R(x_+),g^R(x_+)=y_+$, we have that $\Gamma^{\mu,\nu}_6\neq\emptyset$; here we use that, under the Dispersion Assumption, $\mu\lvert_{(z_-,z_+)}\leq\nu\lvert_{(z_-,z_+)}$, since $(z_-=g^L(x_-),z_+=f^R(x_+))\subset (e^L,e^R)$. 

As above, we can express the candidate points $x_-,x_+,y_-,y_+,z_-,z,z_+\in\R$ in terms of $x_-$ and $x_+$ only. Then the key insight is that by restricting the search to such pairs $(x_-,x_+)$ (which together with the functions $f^L,g^L,f^R,g^R$, give the remaining four points), we can find $(x_-^*,x_+^*)$, for which (at least) one of $\Psi_4^{y^*_-,x^*_-,x^*_+,y^*_+}$, or $\Psi_5^{y^*_-,x^*_-,z^*,x^*_+,y^*_+}$, or $\Psi_6^{y^*_-,x^*_-,z^*_-,z^*_+,x^*_+,y^*_+}$ is convex (recall \eqref{eq:psi4points}, \eqref{eq:psi5points}, \eqref{eq:psi6points}). Then using Lemmas \ref{lem:GammaPsi4}, \ref{lem:GammaPsi5}, \ref{lem:GammaPsi6} (depending on which case we are in) together with Theorem \ref{thm1} it must be the case that we have found the solutions to the primal and dual problems and the strong duality holds. The main difficulty in finding $(x^*_-,x_+^*)$ is that {\em a priori} it is not clear
which of the three types as listed in Theorem \ref{thm1} we should search for. Furthermore, in the case of the 4-point construction, we often need to consider $x_-< e^L$ or $x_+ > e^R$, and in this case the original supporting functions $f^L,g^L,f^R,g^R$ from the left- and right- curtain couplings are no-longer useful and then we need alternative arguments.
The remainder of this section is dedicated to proving Theorem \ref{thm2}. We will do it in several steps. First, in Section \ref{sec:cases}, for given $\mu\leq_{cx}\nu$, we use the supporting functions $(f^L,g^L)$ and $(f^R,g^R)$ to distinguish three cases. Then we prove Theorem \ref{thm2} separately in each case (see Sections \ref{sec:C1}, \ref{sec:C2}, \ref{sec:C3}). 

Throughout this section we make the following standing assumption:
\begin{sass}
\label{ass:simple} $\mu$ and $\nu$ are such that $\mu\leq_{cx} \nu$ and satisfy the Dispersion Assumption, in the sense of Definition \ref{def:dispersion}.
\end{sass}

\subsection{Dividing the problem into cases}\label{sec:cases}
For $x\in(e^L,\beta_\mu)$ define $T^L_x:= T_{b,a}^{f^L(x),x}$ to be the line passing through the points $(f^L(x),b(f^L(x)))$ and $(x,a(x))$. Similarly, let $T^R_x = T^{x,g^R(x)}_{a,b}$. 
Define the functions $C^L: (e^L,\beta_\nu) \to \R$ and $C^R: (\alpha_\nu,e^R) \to \R$ by  $C^L(z)= T^L_{(g^L)^{-1}(z)}(z)$ and $C^R(z)= T^R_{(f^R)^{-1}(z)}(z)$. 
We are interested in the crossing points (if any) of the curves $C^L$ and $C^R$ on $(e^L,e^R)$. There will be three cases, and the optimal hedges and corresponding extremal models
will be different in each case.

\begin{lem}\label{lem:Ccont}
$C^L : (e^L,\beta_\nu) \to \R$ and $C^R : (\alpha_\nu,e^R) \to \R$ are continuous with {$\liminf_{z \downarrow e^L} C^L(z) \geq a(e^L)$ and $\liminf_{z \uparrow e^R} C^R(z) \geq a(e^R)$.}
\end{lem}


Define $\hat{z}^L$ by $\hat{z}^L =  \inf \{ z\in (e^L,{\beta_\mu}): C^L(z) \leq a(z) \}$ where if this set is empty we take $\hat{z}^L = {\beta_\mu} $. Similarly,   define $\hat{z}^R =  \sup \{ z \in ({ \alpha_\mu}, e^R) : C^R(z) \leq a(z) \}$ where if this set is empty we take $\hat{z}^R = { \alpha_\mu} $. Note that $\hat z^L\geq e^L$ and $\hat z^R\leq e^R$. Define $\hat{x}^L = (g^L)^{-1}(\hat{z}^L)$ (with $\hat{x}^L=\beta_\mu$ if { $\hat{z}^L = \beta_\mu=\beta_\nu$}) and $\hat{x}^R = (f^R)^{-1}(\hat{z}^R)$ (with $\hat{x}^R=\alpha_\mu$ if { $\hat{z}^R = \alpha_\mu=\alpha_\nu$}).

\begin{lem}
\label{lem:a=betc}

$a(e^L)=b(e^L)$ if and only if $\hat{z}^L = e^L$ and then $f^L(\hat{x}^L) = \hat{x}^L = \hat{z}^L = g^L(\hat{x}^L)=e^L$.

Similarly, $a(e^R)=b(e^R)$ if and only if $\hat{z}^R = e^R$ and then $f^R(\hat{x}^R) = \hat{z}^R = \hat{x}^R = g^R(\hat{x}^R)=e^R$.
\end{lem}

Note that if 
$e^L < \hat{z}^R < \hat{z}^L < e^R$, then by 
the intermediate value theorem applied to $C^R-C^L$ on $[\hat{z}^R, \hat{z}^L]$, there exists $z^* \in (\hat{z}^R, \hat{z}^L)\subseteq(e^L,e^R)$ such that $C^L(z^*)=C^R(z^*)>a(z^*)$. Further, if $\hat{z}^R=\hat{z}^L \in (e^L,e^R)$ then taking $z^* = \hat{z}^R$ we find $C^L(z^*)=C^R(z^*) = a(z^*)$.  In either case, if there exists $z^* \in (e^L,e^R)$ such that $C^L(z^*)=C^R(z^*)$ then we can define 
$x^*_L := (g^L)^{-1}(z^*)$, $x^*_R := (f^R)^{-1}(z^*)$
and
$\Psi:(\alpha_\nu,\beta_\nu) \mapsto \R$ by
\begin{equation}
  \label{eq:Bdef}
 \Psi(w)  =  
\begin{cases} 
     b(w), & w \leq f^L(x^*_L)  \\
    T^L_{x^*_L}(w), \quad \quad
         &  f^L(x^*_L) < w \leq z^*; \\
     T^{R}_{x^*_R}(w), 
        &  z^* < w < g^R(x^*_R);  \\
     b(w), & w \geq g^R(x^*_R).    
           \end{cases} 
\end{equation}
We will have that $\Psi$ is continuous provided there is continuity at $f^L(x^*_L)$, $z^*$ and $g^R(x^*_R)$. 
By the definition of $T^L_x = T^{f^L(x),x}_{b,a}$ we have that $T^L_{x^*_L}(f^L(x^*_L)) =
T^{f^L(x^*_L),x^*_L}_{b,a}(f^L(x^*_L))=b({f^L(x^*_L)})$ and continuity of $\Psi$ at $x^*_L$; a similar argument applies at $g^R(x^*_R)$. At $z^*$ continuity follows from the fact that $T^L_{x^*_L}(z^*) = C^L(z^*) = C^R(z^*) = T^R_{x^*_R}(z^*)$.
However, although $\Psi$ is continuous, it may or may not be convex.

The three cases will be as follows:
\begin{enumerate}
\item[(C1)] 
\begin{enumerate}
\item[(a)] 
$e^L <\hat{z}^R \leq \hat{z}^L < e^R$ (so that either there exists $z^* \in (\hat{z}^R, \hat{z}^L)$ such that $C^L(z^*)=C^R(z^*) > a(z^*)$, or $\hat{z}^R =\hat{z}^L$ and then taking $z^*=\hat{z}^R = \hat{z}^L$ we have $C^L(z^*)=C^R(z^*) =a(z^*)$); 
moreover 
$\Psi$ defined in \eqref{eq:Bdef} is convex;
\item[(b)] {either $\hat{z}^R = e^R = \hat{z}^L$ or $\hat{z}^R = e^L = \hat{z}^L$.}
\end{enumerate}
\item[(C2)] {$\hat{z}^L < \hat{z}^R$};
\item[(C3)] \begin{enumerate}
    \item $\hat{z}^L \geq e^R$ and/or $\hat{z}^R \leq e^L$; moreover $\hat z^R<\hat z^L$.
    \item 
$e^L < \hat{z}^R \leq \hat{z}^L < e^R$ (so that again $C^L(z^*)=C^R(z^*) \geq a(z^*)$ for some $z^* \in [\hat{z}^R, \hat{z}^L]$) but
$\Psi$ defined in \eqref{eq:Bdef} is not convex.
\end{enumerate}
\end{enumerate}

\begin{remark}
\label{rem:noteq}
{(i) Note that in case (C1)(b), if $\hat{z}^R=\hat{z}^L=e^L$ then by Lemma~\ref{lem:a=betc} we must have ${a(e^L)=b(e^L)}$. Then taking  
$z^*=\hat{z}^R=\hat{z}^L$ so that $x^*_L=f^L(x^*_L)=z^*$ in $\Psi$ defined in \eqref{eq:Bdef}, we find that the second line in \eqref{eq:Bdef} is redundant and $\Psi$ is automatically 
convex. A similar argument applies if 
$\hat{z}^R=\hat{z}^L=e^R$; again $\Psi$ is automatically convex in this case.}

(ii) If $e^L < \hat{z}^R = \hat{z}^L < e^R$ then  $z^* := \hat{z}^R$ is such that $C^L(z^*)=C^R(z^*)=a(z^*)$ and then $\Psi$ defined by \eqref{eq:Bdef} is convex. In particular, we are in Case (C1)(a) and not in Case (C3)(b). Therefore, in Case (C3)(b) we may (and will) take $e^L < \hat{z}^R < \hat{z}^L < e^R$.  
\end{remark}



It is clear that the cases (C1), (C2) and (C3) are mutually exclusive, except that in principle it is possible that there exists $z^*_1,z^*_2 \in (e^L,e^R)$ such that $C^L(z^*_i)= C^R(z^*_i) \geq a(z^*_i)$ for $i=1,2$  and such that for $z^*_1$, $\Psi$ defined in \eqref{eq:Bdef} is convex and  for $z^*_2$, $\Psi$ defined in \eqref{eq:Bdef} is not convex. 
{It can be shown that this cannot happen. However, since the three cases are merely starting points of a construction, it is not crucial to the argument that the construction must start in a unique way, and therefore we omit this argument.}

The goal, in each of the cases (C1), (C2) and (C3), is to choose $\psi$ such that the total cost of the superhedging portfolio generated by $\psi$ coincides with the model-based price of the Bermudan claim (for some model $\sM^* \in M(\mu,\nu)$ and the associated optimal stopping time $\tau^*$).

\begin{figure}
	\centering
\begin{tikzpicture}
\def\el{-18};
\def\er{-4};
\def\xl{-14};
\def\xr{-8};
\begin{axis}[
width=3.521in,
height=5.566in,
at={(0.758in,0.481in)},
scale only axis,
xmin=-4,
xmax=24,
ymin=-3,
ymax=54,
axis line style={draw=none},
ticks=none
]
\draw[thin,gray] (-2,0)--(22,0);
\draw[thin,gray] (-2,30)--(22,30);

\draw[blue,dotted, very thick] (2,20) to[out=280, in=180] (10,4) to[out=360, in=250] (19,17);
\draw[blue,dashed, thick] (0,13) to[out=300, in=180] (11,1) to[out=360, in=260] (20,16);
\draw[red] (1.3,10)--(13,6)--(19,10.5);
\node[circle,fill=red,inner sep=0pt,minimum size=5pt] at (1.3,10) {};
\node[circle,fill=red,inner sep=0pt,minimum size=5pt] at (13,6) {};
\node[circle,fill=red,inner sep=0pt,minimum size=5pt] at (19,10.5) {};
\node[circle,fill=red,inner sep=0pt,minimum size=5pt] at (15.8,8.1) {};
\node[circle,fill=red,inner sep=0pt,minimum size=5pt] at (4.3,8.9) {};

\draw[thin,gray,dotted] (-2,30)--(22,54);
\draw[thin,gray,loosely dashed] (1.3,21)--(1.3,-1);
\draw[thin,gray,loosely dashed] (19,21)--(19,-1);
\draw[thin,gray,loosely dashed] (15.8,21)--(15.8,-1);
\draw[thin,gray,loosely dashed] (4.3,21)--(4.3,-1);
\draw[thin,gray,loosely dashed] (13,21)--(13,-1);

\draw[thin,gray,loosely dashed] (-2,45)--(22,45);
\draw[thin,gray,loosely dashed] (4.3,26)--(4.3,54);
\draw[thin,gray,loosely dashed] (13,26)--(13,54);
\draw[thin,gray,loosely dashed] (15.8,26)--(15.8,54);
\draw[thin,gray,loosely dashed] (1.3,26)--(1.3,45);
\draw[thin,gray,loosely dashed] (19,26)--(19,51)--(4.3,51);
\draw[thin,gray,loosely dashed] (1.3,33.3)--(15.8,33.3);
\draw[thin,gray,loosely dashed] (2,34)--(2,29);
\node[below,scale=0.8] at (2,29) {$e^L$};
\draw[red, thick] (2,34) to[out=60, in=230] (4.3,45);
\draw[red, thick] (2,34) to[out=340, in=150] (4.3,33.3);
\draw[thin,gray,loosely dashed] (18,50)--(18,29);
\node[below,scale=0.8] at (18,29) {$e^R$};
\draw[thin,gray,loosely dashed] (-2,30)--(-2,28);
\node[below,scale=0.8] at (-2,28) {$\alpha_\mu=\alpha_\nu$};
\draw[thin,gray,loosely dashed] (22,30)--(22,28);
\node[below,scale=0.8] at (22,28) {$\beta_\mu=\beta_\nu$};
\draw[red, thick] (18,50) to[out=140, in=300] (15.8,51);
\draw[red, thick] (18,50) to[out=250, in=50] (15.8,45);
\draw[red, thick] (-2,30) -- (2,34);
\draw[red, thick] (18,50) -- (22,54);

\node[below,scale=0.8] at (1.3,26) {$f^L(x^*_L)$};
\node[below,scale=0.8] at (4.3,26) {$x^*_L$};
\node[below,scale=0.8] at (13,26) {$z^*$};
\node[right,scale=0.8] at (22,45) {$z^*$};
\node[below,scale=0.8] at (15.8,26) {$x^*_R$};
\node[below,scale=0.8] at (19,26) {$f^R(x^*_R)$};

\node[below,scale=0.8] at (1.3,-1) {$f^L(x^*_L)$};
\node[below,scale=0.8] at (4.3,-1) {$x^*_L$};
\node[below,scale=0.8] at (13,-1) {$z^*$};
\node[below,scale=0.8] at (15.8,-1) {$x^*_R$};
\node[below,scale=0.8] at (19,-1) {$f^R(x^*_R)$};

\filldraw[ opacity=0.1, red] (1.3,33.3)--(4.3,33.3) -- (4.3,45)--(1.3,45);
\filldraw[ opacity=0.1, red] (19,51)--(15.8,51) -- (15.8,45)--(19,45);
\filldraw[ opacity=0.1, blue] (4.3,54)--(15.8,54) -- (15.8,51)--(4.3,51);
\filldraw[ opacity=0.1, blue] (4.3,33.3)--(15.8,33.3) -- (15.8,30)--(4.3,30);
\node[blue,scale=0.8] at (17.5,14) {$a$};
\node[blue,scale=0.8] at (17.5,5) {$b$};
\node[red,scale=0.8] at (8,9) {$T^L_{x^*_L}$};
\node[red,scale=0.8] at (14.5,9) {$T^R_{x^*_R}$};
\end{axis}
\end{tikzpicture}
\caption{The Case (C1)(a) with $z^*\in(\hat z^R,\hat z^L)$. The top part of the figure represents the stylized plots of functions $f^L$ and $g^L$ (on $[e^L,x^*_L=(g^L)^{-1}(z^*)]$), and  $f^R$ and $g^R$ (on $[e^R,x^*_R=(f^R)^{-1}(z^*)]$), that support the left and right-curtain martingale couplings, respectively. Note that $g^L$ (resp., $f^R$) is {in}creasing, while $f^L$ (resp., $g^R$) is {de}creasing on $[e^L,x^*_L]$ (resp.,  $[e^R,x^*_R]$). Furthermore, the shaded areas correspond to the sets (and associated exercise rules) on which the optimal models $\sM^*$ concentrate: the Bermudan option is exercised at time-1 if $Z_1\notin(x^*_L,x^*_R)$, and the mass in $(\alpha_\mu,f^L(x^*_L))\cup(g^R(x^*_R),\beta_\mu)$ stays put (i.e., remains on the diagonal), while the mass in $(f^L(x^*_L),x^*_L)$ (resp., $(x^*_R,g^R(x^*_R))$) is mapped to $(f^L(x^*_L),g^L(x^*_L))$ (resp., $(f^R(x^*_R),g^R(x^*_R))$). On the other hand, if $Z_1\in(x^*_L,x^*_R)$, then the option will be exercised at time-2 and the mass in $(x^*_L,x^*_R)$ is mapped to the tails $(\alpha_\nu,f^L(x^*_L))\cup(g^R(x^*_R),\beta_\nu)$. The bottom part of the figure depicts the payoff functions $a$ and $b$ (with $a>b$), and the candidate convex function $\Psi=\Psi^{y^*_-,x^*_-,z^*, x^*_+,y^*_+}_{5,a,b}$. In particular, we have that $\Psi=T_{b,a}^{f^L(x^*_L),x^*_L}=T^L_{x^*_L}$  on $[f^L(x^*_L),z^*]$ and $\Psi=T_{a,b}^{x^*_L,g^R(x^*_R)}=T^R_{x^*_R}$ on $[z^*,g^R(x^*_R)]$, while $\Psi=b$ on $(\alpha_\nu,f^L(x^*_L))\cup(g^R(x^*_R),\beta_\nu)$.}
\label{fig:case1}
\end{figure}

\subsection{Proof of Theorem \ref{thm2} in Case (C1)}\label{sec:C1}
In this case we take $(y^*_-,x^*_-,z^*, x^*_+,y^*_+) = (f^L(x_L^*),x^*_L,z^*,x^*_R,g^R(x^*_R))$. Note that $g^L(x^*_L)=z^*=f^R(x^*_R)$. Then $\Psi^{y^*_-,x^*_-,z^*, x^*_+,y^*_+}_{5,a,b}$ (defined in \eqref{eq:psi5points}) and $\Psi$ (defined in \eqref{eq:Bdef}) are such that $\Psi^{y^*_-,x^*_-,z^*, x^*_+,y^*_+}_{5,a,b} = \Psi$ (in Case (C1)(b), if $\hat z^R=e^R=\hat z^L$ (resp., $\hat z^R=e^L=\hat z^L$), then $y^*_+=x^*_+=z^*=e^R$ (resp., $y^*_-=x^*_-=z^*=e^L$) and thus the third (resp., second) lines in the definitions of both  $\Psi^{y^*_-,x^*_-,z^*, x^*_+,y^*_+}_{5,a,b}$ and $\Psi$ are redundant). By construction, $\Psi$ is convex if we are in Case (C1)(a). On the other hand, in Case (C1)(b), by Lemma \ref{lem:a=betc} we have that $a(\hat z^R)=b(\hat z^R)$, and thus $\Psi$ is again convex. It follows that $(y^*_-,x^*_-,z^*, x^*_+,y^*_+) \in \sL_5^{a,b}$ in Case (C1). See Figure \ref{fig:case1}.

On the other hand, since $e^L\leq \hat z^R\leq z^*\leq\hat z^L\leq e^R$, we have that $e^L\leq x^*_-<x^*_+\leq e^R$ and $x^*_+=B(x^*_-)$ (where $B(\cdot)$ is given by \eqref{eq:BoundaryB}). Then by Lemma \ref{lem:new:FormOfShadow}(i) we have that $S^\nu(\bar{\mu}_{x^*_-}^{x^*_+}) = S^\nu(\mu_{\alpha_\mu}^{x^*_-})+S^{\nu}(\mu_{x^*_+}^{\beta_\mu})$. By Corollary \ref{cor:LCandRC}(i), $S^\nu(\mu_{\alpha_\mu}^{x^*_-})=\mu_{\alpha_\mu}^{y^*_-} + \nu_{y^*_-}^{z^*}\leq\nu$, and therefore $\mu_{\alpha_\mu}^{y^*_-}\leq\nu_{\alpha_\mu}^{y^*_-}$ (so that $S^\nu(\mu_{\alpha_\mu}^{y^*_-})=\mu_{\alpha_\mu}^{y^*_-}$). By the definition and the associativity property of the shadow measure (see Proposition \ref{prop:shadow_assoc}) it then follows that $S^\nu(\mu_{\alpha_\mu}^{y^*_-})+S^{\nu-S^\nu(\mu_{\alpha_\mu}^{y^*_-})}(\mu_{y^*_-}^{x^*_-})=S^\nu(\mu_{\alpha_\mu}^{x^*_-})=\mu_{\alpha_\mu}^{y^*_-} + \nu_{y^*_-}^{z^*}$ and thus $S^{\nu-S^\nu(\mu_{\alpha_\mu}^{y^*_-})}(\mu_{y^*_-}^{x^*_-})=\nu_{y^*_-}^{z^*}$, from which we conclude that $\mu_{y^*_-}^{x^*_-}\leq_{cx}\nu_{y^*_-}^{z^*}$. By applying similar arguments to $(z^*,x^*_+,y^*_+)$ (together with Corollary \ref{cor:LCandRC}(ii)) we have that $\mu_{y^*_+}^{\beta_\mu}\leq\nu_{y^*_+}^{\beta_\mu}$ and $\mu_{x^*_+}^{y^*_+}\leq_{cx}\nu_{z^*}^{y^*_+}$. Combining both cases shows that $(y^*_-,x^*_-,z^*, x^*_+,y^*_+)\in \Gamma_5^{\mu,\nu}$.

We conclude that $(f^L(x_L^*),x^*_L,z^*,x^*_R,g^R(x^*_R)) \in \Gamma_5^{\mu,\nu} \cap \sL_5^{a,b}$, which finishes the proof in Case (C1).

\subsection{Proof of Theorem \ref{thm2} in Case (C2)}
\label{sec:C2}
In this case we take $(\hat y_-,\hat x_-,\hat z_-,\hat z_+,\hat x_+,\hat y_+)=(f^L(\hat x^L),\hat x^L,\hat z^L,\hat z^R,\hat x^R,g^R(\hat x^R))$. Note that $e^L\leq \hat z^L<\hat z^R\leq e^R$.

Suppose first that $e^L< \hat z^L<\hat z^R<e^R$, so that (by Lemma \ref{lem:a=betc}) $b(e^L)<a(e^L)$ and $b(e^R)<a(e^R)$.
In that case we define a candidate convex function $\psi^{*,2}$ such that $\psi^{*,2}\geq b$ everywhere (so that, as before, it generates a superhedge with total cost $\int(a- \psi^{*,2})^+d\mu+\int\psi^{*,2}d\nu$)):
\begin{equation}\label{eq:psi1}
\psi^{*,2}(x)=\begin{cases}
b(x),\quad x\in(-\infty,\hat y_-]\cup[\hat y_+,\infty)\\
a(x),\quad x\in (\hat{z}_-,\hat{z}_+)\\
T^L_{\hat{x}_L}(x),\quad x\in(\hat y_-,\hat{z}_-]\\
T^R_{\hat{x}_R}(x),\quad x\in[\hat{z}_+,\hat y_+).
\end{cases}
\end{equation}
Then $\Psi^{\hat y_-,\hat x_-,\hat z_-,\hat z_+,\hat x_+,\hat y_+}_{6,a,b}$ (defined in \eqref{eq:psi6points}) is such that $\Psi^{\hat y_-,\hat x_-,\hat z_-,\hat z_+,\hat x_+,\hat y_+}_{6,a,b} = \psi^{*,2}$ (note that, due to the definitions of $\hat z_-=\hat z^L$ and $\hat z_+=\hat z^R$, and the fact that we are in Case (C2), $\psi^{*,2}$ is convex). It follows that $(\hat y_-,\hat x_-,\hat z_-,\hat z_+,\hat x_+,\hat y_+) \in \sL_6^{a,b}$.

Since $e^L<\hat x_-=\hat x^L<g^L(\hat x^L)=\hat z^L=\hat z_- < \hat z_+=\hat z^R=f^R(\hat x^R)<\hat x^R=\hat x_+<e^R$, and similarly as in Case (1), by Lemma \ref{lem:new:FormOfShadow}(i) we have that $\mu_{\alpha_\mu}^{\hat y_-}+\mu_{\hat y_+}^{\beta_\mu}\leq\nu_{\alpha_\mu}^{\hat y_-}+\nu_{\hat y_+}^{\beta_\mu}$, $\mu_{\hat y_-}^{\hat x_-}\leq_{cx}\nu_{\hat y_-}^{\hat z_-}$ and $\mu_{\hat x_+}^{\hat y_+}\leq_{cx}\nu_{\hat z_+}^{\hat y_+}$. On the other hand, by the Dispersion Assumption (recall Definition \ref{def:dispersion}), $\nu\leq\mu$ on $(\hat z_-,\hat z_+)$. It follows that
$(\hat y_-,\hat x_-,\hat z_-,\hat z_+,\hat x_+,\hat y_+)\in \Gamma_6^{\mu,\nu}$, which completes the proof in this case. See Figure \ref{fig:case2}.

If $b(e^L)=a(e^L)$, or $b(e^R)=a(e^R)$, or both, then the above construction simplifies. Indeed, by Lemma~\ref{lem:a=betc}, if $a(e^L)=b(e^L)$, then $\hat y_- =\hat x_- =\hat z_-$, and therefore $\nu_{\hat y_-}^{\hat z_-}=\nu_{\hat z_-}^{\hat z_-}$ is the zero measure, and the third line in the definition of $\psi^{*,2}$ (see \eqref{eq:psi1}) is redundant, so that $\psi^{*,2}=b$ on $(-\infty,e^L=\hat z_-]$. Similarly, if $a(e^R)=b(e^R)$, then $\hat z_+=\hat x_+=\hat y_+$, and therefore $\nu_{\hat z_+}^{\hat y_+}=\nu_{\hat z_+}^{\hat z_+}$ is the zero measure, and the fourth line in \eqref{eq:psi1} is redundant, so that $\psi^{*,2}=b$ on $[e^R=\hat z_+,\infty)$. Finally, if both $a(e^L)=b(e^L)$ and $a(e^L)=b(e^L)$, then $(\nu_{\hat y_-}^{\hat z_-}+\nu_{\hat z_+}^{\hat y_+})$ is the zero measure, and both the third and fourth lines in \eqref{eq:psi1} are redundant, so that $\psi^{*,2}$ simplifies to $\psi^{*,2}=a$ on $(e^L,e^R)$ and $\psi^{*,2}=b$ otherwise. Note that in any of these cases we still have that $\Psi^{\hat y_-,\hat x_-,\hat z_-,\hat z_+,\hat x_+,\hat y_+}_{6,a,b} = \psi^{*,2}$ is convex, and $(\hat y_-,\hat x_-,\hat z_-,\hat z_+,\hat x_+,\hat y_+) \in \sL_6^{a,b}\cap \Gamma_6^{\mu,\nu}$, which completes the proof.

\subsection{Proof of Theorem \ref{thm2} in Case (C3)}
\label{sec:C3}
Recall the definition of $\sL^{a,b}_4$ in \eqref{eq:L4def} and note that, at least if $x>\alpha_\mu$ and $y<\beta_\mu$, $(w,x,y,z)\in \sL^{a,b}_4$ if $(w,b(w))$, $(x,a(x))$, $(y,a(y))$ and $(z,b(z))$ lie on a straight line. Recall also the definition of $\Gamma^{\mu,\nu}_4$ in \eqref{eq:Gamma4def}. In this section we write $\sL_4$ as shorthand for $\sL^{a,b}_4$ and $\Gamma_4$ as shorthand for $\Gamma_4^{\mu,\nu}$. 

Our goal is to find $(w,x,y,z) \in \sL_4 \cap \Gamma_4$. 
Then we can build a model including a coupling $\pi$ (of $\mu$ and $\nu$) such that the mass outside $(x,y)$ is mapped to its shadow, and the associated optimal strategy is to take $\tau^*=1$ on $X_1 \notin (x,y)$. The corresponding optimal hedge is generated by $\psi = b \vee L$ where $L$ is the straight line that passes through the four collinear points  $(w,b(w))$, $(x,a(x))$, $(y,a(y))$ and $(z,b(z))$. See {Case 1} of Theorem \ref{thm1}.

We will find $(w,x,y,z) \in \sL_4 \cap \Gamma_4$ in four steps. Step 1: first we will construct maps $w,z:(\alpha_\mu,\beta_\mu)\to(\alpha_\nu,\beta_\nu)$ such that $(w(x,y),x,y,z(x,y)) \in \Gamma_4$; see Lemma \ref{lem:PropertiesofsA}. Step 2: we will use the functions $C^L$ and $C^R$ above to find a starting point $(x_0,y_0)$ and then a modified starting point $(x_1,y_0)$. Step 3: we show that there exists a continuous function $y^*:[x_1,\beta_\mu)\to(\alpha_\mu,\beta_\mu)$ such that $(w(x,y^*(x)),x,y^*(x),z(x,y^*(x))) \in \Gamma_4$ and the points $(x,a(x))$, $(y^*(x),a(y^*(x)))$ and $(z(x,y^*(x)),b(z(x,y^*(x)))$ lie on the same line; see Corollary \ref{cor:existenceContinuity}. Step 4: finally, by adjusting $x$, we will show that there exists $x^*\in[x_1,\alpha_\mu)$, for which the remaining point $(w(x^*,y^*(x^*)),b(w(x^*,y^*(x^*))))$ also lies on the same line and thus $(w(x^*,y^*(x^*)),x^*,y^*(x^*),z(x^*,y^*(x^*))) \in \sL_4\cap \Gamma_4$; see Proposition \ref{prop:solution}.
\begin{figure}[H]
	\centering
\begin{tikzpicture}
\def\el{-18};
\def\er{-4};
\def\xl{-14};
\def\xr{-8};
\begin{axis}[
width=3.521in,
height=4.766in,
at={(0.758in,0.481in)},
scale only axis,
xmin=-4,
xmax=24,
ymin=-3,
ymax=54,
axis line style={draw=none},
ticks=none
]
\draw[thin,gray] (-2,0)--(22,0);
\draw[thin,gray] (-2,30)--(22,30);

\draw[blue,dotted, very thick] (2,20) to[out=280, in=180] (10,4) to[out=360, in=250] (19,17);
\draw[blue,dashed, thick] (0,13) to[out=300, in=180] (11,1) to[out=360, in=260] (20,16);
\draw[red] (1.3,10)--(10,4);
\node[circle,fill=red,inner sep=0pt,minimum size=5pt] at (1.3,10) {};
\node[circle,fill=red,inner sep=0pt,minimum size=5pt] at (10,4) {};
\node[circle,fill=red,inner sep=0pt,minimum size=5pt] at (5.2,7.3) {};
\draw[red] (12,4.4)--(19.5,12.5);
\node[circle,fill=red,inner sep=0pt,minimum size=5pt] at (12,4.4) {};
\node[circle,fill=red,inner sep=0pt,minimum size=5pt] at (19.5,12.5) {};
\node[circle,fill=red,inner sep=0pt,minimum size=5pt] at (16.3,9) {};

\draw[thin,gray,dotted] (-2,30)--(22,54);
\draw[thin,gray,loosely dashed] (1.3,21)--(1.3,-1);
\draw[thin,gray,loosely dashed] (10,21)--(10,-1);
\draw[thin,gray,loosely dashed] (5.2,21)--(5.2,-1);
\draw[thin,gray,loosely dashed] (12,21)--(12,-1);
\draw[thin,gray,loosely dashed] (16.3,21)--(16.3,-1);
\draw[thin,gray,loosely dashed] (19.5,21)--(19.5,-1);

\draw[thin,gray,loosely dashed] (1.3,33.3)--(5.2,33.3)--(5.2,42)--(1.3,42)--(1.3,26);
\draw[thin,gray,loosely dashed] (5.2,33.3)--(5.2,26);
\draw[thin,gray,loosely dashed] (5.2,42)--(10,42)--(10,26);
\draw[thin,gray,loosely dashed] (5.2,42)--(5.2,54);
\draw[thin,gray,loosely dashed] (10,42)--(10,54);
\draw[thin,gray,loosely dashed] (12,54)--(12,26);
\draw[thin,gray,loosely dashed] (19.5,51.5)--(16.3,51.5)--(16.3,44.4)--(19.5,44.4)--(19.5,51,5);
\draw[thin,gray,loosely dashed] (16.3,44.4)--(12,44.4);
\draw[thin,gray,loosely dashed] (16.3,51.5)--(16.3,54);
\draw[thin,gray,loosely dashed] (19.5,51,5)--(19.5,54);
\draw[thin,gray,loosely dashed] (19.5,44.4)--(19.5,26);
\draw[thin,gray,loosely dashed] (16.3,44.4)--(16.3,26);

\draw[thin,gray,loosely dashed] (2,34)--(2,29);
\node[below,scale=0.8] at (2,29) {$e^L$};
\draw[red, thick] (2,34) to[out=60, in=230] (5.2,42);
\draw[red, thick] (2,34) to[out=340, in=150] (5.2,33.3);
\draw[thin,gray,loosely dashed] (18.5,50.5)--(18.5,29);
\node[below,scale=0.8] at (18.5,29) {$e^R$};
\draw[red, thick] (18.5,50.5) to[out=140, in=300] (16.3,51.5);
\draw[red, thick] (18.5,50.5) to[out=250, in=50] (16.3,44.4);
\draw[red, thick] (10,42) -- (12,44.4);
\draw[red, thick] (18.5,50.5) -- (22,54);
\draw[red, thick] (-2,30) -- (2,34);

\node[below,scale=0.8] at (1.3,26) {$f^L(\hat x_L)$};
\node[below,scale=0.8] at (5.2,26) {$\hat x_L$};
\node[below,scale=0.8] at (10,26) {$\hat z_L$};
\node[below,scale=0.8] at (12,26) {$\hat z_R$};
\node[below,scale=0.8] at (16.3,26) {$\hat x_R$};
\node[below,scale=0.8] at (19.5,26) {$g^R(\hat x_R)$};

\node[below,scale=0.8] at (1.3,-1) {$f^L(\hat x_L)$};
\node[below,scale=0.8] at (5.2,-1) {$\hat x_L$};
\node[below,scale=0.8] at (10,-1) {$\hat z_L$};
\node[below,scale=0.8] at (12,-1) {$\hat z_R$};
\node[below,scale=0.8] at (16.3,-1) {$\hat x_R$};
\node[below,scale=0.8] at (19.5,-1) {$g^R(\hat x_R)$};
\draw[thin,gray,loosely dashed] (1.3,42) -- (-2,42);
\node[left,scale=0.8] at (-2,42) {$\hat z_L$};
\filldraw[ opacity=0.1, red] (1.3,33.3)--(5.2,33.3) -- (5.2,42)--(1.3,42);
\draw[thin,gray,loosely dashed] (19.5,44) -- (22,44);
\node[right,scale=0.8] at (22,44) {$\hat z_R$};
\filldraw[ opacity=0.1, red] (19.5,51.5)--(16.3,51.5)--(16.3,44.4)--(19.5,44.4);
\filldraw[ opacity=0.3, blue] (5.2,33.3)--(10,33.3) -- (10,30)--(5.2,30);
\filldraw[ opacity=0.1, blue] (10,33.3)--(12,33.3) -- (12,30)--(10,30);
\filldraw[ opacity=0.3, blue] (12,33.3)--(16.3,33.3) -- (16.3,30)--(12,30);
\filldraw[ opacity=0.3, blue] (5.2,51.5)--(5.2,54) -- (10,54)--(10,51.5);
\filldraw[ opacity=0.1, blue] (10,51.5)--(12,51.5) -- (12,54)--(10,54);
\filldraw[ opacity=0.3, blue] (12,51.5)--(16.3,51.5) -- (16.3,54)--(12,54);
\draw[thin,gray,loosely dashed] (-2,30)--(-2,28);
\node[below,scale=0.8] at (-2,28) {$\alpha_\mu=\alpha_\nu$};
\draw[thin,gray,loosely dashed] (22,30)--(22,28);
\node[below,scale=0.8] at (22,28) {$\beta_\mu=\beta_\nu$};
\node[red,scale=0.8] at (8,7) {$T^L_{\hat x_L}$};
\node[red,scale=0.8] at (14,8.5) {$T^R_{\hat x_R}$};
\node[blue,scale=0.8] at (17.5,14.5) {$a$};
\node[blue,scale=0.8] at (17.7,5) {$b$};
\end{axis}
\end{tikzpicture}
\caption{The Case (C2) with $e^L<\hat z^L<\hat z^R<e^R$. The top part of the figure represents the stylized plots of functions $f^L$ and $g^L$ (on $(\alpha_\mu, \hat x^L=(g^L)^{-1}(\hat z^L))$), and  $f^R$ and $g^R$ (on $(\hat x^R=(f^R)^{-1}(\hat z^R),\beta_\mu)$), that support the left and right-curtain martingale couplings, respectively. Furthermore, the shaded areas in the top part of the figure represent the sets (and associated exercise rules) on which the optimal models $\sM^*$ concentrate: the Bermudan option is exercised at time-1 if $Z_1\notin(-\hat x^L,\hat x^R)$ (and then the mass in $(\alpha_\mu,f^L(
\hat x^L))\cup(g^R(\hat x^R),\beta_\mu)$ stays put, {while the mass in 
$(f^L(\hat x^L),\hat x^L)$ is mapped to $(f^L(\hat x^L),\hat z^L)$ and the mass in
$(\hat x^R,g^R(\hat x^R))$ is mapped to $(\hat z^R,g^R(\hat x^R))$}),
and if $Z_1\in(\hat z^L,\hat z^R)$ and $U\leq(\eta(Z_1)/\rho(Z_1))$ (note that, due to the Dispersion Assumption \ref{ass:simple}, $\eta>\rho$ on $(\hat z^L,\hat z^R)\subset (e^L,e^R)$, and thus only a portion of the mass in $(\hat z^L,\hat z^R)$ stays put). On the other hand, the option will be exercised at time-2 if either $Z_1\in(\hat x^L,\hat z^L)\cup(\hat z^R,\hat x^R)$ (and then the mass in $(\hat x^L,\hat z^L)\cup(\hat z^R,\hat x^R)$ is mapped to the tails $(\alpha_\nu,f^L(\hat x^L))\cup(g^R(\hat x^R),\beta_\nu)$), or $Z_1\in(\hat z^L,\hat z^R)$ and $U>(\eta(Z_1)/\rho(Z_1))$ (and then this portion of mass in $(\hat z^L,\hat z^R)$ is (again) mapped to the tails $(\alpha_\nu,f^L(\hat x^L))\cup(g^R(\hat x^R),\beta_\nu)$).
In the bottom part of the figure we identify the candidate convex function $\psi^{*,2}$. In particular, $\psi^{*,2}=b$ on $(\alpha_\nu,f^L(\hat x^L))\cup(g^R(\hat x^R),\beta_\nu)$, $\psi^{*,2}=T_{a,b}^{f^L(\hat x^L),\hat x^L}=T^L_{\hat x^L}$ on $[f^L(\hat x^L),\hat z^L=g^R(\hat z^L)]$, $\psi^{*,2}=T_{a,b}^{\hat x^R,g^R(\hat x^R)}=T^R_{\hat x^R}$ on $[\hat z^R=f^R(\hat x^R),g^R(\hat z^R)]$ and $\psi^{*,2}=a$ on $(\hat z^L,\hat z^R)$. }
\label{fig:case2}
\end{figure}

{\textit{{Step 1: constructing maps $w,z$.}} Recall the definition of $B:(\alpha_\mu,\beta_\mu)\to(\alpha_\mu,\beta_\mu)$ given in \eqref{eq:BoundaryB}. The following lemma is an immediate consequence of Lemma \ref{lem:couplingsDispersionAssumption}.}
\begin{lem}
    $B$, as defined in \eqref{eq:BoundaryB}, is strictly increasing on $(\alpha_\mu,(g^L)^{-1}(\beta_\mu))$, continuous on $(\alpha_\mu,\beta_\mu)$, and $B(x)>x$ for all $x\in(\alpha_\mu,\beta_\mu)$. Furthermore, $\lim_{x\downarrow\alpha_\mu}B(x)=\lim_{x\downarrow\alpha_\mu}(f^R)^{-1}(x)=(f^R)^{-1}(\alpha_\mu)$ and $\lim_{x\uparrow\beta_\mu}B(x)=\lim_{x\uparrow \beta_\mu}(f^R)^{-1}(x)=\beta_\mu$.
\end{lem}

Note that the restriction of $B(\cdot)$ to $(\alpha_\mu,(g^L)^{-1}(\beta_\mu))$ admits an inverse $B^{-1}:((f^R)^{-1}(\alpha_\mu),\beta_\mu)\to(\alpha_\mu,(g^L)^{-1}(\beta_\mu))$. In the case $\alpha_\nu<\alpha_\mu$, we extend the domain of $B^{-1}(\cdot)$ to $(\alpha_\mu,\beta_\mu)$ by setting $B^{-1}(y)=\alpha_\mu$ for all $y\in(\alpha_\mu,(f^R)^{-1}(\alpha_\mu)]$. Then $B^{-1}(\cdot)$ is continuous (and non-decreasing) on $(\alpha_\mu,\beta_\mu)$ and strictly increasing on $((f^R)^{-1}(\alpha_\mu),\beta_\mu)$.

Define $\sA:=\{(x,y): x\in[\alpha_\mu, \beta_\mu)\cap(\alpha_\nu,\beta_\mu),~y\in(\alpha_\mu,\beta_\mu]\cap(\alpha_\mu,\beta_\nu),~x\leq y\leq B(x)\}$. Also, for $x\in[\alpha_\mu, \beta_\mu)\cap(\alpha_\nu,\beta_\mu)$ and $y\in(\alpha_\mu,\beta_\mu]\cap(\alpha_\mu,\beta_\nu)$, set $\sA_x=\{z:(x,z)\in\sA\}$ and $\sA_y:=\{z:(z,y)\in\sA\}$, respectively. Set $\sA^<=\{(x,y)\in\sA:x<y\}$, $\sA^<_x=\{z\in\sA_x:x<z\}$ and $\sA^<_y=\{z\in\sA_y:z<y\}$.


\begin{figure}
	\centering
\begin{tikzpicture}
\def\el{-18};
\def\er{-4};
\def\xl{-14};
\def\xr{-8};
\begin{axis}[
width=3.521in,
height=5.566in,
at={(0.758in,0.481in)},
scale only axis,
xmin=-4,
xmax=24,
ymin=-11,
ymax=54,
axis line style={draw=none},
ticks=none
]
\draw[thin,gray] (-2,0)--(22,0);
\draw[thin,gray] (-2,30)--(22,30);

\draw[blue,dotted, very thick] (0,20) to[out=280, in=180] (10,7) to[out=360, in=250] (20,22);
\draw[blue,dashed, thick] (0,8) to[out=290, in=180] (11,1) to[out=360, in=260] (20,16);

\draw[red,thick] (2.5,3.2)--(7,7.5)--(19.5,12.5);
\node[circle,fill=red,inner sep=0pt,minimum size=5pt] at (2.5,3.2) {};
\node[circle,fill=red,inner sep=0pt,minimum size=5pt] at (7,7.5) {};
\node[circle,fill=red,inner sep=0pt,minimum size=5pt] at (16.2,11.2) {};
\node[circle,fill=red,inner sep=0pt,minimum size=5pt] at (19.5,12.5) {};

\draw[red, dash dot] (4.1,2.2)--(5.8,8.3)--(16.2,11.2)--(17.5,6);
\node[circle,fill=red,inner sep=0pt,minimum size=5pt] at (4.1,2.2) {};
\node[circle,fill=red,inner sep=0pt,minimum size=5pt] at (5.8,8.3) {};
\node[circle,fill=red,inner sep=0pt,minimum size=5pt] at (17.5,6) {};

\draw[thick,dotted] (-2,30)--(22,54);
\node[circle,fill=blue,inner sep=0pt,minimum size=5pt] at (5.8,48.2) {};
\node[circle,fill=red,inner sep=0pt,minimum size=5pt] at (7,48.2) {};
\draw[blue] (-2,30) to[out=60, in=220] (5.8,48.2) to[out=40, in=190] (13,54);
\draw[thin,gray,loosely dashed] (13,54)--(13,28);
\node[scale=0.8,below] at (13,28) {$(g^L)^{-1}(\beta_\mu)$};
\draw[blue,dashed] (13,54)--(22,54);
\draw[red] (7,48.2) to[out=10, in=220] (10,44) to[out=40, in=190] (16,54)--(18,54)to[out=360, in=160] (20,52);
\node[red,scale=0.8] at (15.4,50.5) {$x\mapsto y^*(x)$};
\node[blue,scale=0.8] at (5.4,51.5) {$x\mapsto B(x)$};
\node[scale=0.8] at (10.4,40) {$x\mapsto x$};

\draw[thin,gray,loosely dashed] (5.8,48.2)--(5.8,28);
\draw[thin,gray,loosely dashed] (7,48.2)--(7,28);
\draw[thin,gray,loosely dashed] (20,52)--(20,28);
\draw[thin,gray,loosely dashed] (22,52)--(22,28);
\draw[thin,gray,loosely dashed] (-2,30)--(-2,28);
\draw[thin,gray,loosely dashed] (5.8,48)--(16,48)--(16,25);
\node[below,scale=0.8] at (5.8,28) {$x_0$};
\draw[thin, gray, loosely dashed] (5.8,26) -- (5.8,-1);
\node[below,scale=0.8] at (5.8,-1) {$x_0$};
\node[below,scale=0.8] at (7,28) {$x_1$};
\node[below,scale=0.8] at (15,25) {$y_0=B(x_0)=y^*(x_1)$};
\draw[thin,gray,loosely dashed] (16.2,23)--(16.2,-1);
\node[below,scale=0.8] at (16.2,-1) {$y_0$};
\node[below,scale=0.8] at (20,28) {$\tilde x$};
\node[below,scale=0.8] at (-2,28) {$\alpha_\mu=\alpha_\nu$};
\node[below,scale=0.8] at (22,28) {$\beta_\mu$};
\draw[thin,gray,loosely dashed] (16,48)--(22,48);
\node[right,scale=0.8] at (22,48) {$y_0$};
\draw[thin,gray,loosely dashed] (4.1,2.2) -- (4.1,-3.5);
\node[below,scale=0.8] at (5.1,-3.5) {$w(x_0,y_0)$};
\draw[thin,gray,loosely dashed] (17.5,6) -- (17.5,-3.5);
\node[below,scale=0.8] at (17,-3.5) {$z(x_0,y_0)$};
\draw[thin,gray,loosely dashed] (7,26) -- (7,-1);
\node[below,scale=0.8] at (7,-1) {$x_1$};
\draw[thin,gray,loosely dashed] (2.5,3.2)--(2.5,-7.5);
\node[below,scale=0.8] at (2.5,-7.5) {$w(x_1,y_0)$};
\draw[thin,gray,loosely dashed] (19.5,12.5)--(19.5,-7.5);
\node[below,scale=0.8] at (19.5,-7.5) {$z(x_1,y_0)$};

\node[red,scale=0.8] at (13,12) {$\Psi_{x_0,y_0}$};
\node[red,scale=0.8] at (13,5) {$\Psi_{x_1,y_0}$};
\draw [thin, gray, ->] (12,4) to [out=190,in=300] (9,8);
\draw [thin, gray, ->] (12,11) to [out=190,in=60] (9,9.5);

\node[blue,scale=0.8] at (2.5,14) {$a$};
\node[blue,scale=0.8] at (1.3,6.5) {$b$};

\end{axis}
\end{tikzpicture}
\caption{The Case (C3)(b). The top part of the figure represents the stylized plots of functions $x\mapsto B(x)$, $x\mapsto y^*(x)$, and the locations of points $(x_0,y_0)$ and $(x_1.y_0)$, in the case $\alpha_\mu=\alpha_\mu<(g^L)^{-1}(\beta_\mu)<\tilde x<\beta_\mu<\beta_\nu$. Note that $B$ is non-decreasing, $(x_0,y_0)=(x_0,B(x_0))$ (so that $(x_0,y_0)$ lies on the boundary of $\sA^<$) and $x_0<x_1<y_0$ (so that $(x_1,y_0)$ lies in the interior of $\sA^<$). On the other hand, $y^*$ may not be monotone, $\lim_{x\to\tilde x}y^*(x)=\tilde x$ and we could have that $y^*\equiv\beta_\mu$ on (some parts of) $((g^L)^{-1}(\beta_\mu),\tilde x)$. In the bottom part of the figure we identify the function $\Psi_{x,y}$ (see \eqref{eq:PsiXY}) for $(x,y)=(x_0,y_0)$ and $(x,y)=(x_1,y_0)$. In particular, $\Psi_{x_0,y_0}$ (see the dash-dotted curve) and $\Psi_{x_1,y_0}$ (see the solid curve) are both piece-wise linear and concave on $(w(x_0,y_0),z(x_0,y_0))$ and $(w(x_1,y_0),z(x_1,y_0))$, respectively. Observe that $\Psi_{x_0,y_0}$ is formed by three distinct lines on $(w(x_0,y_0),z(x_0,y_0))$, while $\Psi_{x_1,y_0}$ corresponds to a single line on $(x_1,z(x_1,y_0))$ (since the points $(x_1,a(x_1))$, $(y_0,a(y_0))$ and $(z(x_1,y_0),b(z(x_1,y_0)))$ are collinear).}
\label{fig:case3}
\end{figure}

\begin{lem}
    \label{lem:PropertiesofsA}
Suppose that $\mu\leq_{cx}\nu$ satisfy the Dispersion Assumption (Definition \ref{def:dispersion}).

\begin{enumerate}
    \item[(i)] There exists $w,z:\sA^<\to(\alpha_\nu,\beta_\nu)$ such that, for all $(x,y)\in\sA^<$, $w(x,y)\leq x\wedge e^L<y\vee e^R\leq z(x,y)$ and $(w(x,y), x,y, z(x,y))\in \Gamma_4$.
    
    \item[(ii)] Fix $(x,y)\in\sA$. Then $w(\cdot,y)$ (resp., $w(x,\cdot)$) is continuous and strictly decreasing (resp., increasing) on $\sA^<_y$ (resp., $\sA^<_x$). Similarly,  $z(\cdot,y)$ (resp., $z(x,\cdot)$) is continuous and strictly increasing (resp., decreasing) on $\sA^<_y$ (resp., $\sA^<_x$).  Furthermore, $\lim_{l\uparrow y}w(l,y)=\lim_{l\downarrow x}w(x,l)=\alpha_\nu$ and $\lim_{l\uparrow y}z(l,y)=\lim_{l\downarrow x}z(x,l)=\beta_\nu$.
\end{enumerate}
\end{lem}
\begin{cor}\label{cor:new:wzContinuous}
    The functions $w,z:\sA^<\to(\alpha_\nu,\beta_\nu)$ (as in Lemma \ref{lem:PropertiesofsA}) are jointly continuous. 
\end{cor}
\begin{proof}
It is a well known fact that the separate continuity and monotonicity imply joint continuity; see Young \cite{Young}.
\end{proof}

\begin{remark}\label{rem:zwAtB}
    The proof of Lemma \ref{lem:PropertiesofsA} reveals that, in the case $\beta_\mu<\beta_\nu$, for $(g^L)^{-1}(\beta_\mu)\leq x<\beta_\mu=B(x)=y$ we have that $z(x,y)=z(x,B(x))=g^L(x)$ and $w(x,y)=w(x,B(x))=f^L(x)$.
\end{remark}

\textit{Step 2: finding the starting point $(x_1,y_0)\in\sA^<$.} For each $(x,y) \in { \sA^<}$ define $\Psi_{x,y}$ by setting 
\begin{equation}\label{eq:PsiXY} \Psi_{x,y}(u) = \begin{cases}
    b(u) \hspace{20mm} & x \leq w(x,y) \\
    T_{b,a}^{w(x,y),x}(u) & w(x,y) < u < x \\
    T_{a,a}^{x,y}(u) & x \leq  u \leq y \\
    T_{a,b}^{y,z(x,y)}(u) & y < u < z(x,y) \\
    b(u) & z(x,y) \leq  u.
\end{cases}
\end{equation}
Note that, if $w(x,y)=x$ (resp., $z(x,y)=y$) then the second (resp.,  fourth) line in the definition of $\Psi_{x,y}$ is redundant. This happens when $x\leq e^L$ (resp., $y\geq e^R$) and $(x,y)\in\sA$ lies on the boundary, i.e., $y=B(x)$ (resp., $x=B^{-1}(y)$).

For $(x,y)\in\sA^<$, let $\tilde S_{b,a}^{w(x,y),x}, \tilde S_{a,a}^{x,y}$ and $\tilde S_{a,b}^{y,z(x,y)}$ be the slope of $\Psi_{x,y}$ on $(w(x,y),x), (x,y)$ and $(y,z(x,y))$, respectively. If $w(x,y)=x$ (resp.,  $z(x,y)=y$) then we set $\tilde S_{b,a}^{w(x,y),x}=\infty$ (resp., $\tilde S_{a,b}^{y,z(x,y)}=-\infty$).

Recall that in Case (C3) we are in one of the following situations: 
\begin{enumerate}
\item[$(a)$] $\hat{z}^L \geq e^R$ and/or $\hat{z}^R \leq e^L$; moreover  $\hat z^R<\hat z^L$;  
\item[$(b)$] 
$e^L < \hat{z}^R < \hat{z}^L < e^R$, so that there exists $z^* \in [\hat z^R,\hat z^L]\subset(e^L,e^R)$ such that $C^L(z^*)=C^R(z^*) \geq a(z^*)$, but $\Psi$ defined by \eqref{eq:Bdef} is not convex. 
\end{enumerate}
where the strict inequality in Case $(b)$ follows from Remark~\ref{rem:noteq}.

We subdivide Case $(a)$ into the disjoint union $D_1 \cup D_2 \cup D_3 \cup D_4 \cup D_5$ where $D_1=(\hat{z}^L \geq e^R > e^L  \geq \hat{z}^R)$, $D_2= (\hat{z}^L \geq e^R > \hat{z}^R > e^L)$, $D_3 = (\hat{z}^L > e^R = \hat{z}^R)$, $D_4 = (e^R > \hat{z}^L > e^L  \geq \hat{z}^R)$ and $D_5 = (\hat{z}^L = e^L > \hat{z}^R)$.

Note that on $D_2$ we must have that either $C^L > C^R$ on $(\hat{z}^R,e^R)$ or there exists a point $z^* \in (\hat z^R,e^R)\subset(e^L,e^R)$ such that $C^L(z^*)=C^R(z^*)\geq a(z^*)$ (indeed, if $C^L(z)\leq C^R(z)$ for some $z\in(\hat{z}^R,e^R)$, then the existence of $z^*$ follows from the continuity of $C^L,C^R$ and the fact that $C^L(\hat z^R)>a(\hat z^R)=C^R(\hat z^R)$). In the latter case, without loss of generality, we further assume that $\Psi$ defined in \eqref{eq:Bdef} is not convex\footnote{As in the discussion after the initial introduction of Cases (C1), (C2) and (C3), it must be the case that $\Psi$ is not convex, but we do not show this.} (if $\Psi$ was convex then it would be like being in Case (C1) and we already know how to find an optimal model and optimal hedge in that setting). Similarly, on $D_4$ we must have that either $C^L < C^R$ on $(e^L,\hat{z}^L)$ or there exists a point $z^* \in (e^L, \hat z^L)\subset(e^L,e^R)$ such that $C^L(z^*)=C^R(z^*)\geq a(z^*)$ and $\Psi$ (as in \eqref{eq:Bdef}) is not convex. Moreover, we further divide $D_1$ according as $C^L>C^R$ on $(e^L,e^R)$, $C^L<C^R$ on $(e^L,e^R)$
or there exists a point $z^* \in (e^L,e^R)$ such that $C^L(z^*)=C^R(z^*)\geq a(z^*)$. In the last case we again assume that $\Psi$ defined in \eqref{eq:Bdef} is not convex.

Recombining various cases we have that Case (C3) can be divided into
\begin{enumerate}
\item[$(i)$] $\hat{z}^L \geq e^R > \hat{z}^R$ 
moreover, $C^L(z) > C^R(z)$ on $(\hat{z}^R \vee e^L,e^R)$;
\item[$(ii)$] $\hat{z}^L > e^R = \hat{z}^R$; 
\item[$(iii)$] $\hat{z}^R \leq e^L <\hat z^L$; moreover, $C^L(z) < C^R(z)$ on $(e^L,\hat{z}^L \wedge e^R)$; 
\item[$(iv)$] $\hat{z}^R < e^L = \hat z^L$; 
\item[$(v)$] 
$\hat{z}^R < \hat{z}^L$ and there exists $z^* \in (e^L,e^R)$ such that $C^L(z^*)=C^R(z^*) \geq a(z^*)$, but $\Psi$ defined by \eqref{eq:Bdef} is not convex. 
\end{enumerate}

Note that if $e^R = \hat{z}^R$ then the set $(\hat{z}^R \vee e^L,e^R)$ is empty, so we have omitted this condition from $(ii)$; a similar consideration applies in $(iv)$.


Suppose we are in Case $(ii)$.  Then $a(e^R) = b(e^R)$. We take $z_0=e^R$ and define $(x_0 = (g^L)^{-1}(z_0),y_0=e^R)$.
Then $g^R(y_0) = y_0 = e^R=z_0$, and the fourth line in the definition of $\Psi_{x_0,y_0}$ is redundant. 
Similarly, in Case $(iv)$ we take $z_0=e^L$ and $(x_0 = e^L, y_0 =(f^R)^{-1}(z_0))$.
Then $f^L(x_0) = x_0 = e^L=z_0$ and the second line in the definition of $\Psi_{x_0,y_0}$ is redundant.

Suppose now that we are in Case $(v)$. In this case take $x_0 = (g^L)^{-1}(z^*)$ and $y_0 = (f^R)^{-1}(z^*)$.

It remains to define $(x_0,y_0)$ in  Cases $(i)$ and $(iii)$. We consider Case $(i)$; the construction in Case $(iii)$ follows a symmetric argument.
Since $\hat{z}^R<e^R$ we must have $a(e^R) > b(e^R)$. For $z \in (\hat{z}^R,e^R)$ we can define $(\check{x}(z) = (g^L)^{-1}(z),\check{y}(z) = (f^R)^{-1}(z))$. By taking $z$ close enough to $e^R$ (in Case $(iii)$ we take $z\in(e^L,\hat z^L)$ close enough to $e^L$) we can ensure that $-\infty < \tilde{S}^{\check{y}(z),z}_{a,b} < \tilde{S}^{\check{x}(z),\check{y}(z)}_{a,a}< \tilde{S}_{b,a}^{f^L(\check{x}(z)),\check{x}(z)}<\infty$. Let $z_0\in(e^L,e^R)$ be a value with this property. Finally, let $(x_0 = \check{x}(z_0),y_0 = \check{y}(z_0))$.

Note that (by construction) in each case we have that $(x_0,y_0)\in\sA^<$ is such that both $x_0,y_0$ are finite, $y_0=B(x_0)\leq\beta_\mu$ and $x_0\leq (g^L)^{-1}(\beta_\mu)$.

Our aim is to find $x_1$ with  $x_0 \leq x_1 < y_0$  (so that  $(x_1,y_0) \in \sA^<$) such that 
either $\Psi_{x_1,y_0}=T_{a,a}^{x_1,y_0}=T_{a,b}^{x_1,z(x_1,y_0)}$ on $(x_1,z(x_1,y_0))$ and hence $\tilde S_{b,a}^{w(x_1,y_0),x_1} \geq \tilde S_{a,a}^{x_1,y_0}=\tilde S_{a,b}^{y_0,z(x_1,y_0)}$ (at least if $w(x_1,y_0) < x_1 < y_0 < z(x_1,y_0)$), or $\Psi_{x_1,y_0}=T_{b,a}^{w(x_1,y_0),x_1}=T_{b,a}^{w(x_1,y_0),y_0}$ on $(w(x_1,y_0),y_0)$ and $\tilde S_{b,a}^{w(x_1,y_0),x_1}=\tilde S_{a,a}^{x_1,y_0} \geq \tilde S_{a,b}^{y_0,z(x_1,y_0)}$ (again, at least if $w(x_1,y_0) < x_1 < y_0 < z(x_1,y_0)$).

Define 
\begin{eqnarray} 
\label{eq:sLaab}
{\sL}^{a,a,b}_{3} & = & \{(x,y,z) : \alpha_\mu < x < y < \beta_\mu,  y \leq z < \beta_\nu, \\
&& \hspace{30mm} \mbox{$(x,a(x))$, $(y,a(y))$ and $(z,b(z))$ are collinear} \}  \nonumber \\
&& \hspace{4mm} \cup \{(x,\beta_\mu,z) : \alpha_\mu < x < \beta_\mu < z < \beta_\nu, \nonumber \\
&& \hspace{30mm} a(\beta_\mu-)  \leq  \frac{z-\beta_\mu}{z-x} a(x) + \frac{\beta_\mu-x}{z-x} b(z) \nonumber \}.
\end{eqnarray}
Then, for $(x,y,z) \in {\sL}^{a,a,b}_{3}$ with $y<\beta_\mu$, we have that $T_{a,b}^{x,z}=T_{a,a}^{x,y}=T_{a,b}^{y,z}$ is the straight line passing through $(x,a(x))$, $(y,a(y))$ and $(z,b(z))$. When $y=\beta_\mu<\infty$ we have that $T_{a,b}^{x,z}(\beta_\mu)$ lies in the interval $(a(\beta_\mu-),\infty)$. 
Similarly, define 
\begin{eqnarray} 
\label{eq:sLbaa}
{\sL}^{b,a,a}_{3} & = & \{(w,x,y) : \alpha_\mu < x < y < \beta_\mu, \alpha_\mu < w \leq x\\
&& \hspace{30mm} \mbox{$(w,b(w))$, $(x,a(x))$ and $(y,a(y))$  are collinear} \} \nonumber \\
&& \hspace{4mm} \cup \{(w,\alpha_\mu,y) : \alpha_\nu < w < \alpha_\mu < y < \beta_\mu, \nonumber \\
&& \hspace{30mm} a(\alpha_\mu+)  \leq  \frac{\alpha_\mu-w}{y-w} a(y) + \frac{y-\alpha_\mu}{y-w} b(w) \nonumber \} 
\end{eqnarray}

Consider Case $(ii)$ whence $(x_0,y_0) = ((g^L)^{-1}(e^R),e^R)$ and $a(e^R)=b(e^R)$. In this case we set $x_1=x_0$. Then $w(x_1,y_0) = f^L(x_1) < x_1 < y_0 = z(x_1,y_0)$. It follows 
from the twin facts that $a(y_0)=b(y_0)$ and $y_0=z(x_1,y_0)$ that (tautologically)
$T_{a,a}^{x_1,y_0} = T_{a,b}^{x_1,z(x_1,y_0)}$.
Moreover, from the fact that $\hat{z}^L > e^R$ we have that $C^L(e^R)> a(e^R)$ and hence $\tilde S_{b,a}^{w(x_1,y_0),x_1} > \tilde S_{a,a}^{x_1,y_0=e^R}$. We find $(x_1,y_0,z(x_1,y_0))= (x_1,e^R,e^R)$ is such that both $(x_1,a(x_1))$, $(y_0,a(y_0))$ and $(z(x_1,y_0),b(z(x_1,y_0))$ lie on the same straight line (because the last two points are the same point).

Case $(iv)$ is similar. This time we adjust $y$ and leave $x_0=e^L$ unchanged. We find $y_1\in (x_0,y_0)$ such that both $(w(x_0,y_1), b(w(x_0,y_1))$ and $(x_0,a(x_0))$, $(y_1,a(y_1))$ lie on the same straight line (because the last two points are the same point).


Now we consider Cases $(v)$, $(i)$ and $(iii)$.
In Case $(v)$, recall that  $x_0 = (g^L)^{-1}(z^*)$ and $y_0 = (f^R)^{-1}(z^*)$. In Cases $(i)$ and $(iii)$ we have that $x_0 = (g^L)^{-1}(z_0)$ and $y_0= (f^R)^{-1}(z_0)$. In all cases
$e^L < x_0 < y_0 < e^R$. For $x\in[x_0,y_0)$, let $\tilde S^L_x:=\tilde S_{b,a}^{w(x,y_0),x},\tilde S^a_x:=\tilde S_{a,a}^{x,y_0},\tilde S^R_x:=\tilde S_{a,b}^{y_0,z(x,y_0)}$ and note that $\Psi_{x_0,y_0}$ (as in \eqref{eq:PsiXY}) is continuous and $\infty>\tilde S^{L}_{x_0}> \tilde S^{a}_{x_0} > \tilde S^{R}_{x_0}>-\infty$, since we are in Case $(i)$, $(iii)$ or $(v)$. Let $w(x,y_0)$ and $z(x,y_0)$ be as defined in Lemma~\ref{lem:PropertiesofsA} Now, keeping $y_0$ fixed, consider increasing $x$ from $x_0$. By Lemma~\ref{lem:PropertiesofsA}$(ii)$, as $x$ increases we have that
$w(x,y_0)$ decreases and $z(x,y_0)$ increases. Then, by the convexity of $a$ and $b$, the slopes $\tilde{S}^a_x$ and $\tilde{S}^R_x$ increase, while the slope $\tilde{S}^L_x$ may (or may not) be monotone. On the other hand, all three slopes $x\mapsto \tilde{S}^a_x,\tilde{S}^R_x,\tilde{S}^L_x$ are continuous on $(x_0,y_0)$. We increase $x$ until either $\tilde{S}^L_{x} = \tilde{S}^a_{x}$ or $\tilde{S}^a_{x}=\tilde{S}^R_{x}$. Let $x_1 = \inf \{ x_0<x<y_0 : \tilde{S}^L_{x} \leq \tilde{S}^a_{x} \mbox{ or } \tilde{S}^a_{x} \leq \tilde{S}^R_{x} \}$ and set $\inf\emptyset=y_0$. We argue that we must have $x_0<x_1 < y_0$. Note that, as $x \uparrow y_0$ then $\mu ((\alpha_\mu,x) \cup(y_0,\beta_\mu)) \uparrow 1$ and so we must have $w(x,y_0) \downarrow \alpha_\nu$ and $z(x,y_0) \uparrow \beta_\nu$.

Since $\infty>\tilde S^{L}_{x_0}> \tilde S^{a}_{x_0} > \tilde S^{R}_{x_0}>-\infty$, by the continuity we have that there exists $\epsilon>0$ such that for all $x\in[x_0,x_0+\epsilon)$, $\infty>\tilde S^{L}_{x}> \tilde S^{a}_{x} > \tilde S^{R}_{x}>-\infty$. It follows that $x_0<x_0+\epsilon\leq {x_1}$.

Now suppose $x_1=y_0$. Then $\tilde S^{L}_{x}> \tilde S^{a}_{x} > \tilde S^{R}_{x}$ for all $x\in[x_0,y_0)$. By construction, we then have that
$$
a(u)\geq  T_{a,a}^{x,y_0}(u)>T_{b,a}^{w(x,y_0),x}(u)\geq b(u),\quad u\in(w(x,y_0),x)
$$
and
$$
a(u)\geq  T_{a,a}^{x,y_0}(u)>T_{a,b}^{y_0,z(x,y_0)}(u)\geq b(u),\quad u\in(y_0,z(x,y_0)).
$$
It follows (using the continuity w.r.t. $x\in(x_0,y_0)$) that
$$
a(u)\geq \lim_{x\uparrow y_0}T_{a,a}^{x,y_0}(u)=a(y_0)+a'_-(y_0)(u-y_0)\geq b(u),\quad u\in(\alpha_\nu,\beta_\nu),
$$
where $a'_-$ denotes the left derivative of a convex function $a$. But this contradicts the Standing Assumption \ref{sass:ageqLb}, and thus we conclude that $x_1<y_0$.

It is possible that $\tilde{S}^L_{x_1} = \tilde{S}^a_{x_1} =  \tilde{S}^R_{x_1}$. In that (rather special) case we have found \[ (w(x_1,y_0),x_1,y_0,z(x_1,y_0)) \in \sL_4 \cap \Gamma_4\] and the proof of Theorem~\ref{thm2} is complete. More generally, we are in either Case $(I)$:  $\tilde{S}^L_{x_1} > \tilde{S}^a_{x_1} = \tilde{S}^R_{x_1}$ so that $(x_1,y_0,z(x_1,y_0))\in\sL^{a,a,b}_3$; or in Case $(II)$: $\tilde{S}^L_{x_1} = \tilde{S}^a_{x_1} > \tilde{S}^R_{x_1}$ so that $(w(x_1,y_0),x_1,y_0)\in\sL_3^{b,a,a}$. In each case $(w(x_1,y_0),x_1,y_0,z(x_1,y_0)) \in \Gamma_4$ but only one of the points 
$(w=w(x_1,y_0),b(w))$ or $(z=z(x_1,y_0),b(z))$ lies on the straight line $T_{a,a}^{x_1,y_0}$ passing through $(x_1,a(x_1))$ and $(y_0,a(y_0))$.

In all the cases we found $(x_1,y_1) \in \sA^<$ such that either $(w(x_1,y_1),b(w(x_1,y_1)))$ lies on 
$T_{a,a}^{x_1,y_1}$ (this corresponds to $(x_1,y_1)=(x_0,y_1)$ in Case $(iv)$ or $(x_1,y_1)=(x_1,y_0)$ in Case $(II)$) or $(z(x_1,y_1),b(z(x_1,y_1)))$ lies on 
$T_{a,a}^{x_1,y_1}$ (this corresponds to $(x_1,y_1)=(x_1,y_0)$ in Case $(ii)$ or $(x_1,y_1)=(x_1,y_0)$ in Case $(I)$); if both, then we have found a point $(w(x_1,y_1),x_1,y_1,z(x_1,y_1)) \in \sL_4 \cap \Gamma_4$ and we are done.

Without loss of generality we assume that $(z(x_1,y_1),b(z(x_1,y_1)))$ lies on 
$T_{a,a}^{x_1,y_1}$ and $T_{a,a}^{x_1,y_1}(w(x_1,y_1)) > b(w(x_1,y_1))$. In particular, we are either in Case $(ii)$ or Case $(I)$, so that (in both cases) $(x_1,y_1)=(x_1,y_0)$. This is the starting point of Step 3.

\textit{Step 3: the map $x\mapsto y^*(x)$.}
Our aim is to increase $x$ from $x_1$ and to adjust $y=y_x$ (starting from $y_0$), so that $(x,y,z(x,y)) \in \sL^{a,a,b}_{3}$ and $(w(x,y),x,y,z(x,y)) \in \Gamma_4$, and then to keep increasing $x$ (and simultaneously adjusting $y=y_x$) until $(w(x,y), x , y, z(x,y)) \in \sL_4 \cap \Gamma_4$. 

Let $\sA_1=\{(x,y,z):\alpha_\mu< x< y < \beta_\mu,~ y<z < \beta_\nu \}$.
Define 
$\hat{\Lambda}: \sA_1 \to \R 
$ by
\[ \hat{\Lambda}(x,y,z) = a(y) + \frac{z-y}{y-x}
(a(y)-a(x)) - b(z). \] 
Extend the definition to $y=x<z$ by $\hat\Lambda(x,x+,z) = a(x) + (z-x) a'(x+) - b(z)$ and $\hat\Lambda(x-,x,z) = a(x) + (z-x) a'(x-) - b(z)$, 
and to $y = \beta_\mu<\infty$ by taking $a$ to have the value $a(\beta_\mu-)$ there: for $x<\beta_\mu < z$
\begin{equation}
    \label{eq:Lambdabeta}
\hat{\Lambda}(x,\beta_\mu,z) = a(\beta_\mu-) + \frac{z-\beta_\mu}{\beta_\mu-x}
(a(\beta_\mu-)-a(x)) - b(z).
\end{equation}

Let $\Lambda: \sA^< \to \R$ be given by $\Lambda(x,y) = \hat{\Lambda}(x,y,z(x,y))$.
Let $\Lambda(x,{ x+}) = \lim_{y \downarrow x} \hat{\Lambda}(x,y,z(x,y)) = a(x) + \lim_{z \uparrow \beta_\nu} \{ (z-x) a'(x+) - b(z) \}$.
Note that $\Lambda$ may not be continuous at the diagonal if $a'$ is not continuous, but $\Lambda$ is continuous on $\sA^<$.

Recall the definition of $B$ (see \eqref{eq:BoundaryB}) and that $\bar{x}= (g^{L})^{-1}(e^R)$.

Moreover, recall the final part of Standing Assumption~\ref{sass:payoffs} that there are no intervals 
$I \subseteq (\alpha_\mu,\beta_\mu)$ on which $a=b$ 
and $a$ is linear. Proposition~\ref{prop:Lambday} is the only place where this assumption is used.


\begin{prop}
\label{prop:Lambday}
Suppose $(x,y')\in\sA^<$ with $x_1\leq x$ and $\Lambda(x,y')\leq 0$. Then $\Lambda(x,\cdot)$ is non-decreasing on $(x,y')$.

Furthermore, if $\Lambda(x,y')=0$, then  $\Lambda(x,y)>0$  for $y \in (y',B(x){]}$ and $\Lambda(x,y)<0$ for $y \in (x,y')$.

\end{prop}

\begin{prop}
\label{prop:LambdaxB}
$\Lambda(x,B(x)) \geq 0$ for $x \in [{ x_0}, (g^L)^{-1}(\beta_\mu)]$. 
\end{prop}

Let $\tilde{x}= \inf \{ x \in (x_1, \beta_\mu) : \Lambda(x,x+) \geq 0 \}$ where $\inf \emptyset = \beta_\mu$. 

We now argue that $\tilde x>x_1$. It is enough to show that $\Lambda(x,x+)<0$ {on some interval immediately} to the right of $x_1$. Since $\Lambda (x_1,y_0)=0$ and $x_1<y_0$, by the final part of Proposition \ref{prop:Lambday} we have that $\Lambda (x_1,y)<0$ for each $y\in (x_1,y_0)$. Fix $y \in (x_1,y_0)$. Then, by the continuity of $\Lambda(\cdot,y)$, for small enough $\epsilon =\epsilon(x_1,y) \in (0,y-x_1)$ we must have that $\Lambda (x,y)<0$ for all $x\in[x_1,x_1+\epsilon(x_1,y))$. Then using the monotonicity of $\Lambda(x_1,\cdot)$ on $(x_1,y)$ (see the first part of Proposition~\ref{prop:Lambday}),
we have that $\Lambda(x,x+) < 0$. {Hence $\tilde{x} \geq x_1+\epsilon > x_1$.}


For $x \in [x_1, \tilde x)$ let $y^*=y^*(x)$ denote the unique point in $(x,B(x)]$ such that $\Lambda(x,y^*(x))=0$. We will show that $y^*$ exists on a suitable subset of $[x_1, \tilde x)$. Then we will extend the definition of $y^*$ to the whole of $[x_1, \tilde x)$ and show that this extended version is continuous as a map $x\mapsto y^*(x)$.

Write $[x_1,\tilde x)$ as a disjoint union $\sA^\Lambda_{>}\cup\sA^\Lambda_{=}\cup\sA^\Lambda_{<}$, where $\sA^\Lambda_{>}:=\{x\in[x_1,\tilde x):\Lambda(x,B(x))>0\}$, $\sA^\Lambda_{=}:=\{x\in[x_1,\tilde x):\Lambda(x,B(x))=0\}$ and $\sA^\Lambda_{<}:=\{x\in[x_1,\tilde x):\Lambda(x,B(x))<0\}$. Note that, due to Proposition \ref{prop:LambdaxB}, $\sA^\Lambda_{<}\subseteq ((g^L)^{-1}(\beta_\mu)\wedge\tilde x,\tilde x)$. Furthermore, $x_1\in \sA^\Lambda_{>}\cup\sA^\Lambda_{=}$.

\begin{cor}
\label{cor:existenceContinuity}
$y^*(x) \in (x,B(x)]$ exists for all $x \in [x_1, \tilde x)\setminus\sA^\Lambda_{<}$. Furthermore, by setting $y^*(x)=B(x)=\beta_\mu$ for $x\in\sA^\Lambda_<$,  we have that $x\mapsto y^*(x)$ is continuous on $[x_1, \tilde{x})$.
\end{cor}

\begin{lem}\label{lem:tildeX}
$\lim_{x\uparrow\tilde x}y^*(x)=\tilde x$.
\end{lem}

\textit{Step 4: finding $(w,x,y,z)\in\sL_4\cap\Gamma_4$.}
Define $\Upsilon:[x_1,\tilde{x} ) \to \R$ by
\[ \Upsilon(x) = a(x) + (w(x,y^*(x)) -x) \frac{b(z(x,y^*(x)) - a(x)}{z(x,y^*(x))-x} - b(w(x,y^*(x))) \]
Note that $\Upsilon$ has been defined in such a way that it does not depend on the value of $a(y^*(x))$. This is important in the case where $y^*(x)= \beta_\mu$.

\begin{lem}\label{lem:Upsilon0}
If $\Upsilon(x)=0$ then $(w(x,y^*(x)),x,y^*(x),z(x,y^*(x))) \in \sL_4 \cap \Gamma_4$.
\end{lem}

\begin{proof}
Clearly, $(w(x,y^*(x)),x,y^*(x),z(x,y^*(x))) \in \Gamma_4$ by definition.

Again, by definition, if $y^*(x)<\beta_\mu$ then $(x,y^*(x),z(x,y^*(x))) \in \sL^{a,a,b}_3$ where $\sL^{a,a,b}_3$ is defined in \eqref{eq:sLaab} so that 
\begin{equation}
    \label{eq:a=T}
a(y^*(x)) = T_{a,b}^{x,z(x,y^*(x))}(y^*(x)).
\end{equation}
Otherwise, if $y^*(x)=\beta_\mu$ then $\hat{\Lambda}(x,\beta_\mu,z(x,\beta_\mu)) \leq 0$ and then, by \eqref{eq:Lambdabeta},
\[ b(z(x,\beta_\mu)) \geq  
 a(\beta_\mu-) + \frac{z(x,\beta_\mu)-\beta_\mu}{\beta_\mu-x}
(a(\beta_\mu-)-a(x)), 
\]
or equivalently 
\[ a(\beta_\mu-) \leq \frac{z(x,\beta_\mu)-\beta_\mu}{z(x,\beta_\mu)-x} a(x) +   \frac{\beta_\mu-x}{z(x,\beta_\mu)-x} b(z(x,\beta_\mu)) = T_{a,b}^{x,z(x,\beta_\mu)}(\beta_\mu) \]
so that $(x, \beta_\mu,z(x,\beta_\mu)) \in \sL^{a,a,b}_3$.

Now suppose $\Upsilon(x)=0$. 
Then 
\[ a(x) = \frac{z(x,y^*(x))-x}{z(x,y^*(x))-w(x,y^*(x))} b(w(x,y^*(x))) + \frac{x - w(x,y^*(x))}{z(x,y^*(x))-w(x,y^*(x))} b(z(x,y^*(x))), \]
so that 
\begin{equation}
\label{eq:axL4}    
a(x)=T^{w(x,y^*(x)),z(x,y^*(x))}_{b,b}(x).
\end{equation}
Note that, for $w<x<z$, if $a(x)=T^{w,z}_{b,b}(x)$ then 
$T^{w,z}_{b,b} = T^{x,z}_{a,b} = T^{w,x}_{b,a}$. Then it follows from \eqref{eq:a=T} and \eqref{eq:axL4} that
\begin{equation}
\label{eq:ayL4}  a(y^*(x)) = T_{a,b}^{x,z(x,y^*(x))}(y^*(x)) = T_{b,b}^{w(x,y^*(x)),z(x,y^*(x))}(y^*(x)).
\end{equation}
Then, since \eqref{eq:axL4} and \eqref{eq:ayL4} hold, we have that
$(w(x,y^*(x)),x,y^*(x),z(x,y^*(x)) \in \sL_4$.

If $y^*(x) = \beta_\mu$ a similar argument gives that $(w(x,\beta_\mu),x,\beta_\mu,z(x,\beta_\mu)) \in \sL_4$.
\end{proof}

Since we are in Case (C3) we have that $\Upsilon(x_1)>0$.




\begin{prop}\label{prop:solution}
    Suppose we are in Case (C3) and $\Upsilon(x_1)>0$. Then there exists $x^* \in (x_1,\tilde{x})$ such that $\Upsilon(x^*)=0$.
\end{prop}

\begin{proof}
Since $y^*$ is continuous on $[x_1,\tilde x)$ (recall Corollary \ref{cor:existenceContinuity}), and $z$ and $w$ are jointly continuous (see Corollary~\ref{cor:new:wzContinuous}), it follows that $\Upsilon$ is continuous
on $(x_1, \tilde{x} )$. Then since $\Upsilon(x_1)>0$ the result will follow if we can find
$x < \tilde{x}$ such that $\Upsilon(x)<0$.


Recall the definitions of $\hat{w}$, $\hat{u}$ and $\hat{z}$ in Lemma~\ref{lem:SA3}. Due to convexity of $a$ and $b$, without loss of generality we can (and do) assume that $\hat u\in(\alpha_\mu,\beta_\nu)$ and $\hat w,\hat z\in(\alpha_\nu,\beta_\nu)$.
Define $\hat{u}_0 = \inf\{ u : a(u) \leq T_{b,b}^{\hat{w},\hat{z}}(u) \}$; note that $\hat u_0\in[\alpha_\mu,\hat u)$ is finite and $\hat w\leq \hat u_0$.

We first show that if $x>\hat{u}_0$,
$z(x,y^*(x)) > \hat{z}$ and $w(x,y^*(x))< \hat{w}$ then $\Upsilon(x)<0$.

Let $T^x= T_{a,b}^{x, z(x,y^*(x))}$ be the line joining $(x,a(x))$ and $(z(x,y^*(x)),b(z(x,y^*(x))))$, and note that, by construction, $T^x=T_{a,a}^{x,y^*(x)}=T_{a,b}^{y^*(x),z(x,y^*(x))}$.
Suppose that $z(x,y^*(x)) > \hat{z}$. Then $T^x \leq a$ on $(-\infty, x)$ and moreover, $T^x \leq T^{\hat{w},\hat{z}}_{b,b}$ on $( -\infty, \hat{u}_0)$.  On the other hand, $T_{b,b}^{\hat w,\hat z}<b$ on $(-\infty,\hat w)$. Since $w(x,y^*(x))< \hat{w}\leq\hat u_0$, it follows that
\[ T^x(w(x,y^*(x))) = a(x) + \frac{(w(x,y^*(x))-x)}{(z(x,y^*(x))-x)}(b(z(x,y^*(x)))-a(x)) \leq T^{\hat{w},\hat{z}}_{b,b}(w(x,y^*(x)))< b(w(x,y^*(x))), \]
and therefore $\Upsilon(x)<0$.

{Case A:} $\tilde{x} =\beta_\mu$. 
Then $y^*(x)>x$ on $[x_1,\beta_\mu)$.
Then as $x \uparrow \beta_\mu$ we have $y^*(x) \rightarrow \beta_\mu$ and $z(x,y^*(x)) \uparrow \beta_\nu$ and $w(x,y^*(x)) \downarrow \alpha_\nu$.

In particular, for large enough $x$ we have $x> \hat{u}_0$, $z(x,y^*(x))> \hat{z}$ and $w(x,y^*(x)) < \hat{w}$. Then $\Upsilon(x)<0$.

{Case B:} $\tilde{x}<\beta_\mu$. As $x \uparrow \tilde{x}$ we have that $y^*(x) \rightarrow \tilde{x}$, $z(x,y^*(x)) \uparrow \beta_\nu$ and $w(x,y^*(x)) \downarrow \alpha_\nu$.

Suppose that $\tilde{x}> \hat{u}_0$. Then for $x$ close enough to $\tilde{x}$ we have $x>\hat{u}_0$, $z(x,y^*(x)) > \hat{z}$ and $w(x,y^*(x))< \hat{w}$. Then, just as in Case A, $\Upsilon(x)<0$.

Now suppose that $\tilde{x} \leq \hat{u}_0$. We find a contradiction and conclude that this case cannot happen. 
Let $z_0$ be the point above $\hat{u}_0$ where $T_{a,a}^{\hat{u}_0,\hat{u}}$
crosses $b$. Then $z_0<\hat{z}$, and for $x$ close enough to $\tilde{x}$ so that $y^*(x)<\hat{u}$,  $T_{a,a}^{x,y^*(x)}$ has a smaller slope than 
$T_{a,a}^{\hat{u}_0,\hat{u}}$ and
hence, since $(x,a(x))$, $(y^*(x),a(y^*(x))$ and $(z(x,y^*(x)), b(z(x,y^*(x)))$
are co-linear, we have that
$z(x,y^*(x))<z_0<\hat{z}$. But, $\lim_{x \uparrow \tilde{x}} z(x,y^*(x)) = \beta_\nu$, a contradiction.

\end{proof}


{It follows from the proposition that in Case (C3) there exists $x \in (x_1,\tilde{x})$ such that $\Upsilon(x)=0$. Then, by Lemma \ref{lem:Upsilon0}, $(w(x,y^*(x)),x,y^*(x),z(x,y^*(x))) \in \sL_4 \cap \Gamma_4$. The proof of Theorem~\ref{thm2} is then complete.
Further, we have constructed an optimal model and an optimal superhedge. The model and superhedge are illustrated in Figure~\ref{fig:case3final}.}

\begin{figure}
	\centering
\begin{tikzpicture}
\def\el{-18};
\def\er{-4};
\def\xl{-14};
\def\xr{-8};
\begin{axis}[
width=4.421in,
height=5.566in,
at={(0.758in,0.481in)},
scale only axis,
xmin=-10,
xmax=30,
ymin=-5,
ymax=54,
axis line style={draw=none},
ticks=none
]
\draw[thin,gray] (-2,0)--(22,0);
\draw[thin,gray] (-2,30)--(22,30);

\draw[blue,dotted, very thick] (2,20) to[out=280, in=180] (10,4) to[out=360, in=250] (19,17);
\draw[blue,dashed, thick] (0,13) to[out=300, in=180] (11,1) to[out=360, in=260] (20,16);
\draw[red] (4,5.2)--(19,10.5);
\node[circle,fill=red,inner sep=0pt,minimum size=5pt] at (4,5.2) {};
\node[circle,fill=red,inner sep=0pt,minimum size=5pt] at (6.1,5.9) {};
\node[circle,fill=red,inner sep=0pt,minimum size=5pt] at (16.6,9.7) {};
\node[circle,fill=red,inner sep=0pt,minimum size=5pt] at (19,10.5) {};

\draw[thin,gray,dotted] (-2,30)--(22,54);
\draw[thin,gray,loosely dashed] (4,25)--(4,51)--(6.1,51)--(6.1,28);
\draw[thin,gray,loosely dashed] (4,36)--(19,36);
\draw[thin,gray,loosely dashed] (16.6,28)--(16.6,51)--(19,51)--(19,25);
\draw[thin,gray,loosely dashed] (6.1,51)--(16.6,51)--(16.6,54)--(6.1,54);

\draw[red, thick] (19,51) -- (22,54);
\draw[red, thick] (-2,30) -- (4,36);

\node[below,scale=0.8] at (6.1,28) {$x^*$};
\node[below,scale=0.8] at (16.6,28) {$y^*(x^*)$};
\node[below,scale=0.8] at (2,25) {$w(x^*,y^*(x^*))$};
\node[below,scale=0.8] at (20,25) {$z(x^*,y^*(x^*))$};

\draw[thin,gray,loosely dashed] (4,22)--(4,-3);
\draw[thin,gray,loosely dashed] (22,54)--(22,30);
\draw[thin,gray,loosely dashed] (19,22)--(19,-3);
\draw[thin,gray,loosely dashed] (6.1,-1)--(6.1,27);
\draw[thin,gray,loosely dashed] (4,36)--(19,36);
\draw[thin,gray,loosely dashed] (16.6,27)--(16.6,-1);
\node[below,scale=0.8] at (6.1,-1) {$x^*$};
\node[below,scale=0.8] at (16.6,-1) {$y^*(x^*)$};
\node[below,scale=0.8] at (4.1,-3) {$w(x^*,y^*(x^*))$};
\node[below,scale=0.8] at (19,-3) {$z(x^*,y^*(x^*))$};

\draw[thin,gray,loosely dashed] (-2,36) -- (4,36);
\node[left,scale=0.8] at (-2,36) {$w(x^*,y^*(x^*))$};
\draw[thin,gray,loosely dashed] (19,51) -- (22,51);
\node[right,scale=0.8] at (22,51) {$z(x^*,y^*(x^*))$};
\filldraw[ opacity=0.1, red] (4,36)--(4,51)--(6.1,51)--(6.1,36)--(4,36);
\filldraw[ opacity=0.1, red] (16.6,36)--(16.6,51)--(19,51)--(19,36)--(16.6,36);
\filldraw[ opacity=0.3, blue] (6.1,51)--(16.6,51)--(16.6,54)--(6.1,54)--(6.1,51);
\filldraw[ opacity=0.3, blue] (6.1,36)--(16.6,36)--(16.6,30)--(6.1,30)--(6.1,36);
\draw[thin,gray,loosely dashed] (-2,30)--(-2,28);
\node[below,scale=0.8] at (-2,28) {$\alpha_\mu=\alpha_\nu$};
\draw[thin,gray,loosely dashed] (22,30)--(22,28);
\node[below,scale=0.8] at (22,28) {$\beta_\mu=\beta_\nu$};
\node[red,scale=0.8] at (11,9.5) {$\Psi_{x^*,y^*(x^*)}$};

\node[blue,scale=0.8] at (17.5,14.5) {$a$};
\node[blue,scale=0.8] at (17.7,5) {$b$};
\end{axis}
\end{tikzpicture}
\caption{The Case (C3). The shaded areas in the top part of the figure represent the sets (and associated exercise rules) on which the optimal models $\sM^*$ concentrate: the Bermudan option is exercised at time-1 if $Z_1\notin(x^*,y^*(x^*))$, and the mass in $(\alpha_\mu,w(x^*,y^*(x^*)))\cup(z(x^*,y^*(x^*)),\beta_\mu)$ stays put, while the mass in $(w(x^*,y^*(x^*)),x^*)\cup(y^*(x^*),z(x^*,y^*(x^*)))$ is mapped to $(w(x^*,y^*(x^*)),z(x^*,y^*(x^*)))$. On the other hand, if $Z_1\in(x^*,y^*(x^*))$, then the option will be exercised at time-2 and the mass in $x^*,y^*(x^*))$ is mapped to the tails $(\alpha_\nu,w(x^*,y^*(x^*)))\cup(z(x^*,y^*(x^*)),\beta_\nu)$.
The bottom part of the figure depicts a candidate convex function $\Psi_{x^*,y^*(x^*)}$. In particular, $\Psi_{x^*,y^*(x^*)}=b$ on $(\alpha_\nu,w(x^*,y^*(x^*))]\cup[z(x^*,y^*(x^*)),\beta_\nu)$ and $\Psi_{x^*,y^*(x^*)}=T_{b,a}^{w(x^*,y^*(x^*)),x^*}=T_{a,a}^{x^*,y^*(x^*)}=T_{b,a}^{y^*(x^*),z(x^*,y^*(x^*))}$ on $[w(x^*,y^*(x^*)),z(x^*,y^*(x^*))]$.  }
\label{fig:case3final}
\end{figure}

\section{$\sP^{can} = \sP^{rand}$ }
\label{sec:can=rand}

\begin{prop}\label{prop:can=rand}
Suppose the Dispersion Assumption (Definition~\ref{def:dispersion}) holds. Then $\sP^{can} = \sP^{rand} = \sP = \sD$.

In particular, there is a sequence of canonical models $(\sM_n)_{n \geq 1}$ where $\sM_n \in M^{can}(\mu,\nu)$ induced by a sequence of martingale couplings $(\pi_n)_{n \geq 1}$ where $\pi_n \in \Pi_M(\mu,\nu)$ and associated stopping times $\tau_n$ such that $\lim_n \E^{\pi_n}[c(X_{\tau}^n, \tau^n)] = \sP$.

On the other hand there exist $(c=(a,b);\mu,\nu)$ such that, for each $\sM \in M^{can}(\mu,\nu)$ and each stopping time $\tau$, we have $\E^\sM[c(Z_\tau,\tau)] < \sP^{can}$ so that there is no maximizer in $M^{can}(\mu,\nu)$ for $\sP$.
\end{prop}

\begin{remark}
    For each canonical model $\sM\in M^{can}(\mu,\nu)$ we can define the Snell envelope process $S^\sM$ (see Snell \cite{Snell:52}), by setting $S^\sM_2=c(Z_2,2)=b(Z_2)$, $S^\sM_1=\max\{c(Z_1,1)=a(Z_1),\mathbb{E}^\sM[S^\sM_2\lvert \sF^{can}_1]\}$ and $S^\sM_0=\mathbb{E}^\sM[S^\sM_1\lvert\sF^{can}_0]$; here we use that, in our setup, stopping at time-0 is not allowed. It is well known  that $\sup_{\tau\in\sT^{can}_{1,2}}\mathbb{E}^\sM[c(Z_\tau,\tau)]=S^\sM_0$. Now, since $\sF^{can}_0$ is trivial, $\sF^{can}_1$ is generated by $Z_1\sim\mu$ and each $\sM$ is induced by a martingale coupling $\pi\in\Pi_M(\mu,\nu)$, we have that $\sP^{can}(\mu,\nu;a,b)=\sup_{\pi\in\Pi_M(\mu,\nu)}\int \max\left\{a(x),\int b(y)\pi_x(dy)\right\}\mu(dx)$. This corresponds to a weak martingale optimal transport problem with a cost function $(x,p)\mapsto \tilde c(x,p)=\max\left\{a(x),\int b(y)p(dy)\right\}$; here $x\in\R$ and $p$ is an integrable probability measure on $\R$. 
    

Beiglb\"ock et al. \cite{BJMP:23} study the weak martingale optimal transport problem $\inf_{\pi\in\Pi_M(\mu,\nu)}\int\tilde c(x,\pi_x)\mu(dx)$ under an assumption that $p\mapsto \tilde c(x,p)$ is convex for each $x\in\R$. This is part of a set of sufficient conditions given in Beiglb\"ock et al. \cite[Theorem 2.4]{BJMP:23} for there to exist a minimizing martingale coupling $\pi^*$.
Translating their problem to our setting (i.e., moving from a minimisation problem to a maximisation problem) their existence of a maximiser within the class of canonical models result holds under the assumption that $p\mapsto \max\left\{a(x),\int b(y)p(dy)\right\}$ is concave. But, (except when $p\mapsto \max\left\{a(x),\int b(y)p(dy)\right\}$ is linear) this is not the case. Hence, Theorem~2.4 of \cite{BJMP:23} does not apply and we need not expect existence of a maximiser in the space of canonical models. This is what we find: in general, the maximiser is only attained in the larger space of randomised models.

\end{remark}

\begin{proof}[Proof of Proposition \ref{prop:can=rand}]
First we show that $\sP^{can} = \sP^{rand}$.

In Cases (C1) and (C3) we have constructed an optimal model $\sM \in M^{can}$ and associated stopping time 
for which the expected payoff under the optimal exercise rule attains $\sP = \sD$. Therefore it is sufficient to only consider Case (C2). The result will follow if we can find a sequence of martingale couplings $(\pi_n)_{n \geq 1}$ with $\pi_n \in \Pi_M(\mu,\nu)$ and sequence of sets $(B_n)_{n \geq 1}$ with $B_n \in \sB((\alpha_\mu,\beta_\mu))$ such that
\[ \lim_n 
\left\{  \int_{B_n} a(x) \mu(dx) + \int_{(\alpha_\mu,\beta_\mu) \setminus B_n} b(y) \pi_x(dy)\mu(dx) \right\} = \sP  = \sP^{rand}. \]

\noindent{Step 1}: Suppose $\xi$ is an absolutely continuous measure with density $\eta_\xi$ with support in a closed interval $[c_\xi,d_\xi]$ and suppose $\xi \leq \chi$ where $\chi$ is an absolutely continuous measure with density $\rho_\chi$. Then $0 \leq \eta_\xi \leq \rho_\chi$ on $[c_\xi,d_\xi]$ and $0 = \eta_\xi \leq \rho_\chi$ on $\R \setminus [c_\xi,d_\xi]$. Trivially $\xi \leq_E \chi$. Later $\chi$ will play the role of the initial law and $\xi$ the target law restricted to an interval.

Let $\zeta = \xi(\R) \delta_{\bar{\xi}}$  so that $\zeta \leq_{cx} \xi \leq \chi$. Then $\zeta \leq_{cx} S^\chi(\zeta)$. Further, $S^\chi(\zeta) = \chi|_{(c \chi,d_\chi)}$ where $c_\xi \leq c_\chi \leq d_\chi \leq d_\xi$ and $c_\chi < d_\chi$ unless $\xi(\R) = 0$. We also have that $S^\chi(\zeta) \leq_{cx} \xi$. (To see this note that $S^\chi(\zeta) = \xi|_{(c \chi,d_\chi)} + (\chi -\xi)|_{(c \chi,d_\chi)}$ and that $(\chi -\xi)|_{(c \chi,d_\chi)} \leq_{cx} \xi|_{[c_\xi,d_\xi] \setminus (c \chi,d_\chi)}$, as the two measures have the same mean and mass and the first measure has support inside an interval whereas the second measure has support outside this interval.) It follows that for any convex function $A$ we have
$\xi(\R) A(\bar{\xi}) \leq \int_{(c \chi,d_\chi)} A d \chi \leq \int_{(c_\xi,d_\xi)} A d \xi$.

\noindent{Step 2}: Let $[c,d]$ be fixed. We say that $\Gamma$ determines a (finite) partition of $[c,d]$ if $\Gamma = \{ (t_j)_{0 \leq j \leq J} \}$ where $0<J<\infty$ and $c = t_0 < t_1 < \cdots < t_j < t_{j+1} < \cdots < t_J = d$. Let $(\Gamma_n)_{ n \geq 1}$ be a nested sequence of sets determining partitions (i.e., $\Gamma_n = \{ c=t^n_0, t^n_1, \ldots t^n_{J(n)}=d \} \subseteq \Gamma_{n+1}$) for which the mesh size tends to zero (i.e., $\lim_{n} \max_{1 \leq j \leq J(n)} ( t^n_j - t^n_{j-1}) = 0$).

Let $\xi$ and $\chi$ be as in the previous step and decompose $\xi$ as
\[ \xi = \sum_{j = 1}^{J(n)} \xi^n_j \] 
where $\xi^n_j = \xi|_{(t^n_{j-1},t^n_j)}$. Let $\zeta^n_j = \xi((t^n_{j-1},t^n_j)) \delta_{{\bar{\xi}}^n_j}$. Note that
$\zeta^n_j \leq_{cx} \xi^n_j$. Further, by the results of Step 1 (applied to $\xi^n_j$ and $\chi^n_j : = \chi|_{(t^n_{j-1},t^n_j)}$) we have that there exists $s^n_{j,-}, s^n_{j,+}$ with $t^n_{j-1} \leq s^n_{j,-} \leq s^n_{j,+} \leq t^n_{j}$ such that $\chi|_{(s^n_{j,-},s^n_{j,+)}} \leq_{cx} \xi^n_j$. Then we have
\[ \zeta^n_j \leq_{cx} \chi|_{(s^n_{j,-},s^n_{j,+)}}   \leq_{cx} \xi^n_j . \]
Now set $B_n = \cup_{1 \leq j \leq J(n)} (s^n_{j,-},s^n_{j,+})$.
Then
\[ \zeta^n := \sum_{j=1}^{J(n)} \zeta^n_j \leq_{cx} 
\chi|_{B_n} \leq_{cx} \sum_{j=1}^{J(n)} \xi^n_j = \xi. \]
Moreover, $\zeta^n$ converges weakly to $\xi$ by construction. Therefore for any convex function $A$ we have
\begin{equation} \label{eq:limit} \lim_n \int_{[c,d]} A(x) \zeta^n(dx) = \lim_n \int_{x \in B_n}  A(x)\chi(dx) = \int_{x \in [c,d]} A(x) \xi(dx) \end{equation}

\noindent{Step 3}:
Suppose that $(a,b)$ and $(\mu,\nu)$ are such that we are in Case (C2). Let $(y_-,x_-,z_-, z_+,x_+,y_+)$ be the six points that arise in the construction of a (randomised) martingale coupling and let $\hat{\pi}^*$ be an optimiser. Recall Section \ref{sec:C2}.

Note that $\hat{\pi}^*$ can be written as
\[ \hat{\pi}^* = \hat{\pi}^*I_{ \{ x \in (z_-,z_+) \} } + \hat{\pi}^* I_{ \{ x \in (\alpha_\mu,\beta_\mu) \setminus (z_-,z_+) \} } =: \hat{\pi}^*_1 + \hat{\pi}^*_2 \]
For $i=1,2$, let $\nu_i$ be given by $\nu_i(dy) = \int_{x \in (\alpha_\mu,\beta_\mu)} \int_{u\in(0,1)}\hat\pi^*_i(dx,du, dy)$. Furthermore, let $\tilde{\nu}_1 = \nu_1 - \nu|_{(z_-,z_+)}$ and observe that $(\alpha_\nu,y_-] \cup [y_+,\beta_\nu)$ is a support of $\tilde{\nu}_1$ (i.e., $\tilde{\nu}_1(\R)=\tilde{\nu}_1((\alpha_\nu,y_-] \cup [y_+,\beta_\nu))$). Note that $\hat{\pi}_2^*$ does not use randomisation so that we can and do identify $\hat{\pi}^*_2$ with a martingale coupling $\pi^*_2$ on the canonical model given by $\pi^*_2(dx,dy) = \int_{u \in (0,1)} \hat{\pi}^*_{2}(dx,du,dy)$. 

Now we construct a sequence $(\pi^n)_{n \geq 1}$ of martingale couplings of $\mu|_{(z_-,z_+)}$ and $\nu_1 = \nu|_{(z_-,z_+)} + \tilde{\nu}_1$ and, { each of which can be constructed on a canonical model.}

We let $[c,d]=[z_-,z_+]$, $\xi = \nu|_{(z_-,z_+)}$ and $\chi = \mu|_{(z_-,z_+)}$. By the Dispersion Assumption, $[c,d]\subseteq[e^L,e^R]$ and $\xi \leq \chi$ on $[c,d]$.

By Step 2 above we find a sequence of sets $B_n\subseteq(z_-,z_+)$ such that $S^{\nu_1}(\mu|_{B_n}) = \nu|_{(z_-,z_+)}$. We let $\pi^n$ be a martingale coupling in which $\pi^n = \pi^{n,1} + \pi^{n,2}$ where $\pi^{n,1} \in \Pi_M( \mu|_{B_n},\nu|_{(z_-,z_+)})$ and $\pi^{n,2} \in \Pi_M( \mu|_{(z_-,z_+) \setminus B_n}, \tilde{\nu}_1)$.

Now we combine $\pi^n$ and $\pi^*_2$ (to give a martingale coupling of $\mu$ and $\nu$), by setting $\pi^{*,n} := \pi^n+ \pi^*_2$.  For each $n\geq 1$, {since the supports of $\mu|_{B_n}$ and $\mu|_{(z_-,z_+)\setminus B_n}$ and $\mu|_{\R \setminus (z_-,z_+)}$ are disjoint, we can do this in a way that $\pi^{*,n}$ is associated with a canonical model.}

Finally, let $\tau^n = 1$ on $(\alpha_\mu,x_-) \cup (x_+,\beta_\mu) \cup B_n$ and $\tau^n=2$ otherwise.
Then the expected payoff of the Bermudan option under $\pi^{*,n}$ (with disintegration $\pi^{*,n}_x$) and stopping rule $\tau^n$ is
\begin{eqnarray*}
\lefteqn{\int_{x \in (\alpha_\mu,x_-) \cup (x_+,\beta_\mu) \cup B_n } a(x) \mu(dx) + \int_{ x \in (x_-,z_-) \cup (z_+,x_+) \cup ((z_-,z_+) \setminus B_n)} \int_{y\in\R}  b(y) \pi^{*,n}_x(dy) \mu(dx) } \\
& = & \int_{ (\alpha_\mu,x_-) \cup (x_+,\beta_\mu)} a(x) \mu(dx) + \int_{ x \in (z_-,z_+)}  a(x) \frac{\eta(x)}{\rho(x)} \mu(dx) \\
&& + \int_{x \in B_n } a(x) \mu(dx) - \int_{ x \in (z_-,z_+)}  a(x) \frac{\eta(x)}{\rho(x)} \mu(dx)  \\
&& + \int_{y\in\R} b(y) \int_{ x \in (x_-,z_-) \cup (z_+,x_+) \cup ((z_-,z_+) \setminus B_n)}  \pi^{*,n}_x(dy) \mu(dx)  \\
& = & \int_{x \in (\alpha_\mu,x_-) \cup (x_+,\beta_\mu)}  a(x) \mu(dx) + \int_{ x \in (z_-,z_+)} a(x)  \nu(dx) + \int_{y \in (\alpha_\nu, y_-) \cup (y_+, \beta_\nu)} b(y) (\nu - \mu)(dy) \\
&& + \int_{x \in B_n } a(x) \mu(dx) - \int_{ x \in (z_-,z_+)}  a(x) \nu(dx)  \\
& = & \sP^{rand} + \int_{x \in B_n } a(x) \mu(dx) - \int_{x \in (z_-,z_+) } a(x) \nu(dx) \rightarrow  \sP^{rand}
\end{eqnarray*}
using the second equality in \eqref{eq:limit} for $\chi = \mu|_{(z_-,z_+)}$ and $\xi = \nu|_{(z_-,z_+)}$ in the last line.

Now we show that there exist payoffs $(a,b)$ and pairs $(\mu,\nu)$ such that for any $\pi \in \Pi_M(\mu,\nu)$ (and the corresponding $\sM_\pi\in M^{can}(\mu,\nu)$) and $S \in \sB((\alpha_\mu,\beta_\nu))$ we have 
\[ \E^{\sM_\pi}[a(X) I_{\{X \in S\}} + b(Y) I_{\{X \notin S\}}] < \sP^{can}. \]

Suppose $(a,b)$ and $(\mu,\nu)$ are such that we are in Case (C2), $a$ is strictly convex and $a>b$ on $(z_-,z_+)$. Also recall that {under the Dispersion Assumption} the densities of $\mu$ and $\nu$ are strictly positive on $(z_-,z_+)$.

Recall the inequalities \eqref{eq:adual} and \eqref{eq:bdual}, and suppose $(\phi,\psi,\theta_1,\theta_2) = (\phi^*=(a-\psi^*)^+,\psi^*, -(\psi^*)',0)$ is the set of optimal dual variables. If a model $\sM \in M^{can}(\mu,\nu)$ (induced by $\pi \in \Pi_{M}(\mu,\nu)$) and associated stopping rule $\{ \tau = 1 \}= \{X \in S \}$ attains the maximum then we must have that ($\pi$ almost surely) 
\[ a(X) I_{ \{ X \in S \} } + b(Y) I_{ \{ X \notin S \} } = (a - \psi^*)^+(X) + \psi^*(Y)  - I_{ \{ x \in S \} } (\psi^*)'(X) (Y-X). \]
{The equality $a(x) I_{ \{ x \in S \} } + b(y) I_{ \{ x \notin S \} } = (a - \psi^*)^+(x) + \psi^*(y)  - I_{ \{ x \in S \} } (\psi^*)'(x) (y-x)$ gives us that for $x \in S$, $a(x) = (a(x) - \psi^*)^+(x) + \psi^*(y)  - (\psi^*)'(x) (y-x) \geq a(x) - \psi^*(x) + \psi^*(x) = a(x)$ and for $x \notin S$, $b(y) = (a(x)-\psi^*(x))^+ + \psi^*(y) \geq b(y)$. 
In particular, $\pi$ can only assign mass to $(x,y)\in\R^2$ such that we have equality, i.e. $\pi$ is concentrated on}
\begin{eqnarray}
    \lefteqn{
\left( \{ x \in S \} \cap \{ a(x) \geq \psi^*(x) \} \cap  \{\psi^*(y) = \psi^*(x) + (y-x)  (\psi^*)'(x) \} \right) } \label{eq:xinS}\\
&& \hspace{20mm} \cup \left( \{ x \notin S \} \cap \{ a(x) \leq \psi^*(x) \} \cap  \{ \psi^*(y) = b(y) \} \right). \nonumber 
\end{eqnarray}

Since $\pi \in \Pi_M(\mu,\nu)$, $a=\psi^*$ is strictly convex and $a=\psi^*>b$ on $(z_-,z_+)$, we find that for $B \in \sB((z_-,z_+))$ we must have $\nu(B)=\int_{y\in B}\int_{x\in\R}\pi(dx,dy)=\int_{\{y\in B\}\cap\{x=y\}\cap\{x\in S\}}\pi(dx,dy)$. 
Suppose $\nu ((z_-,z_+) \setminus S)>0$.
Then $0 < \nu ((z_-,z_+) \setminus S) = \int_{\{y\in ((z_-,z_+) \setminus S \} \cap\{x\in S\}} I_{\{y=x\}}\pi(dx,dy) = 0$. Hence, up to a null set, $(z_-,z_+) \subseteq S$. Then, since $\psi$ convex everywhere, $a$ is strictly convex and $a(x)= \psi(x)$ on $(z_-,z_+)$, and using the top line of \eqref{eq:xinS}, 
\begin{eqnarray*} \pi\lvert_{\{ x \in (z_-,z_+) \} } & = & \pi\lvert_{\{ x \in (z_-,z_+) \}} I_{\{ a(x) \geq \psi^*(x) \} }  I_{ \{\psi^*(y) = \psi^*(x) + (y-x)  (\psi^*)'(x) \} } \\
& = & \pi\lvert_{\{ x \in (z_-,z_+)} I_{\{y=x \}}. \end{eqnarray*}
Then $\pi\lvert_{(z_-,z_+)\times\R}(dx,dy) =\pi\lvert_{(z_-,z_+)\times(z_-,z_+)} (dx,dy) I_{ \{y=x\} }=\pi\lvert_{\R \times(z_-,z_+)}(dx,dy)$. 
Since $\pi$ only charges the diagonal on $\{ x \in (z_-,z_+) \}$ the disintegration of $\pi$ with respect to $x$ must satisfy $\pi(dx,dy) I_{\{x \in (z_-,z_+)\}} = \mu(dx) \pi_x(dy) I_{\{x \in (z_-,z_+)\}} = \mu(dx) \delta_x(dy) I_{\{x \in (z_-,z_+)\}}$.
Then, 
\[ \nu((z_-,z_+)) = \pi(\R \times (z_-,z_+)) = \pi((z_-,z_+)\times\R) = \int_{\R^2} \mu(dx) \pi_x(dy) I_{\{x \in (z_-,z_+)\}} = \mu((z_-,z_+)); \] 
but this contradicts the Dispersion Assumption.
\end{proof}

\section{Further remarks and extensions}
\subsection{Dropping the assumption that $a \geq b$}

In the main text of this section we assumed that $a \geq b$. Here we show that this assumption is not necessary and that the results remain true without this assumption. Instead of requiring $a$ and $b$ are convex and $a \geq b$ it is sufficient that $b$ and $a \vee b$ are convex.

First, we consider the pricing problem. Fix a model $\sM\in M(\mu,\nu)$ and consider the model-based price of the Bermudan claim:
$$P(\sM;a,b):=\sup_{\tau \in {\sT}_{1,2}} \E^{\sM}[ a (Z_1)I_{\{{\tau}=1\}}+b(Z_2)I_{\{{\tau}=2\}}]$$
\begin{lem}\label{lem:aMAXb}
    Suppose $b$ is convex. For all $a$ and all $\sM\in M(\mu,\nu)$ we have $P(\sM;a,b)=P(\sM;a\vee b,b)$. 
\end{lem}

\begin{proof}
Fix $\sM \in M(\mu,\nu)$. Define $\hat{\sT}_{1,2} = \{ \tau \in \sT_{1,2} : \PP( \{\tau=1\} \cap \{a(Z_1)<b(Z_1)\}=0\}$ and $\hat{P}(\sM,a,b) = \sup_{\hat{\tau} \in \hat{\sT}_{1,2}} \E^{\sM}[ a (Z_1)I_{\{\hat{\tau}=1\}}+b(Z_2)I_{\{\hat{\tau}=2\}}]$.

Fix $\tau \in \sT_{1,2}$ and define $A=A_\tau = \{\tau=1\} \cap \{a(Z_1)<b(Z_1)\} 
$ and $\hat{\tau}=\hat\tau(\tau)$ by $\hat{\tau}=\tau$ on $A^c$ and $\hat{\tau}=2$ on $A$ (note that $A\in\sF_1$ and $\hat\tau\in\sT_{1,2}$). Then, 
$\{ \tau=1 \} \setminus A = \{ \hat{\tau} = 1 \}$ and
\begin{eqnarray*} 
a(Z_1)I_{\{\tau=1\}}+b(Z_2)I_{\{\tau=2\}}
& = & a(Z_1)I_{ \{ \tau=1 \} \setminus A} + a(Z_1) I_{A}  +b(Z_2)I_{\{\tau=2\}} \\
& \leq & a(Z_1)I_{\{ \hat{\tau}=1 \} } + b(Z_1) I_{A}  +b(Z_2)I_{\{\tau=2\}}.
\end{eqnarray*}
Then, by the convexity of $b$, $\E^{\sM}[b(Z_1) I_{A}] \leq \E^{\sM}[b(Z_2)I_A]$ and then, 
\[ \E^{\sM}[ a(Z_1)I_{\{\tau=1\}}+b(Z_2)I_{\{\tau=2\}}]
\leq \E^{\sM}[ a(Z_1)I_{\{\hat{\tau}=1\}}+b(Z_2)I_{\{\hat{\tau}=2\}}] \]
with the inequality being strict if $\PP(A)>0$.
It follows that
\begin{equation}
    \label{eq:hattau}
\sup_{\tau \in \sT_{1,2}} \E^{\sM}[ a(Z_1)I_{\{\tau=1\}}+b(Z_2)I_{\{\tau=2\}}]
\leq \sup_{\hat{\tau} \in \hat{\sT}_{1,2}} \E^{\sM}[ a (Z_1)I_{\{\hat{\tau}=1\}}+b(Z_2)I_{\{\hat{\tau}=2\}}] .
\end{equation}
Since the reverse inequality is trivial we must have equality in \eqref{eq:hattau}.
It is clear that for $\hat{\tau} \in \hat{\sT}_{1,2}$ we have that, except on a set of measure zero, $a(Z_1) I_{ \{ \hat{\tau}=1 \} } = (a \vee b)(Z_1) I_{ \{ \hat{\tau}=1 \} }$. Therefore, 
\begin{eqnarray*}
    \hat{P}(\sM;a,b) & = & {\sup_{\hat{\tau} \in \hat{\sT}_{1,2}}} \E^{\sM}[ a (Z_1)I_{\{\hat{\tau}=1\}}+b(Z_2)I_{\{\hat{\tau}=2\}}]  \\
    & = & \sup_{\hat{\tau} \in \hat{\sT}_{1,2}} \E^{\sM}[ (a \vee b) (Z_1)I_{\{\hat{\tau}=1\}}+b(Z_2)I_{\{\hat{\tau}=2\}}] = \hat{P}(\sM;a \vee b,b).
\end{eqnarray*}
Then, with the outer equalities both following from applications of the version of \eqref{eq:hattau} with equality, we conclude that
\[ P(\sM;a,b) = \hat{P}(\sM;a,b) = \hat{P}(\sM;a \vee b,b) = P(\sM; a \vee b, b). \]
\end{proof}
Now we turn to the dual problem and the hedging cost.
\begin{lem}\label{lem:aMAXbdual}
    Suppose $b$ is convex. For all $a$ we have $\sD(\mu,\nu;a,b)=\sD(\mu,\nu;a\vee b,b)$. 
\end{lem}

\begin{proof} 
By Theorem~\ref{thm:simpledual} we have
$\sD(\mu,\nu;a,b) 
= \tilde{\sD}(\mu,\nu;a,b)
$. Similarly, ${\sD}(\mu,\nu;a \vee b,b) = \tilde{\sD}(\mu,\nu;a \vee b,b)$.

Recall that $\tilde{\sD}(\mu,\nu;a,b) = \inf_{\psi \in \tilde{\sS}(b)} \int (a - \psi)^+ d \mu + \int \psi d \nu$. 
Suppose $\psi \in \tilde{S}(b)$ so that
$\psi \geq b$. Then on $\{a<b\}$ we have $(a-\psi)^+ = 0 = (b-\psi)^+$ and hence $(a-\psi)^+ = (a \vee b - \psi)^+$ everywhere.
It follows that $\tilde{\sD}(\mu,\nu;a,b) = \inf_{\psi \in \tilde{\sS}(b)} \int (a \vee b - \psi)^+ d \mu + \int \psi d \nu
= \tilde{\sD}(\mu,\nu;a \vee b,b)$. The result now follows.
\end{proof}

\begin{theorem}
    \label{thm:5}
    Suppose $a \vee b$ and $b$ are convex. Then 
\[ \E^{\sM^{*}}[ a(Z_1) I_{ \{ \tau^*=1 \} } +  b(Z_2) I_{ \{ \tau^*=2 \} } ]   
=\int (a- \psi^{*})^+ d\mu+\int \psi^{*} d\nu, \]
where, depending on which Case (C1)-(C3) {arises from the laws $\mu$ and $\nu$} and payoffs $a\vee b$ and $b$, $(\sM^{*},\tau^*)$ and $\psi^*$ are chosen to be as in Sections \ref{sec:C1}, \ref{sec:C2} and  \ref{sec:C3}, respectively.

It follows that $(\sM^{*},\tau^*)$ gives the highest model-based price of the Bermudan option, and $\psi^{*}$ generates the cheapest superhedge.
There is no duality gap.
\end{theorem}

\begin{proof}
    It follows from Lemmas~\ref{lem:aMAXb} and \ref{lem:aMAXbdual}, and the fact that there is no duality gap for the problem with payoffs $(a \vee b, b$), that $\sP(\mu,\nu,a,b)=\sD(\mu,\nu,a,b)$. The fact that the stated models (together with associated stopping times) and the stated superhedging strategies are optimal, then also follows from their optimality in the case with payoffs $(a \vee b,b)$.
\end{proof}
\subsection{Stopping at time-0}

From a financial point of view, it is natural to allow immediate exercise for Bermudan (or American) options; of course the initial distribution of the underlying asset is trivial (i.e., $Z_0\sim\delta_{\bar\nu}$). On the other hand, in the earlier parts of the paper the set $\sT$ of admissible stopping rules did not include $t=0$. In this subsection, by allowing {the option holder} to stop immediately, but restricting the stopping rules to take values in $\{0,1\}$, we show how to recover the highest no-arbitrage price of the Bermudan option, together with the cheapest superhedging strategy. In particular, our goal is to solve
$$
\sup_{\sM \in M_0(\mu=\delta_{\bar\nu},\nu)} \sup_{\tau \in \sT_{0,1}} \E^{\sM}[ a(Z_0)I_{\{\tau=0\}}+b(Z_1)I_{\{\tau=1\}}],
$$
where $M_0(\mu=\delta_{\bar\nu},\nu)$ is the set of consistent models (i.e., filtered probability spaces $(\Omega, \sF, \bF, \QQ)$) supporting a stochastic process $Z=(Z_0,Z_1)$ such that $Z$ is a $(\bF,\QQ)$-martingale with given marginals $Z_0 \sim \mu =\delta_{\bar\nu}$ and $Z_1 \sim \nu$, and $\sT_{0,1}$ is the set of $\bF$-stopping times taking values in $\{0,1\}$. 
{We still assume that $b$ is convex and $\nu$ is continuous 
on $(\alpha_\nu,\beta_\nu)$.}


Note that, by considering candidate stopping rules $\tau_0=0$ and $\tau_1=1$, we have that
\begin{equation}
    \label{eq:t01}
\sup_{\sM \in M_0(\mu=\delta_{\bar\nu},\nu)} \sup_{\tau \in \sT_{0,1}} \E^{\sM}[ a(Z_0)I_{\{\tau=0\}}+b(Z_1)I_{\{\tau=1\}}]\geq a(\bar\nu)\vee\int b(y)\nu(dy).
\end{equation}
On the other hand, the right hand side is attained if we only consider models with (canonical) filtration generated by $Z$. Indeed, in this case $\sF_0$ is trivial (since $Z_0\sim\mu={\delta_{\bar\nu}}$), and thus $\sT_{0,1}=\{0,1\}$. We now show that, by considering a richer probabilistic structure (similarly as in Subsection \ref{sec:C2}), the inequality in \eqref{eq:t01} is strict.

{Suppose that $b$ is convex (and not linear). Further suppose that $b(\bar\nu)<a(\bar\nu)$ (otherwise, under any model, it is always optimal to stop at time-1) and $a(\bar{\nu})<\sup_{\beta_l < x < \bar\nu < y < \beta_r} \left\{ \frac{y-\bar\nu}{y-x} b(x) + \frac{\bar\nu - x}{y-x} b(y) \right\}$ (otherwise, under any model, it is always optimal to stop at time-0). Let $(f,g)$ with $f < \bar\nu < g$ solve
\begin{equation}
\label{eq:fgdef}    
\int_f^g (y - \bar\nu) \nu(dy) = 0;
\hspace{10mm}  \frac{b(g)-b(f)}{g-f} = \frac{a(\bar\nu) - b(f)}{\bar\nu-f}. 
\end{equation}
Set $\Lambda = \Lambda_{\nu,a,b} = \frac{b(g)-b(f)}{g-f}$. Then also $\Lambda = \frac{b(g) -a(\bar\nu)}{g-\bar\nu}$.

Let $L(x)= a(\bar\nu) + \Lambda(x- \bar\nu)$
and set $\psi=\max\{b,L\}$. Then by construction $\psi = L$ on $[f,g]$ and $\psi = b$ on $[\alpha_\nu,f] \cup [g,\beta_\nu]$ (note that $\psi(\bar\nu)=L(\bar\nu)=a(\bar\nu)$).
Further, $\psi$ is convex and $\psi\geq b$ on $\R$, and thus $\psi$ generates a superhedge with total cost
\begin{eqnarray*}
\int(a-\psi)^+d\delta_{\bar\nu}+\int\psi d\nu = \int\psi d\nu & = & \int_{(\alpha_\nu,f]\cup[g,\beta_\nu)}b(y)\nu(dy)+ \int_f^g L(y) \nu(dy)  \\
& =&\int_{(\alpha_\nu,f]\cup[g,\beta_\nu)}b(y)\nu(dy)+a(\bar\nu)\nu((f,g))
\end{eqnarray*}
where we use the first part of \eqref{eq:fgdef} to rewrite the final term.
Note that $\int\psi d\nu>a(\bar\nu)\vee\int bd\nu$.

Then, the optimal model is obtained by stopping a $\nu((f,g))$ amount of mass at time-0 at location $\bar\nu$ (this is achieved by working with additional uniform random variable $U$, and a stopping time $\tau$ such that $\{\tau=0\}=\{U\leq\nu((f,g))\}$), while the remaining $1-\nu((f,g))$ mass at $\bar\nu$ is mapped to $\nu\lvert_{\{(\alpha_\nu,f]\}\cup\{[g,\beta_\nu)\}}$. It follows that
\begin{align*}
\sup_{\sM \in M_0(\mu=\delta_{\bar\nu},\nu)} \sup_{\tau \in \sT_{0,1}} &\E^{\sM}[ a(Z_0)I_{\{\tau=0\}}+b(Z_1)I_{\{\tau=1\}}]\\&=\int_{(\alpha_\nu,f]\cup[g,\beta_\nu)}b(y)\nu(dy)+a(\bar\nu)\nu((f,g))> a(\bar\nu)\vee\int b(y)\nu(dy).
\end{align*}
}
Again, models with canonical filtration are not rich enough.
\begin{remark}
    In fact, the conclusions of this section hold for arbitrary (not necessarily continuous) $\nu$. Indeed, let $R,S:(0,1)\to \R$ be the supporting functions of the lifted left-curtain martingale coupling (see, for example, Hobson and Norgilas \cite{HobsonNorgilas:19x}). Then, $R$ (resp., $S$) is non-increasing (resp., non-decreasing), and for each $u\in(0,1)$ {there exists $g\in[S(u-),S(u+)]$, $f\in[R(u-),R(u+)]$, $\chi_g \in [0, \nu( \{g \})]$ and $\chi_f \in [0, \nu( \{f \})]$ such that
    $\int_{(f,g)} d\nu + \chi_g + \chi_f = u$ and 
    $\int_{(f,g)} y d\nu + g \chi_g + f \chi_f = u \bar{\nu}$;
    this} is equivalent to the the first part of \eqref{eq:fgdef} (note that the `flat' sections and the jumps of either $S$ or $R$ (or both) correspond to the atoms and intervals of no-mass for $\nu$, respectively).
    
    For $x<y$ let $L_{x,y}$ be the line that goes through $(x,b(x))$ and $(y,b(y))$. Then, either $R(0+)<S(0+)$ and $L^1:=L_{R(0+),S(0+)}(\bar\nu)\geq a(\bar\nu)$, or there exists $u^*\in(0,1)$ and $f\in[R(u^*+),R( u^*-)]$, $g\in[S(u^*-),S(u^*+)]$ such that $L^2:=L_{f,g}(\bar\nu)=a(\bar\nu)$ (so that $(f,g)$ satisfies the second part of \eqref{eq:fgdef}).

    In the first case we take $\psi=\max\{b,L^1\}$. Then, {since $\psi \geq a(\bar{\nu})$,} the total superhedging cost is $\int\psi d\nu=\int bd\nu$ and thus any model $\sM$ (together with a stopping time $\tau=2$) is optimal.
    
    In the second case, $u^*\geq \nu((f,g))$ is such that, for some $0\leq\chi_f\leq\nu(\{f\})$ and $0\leq\chi_g\leq\nu(\{g\})$, we have $u^*\delta_{\bar\nu}\leq_{cx}\nu^*:=\nu\lvert_{(f,g)}+\chi_f\delta_f+\chi_g\delta_g$. By taking $\psi=\max\{b,L^2\}$ (note that $\psi(z)=b(z)=L^2(z)$ for $z\in\{f,g\}$) we obtain a superhedge with total cost $\int\psi d\nu=\int\psi d(\nu-\nu^*)+\int L^2d\nu^*=\int\psi d(\nu-\nu^*)+a(\bar\nu)u^*$. Then the optimal model is obtained by stopping an amount of mass $u^*=\nu^*([f,g])=\nu^*(\R)$ at time-0 at location $\bar\nu$ (again, this is achieved by working with an additional uniform random variable $U$, and a stopping time $\tau$ such that $\{\tau=0\}=\{U\leq u^*)\}$), while the remaining $(1-u^*)$ mass at $\bar\nu$ is mapped to $(\nu-\nu^*)$. Note that, of the $\nu$-mass at $f$ at time-1, only an amount $(\nu(\{f\})-\chi_f)$ is `exercised' at time-1---the {other $\chi_f$ amount was exercised at time zero}; and similarly for $g$.
\end{remark}

\bibliographystyle{plainnat}

\appendix
\section{Proofs}\label{appendix}
\subsection*{Proof of Theorem \ref{thm1}}
{\em Case 1: 4-point construction.}
Write $\Psi_4$ as shorthand for $\Psi^{y_-,x_-,x_+,y_+}_{4,a,b}$.
{Since $(y_-,x_-,x_+,y_+) \in \sL^4_{a,b}$,}
on $[y_-,y_+]$, $\Psi_4=T^{y_-,y_+}_{b,b}$ is the straight-line passing through $(y_-,b(y_-))$, $(x_-,a(x_-))$, $(x_+,a(x_+))$ and $(y_+,b(y_+))$, {at least if $a$ is continuous at $x_-$ and $x_+$}. Then, by the convexity of $b$, we have that $\Psi_4\geq b$ on $[y_-,y_+]$ and thus $\Psi_4\geq b$ on $\R$. It follows that $\Psi_4\in\tilde\sS(b)$. Also note that, by the convexity of $a$, we have that $\Psi_4\geq a$ on $[x_-,x_+]$ and $\Psi^{y_-,x_-,x_+,y_+}\leq a$ on $\R\setminus[x_-,x_+]$. To ease the notation, set $\psi=\Psi_4=\Psi^{y_-,x_-,x_+,y_+}_{4,a,b}$, $\phi=(a-\psi)^+$ and $T=T^{y_-,y_+}_{b,b}$. Then
\begin{eqnarray*}
\lefteqn{\int_\R\phi(u)\mu(du)+\int_\R\psi(v)\nu(dv)} \\
&=&\int_{(\alpha_\mu,y_-)\cup(y_+,\beta_\mu)}(a(u)-b(u))\mu(du)+\int_{(y_-,x_-)\cup(x_+,y_+)}(a(u)-T(u))\mu(du)\\
&& \quad +\int_{(\alpha_\nu,y_-)\cup(y_+,\beta_\nu)}b(v)\nu(dv)+\int_{(y_-,y_+)}T(v)\nu(dv)\\
&=&\int_{(\alpha_\mu,x_-)\cup(x_+,\beta_\mu)}a(u)\mu(du)+\int_{(\alpha_\nu,y_-)\cup(y_+,\beta_\nu)}b(v)(\nu-\mu)(dv)\\
&& \quad +\left\{\int_{(y_-,y_+)}T(v)\nu(dv)-\int_{(y_-,x_-)\cup(x_+,y_+)}T(u)\mu(du)\right\}.
\end{eqnarray*}
Now, by assumption, $\Pi_M^{y_-,x_-,x_+,y_+}(\mu,\nu)\neq\emptyset$. Then, by Lemma \ref{lem:GammaPsi4}, $(y_-,x_-,x_+,y_+)\in\Gamma^{\mu,\nu}_4$. Hence (see the definition of $\Gamma^{\mu,\nu}_4$ given in \eqref{eq:Gamma4def}) we must have that $\mu\leq\nu$ on $(-\alpha_\nu,y_-)\cup(y_+,\beta_\nu)$, so that, on $(-\alpha_\nu,y_-)\cup(y_+,\beta_\nu)$, $(\nu-\mu)$ is a well-defined (non-negative) measure. 

We also argue that the final bracketed term vanishes. Indeed, since $T$ is a line, it is enough to show that  $\mu|_{(y_-,x_-) \cup(x_+,y_+)}$ and  $\nu|_{(y_-,y_+)}$ have the same total mass and mean.  But this is immediate from the fact that $\mu|_{(y_-,x_-) \cup(x_+,y_+)} \leq_{cx} \nu|_{(y_-,y_+)}$ (here we again use that $(y_-,x_-,x_+,y_+)\in\Gamma^{\mu,\nu}_4$ together with \eqref{eq:Gamma4def}).

Finally, note that, under each $\pi\in\Pi_M^{y_-,x_-,x_+,y_+}(\mu,\nu)$, the mass in $(x_-,x_+)$ is mapped to $(-\alpha_\nu,y_-)\cup(y_+,\beta_\nu)$ and $\mu\lvert_{(x_-,x_+)}\leq_{cx} (\nu-\mu)\lvert_{(-\alpha_\nu,y_-)\cup(y_+,\beta_\nu)}$; see Lemma \ref{lem:musthave4} and \eqref{eq:4points3}. In particular, $(\nu - \mu)(dv)\left[I_{ \{ v \geq y_+ \} } +I_{ \{ v \leq y_- \}} \right] = \int_{w\in(x_-,x_+)} \pi_w(dv) \mu(dw) $. It follows that
\begin{eqnarray*}
\lefteqn{\int_\R\phi(u)\mu(du)+\int_\R\psi(v)\nu(dv)} \\
&=&\int_{(\alpha_\mu,x_-)\cup(x_+,\beta_\mu)}a(u)\mu(du)+\int_{(\alpha_\nu,y_-)\cup(y_+,\beta_\nu)}b(v)(\nu-\mu)(dv)\\
& = & \E^{\sM_\pi}[a(Z_1) I_{ \{ \tau^{x_-,x_+}=1 \} } + b(Z_2) I_{ \{ \tau^{x_-,x_+} = 2 \} }],
\end{eqnarray*}
which finishes the proof in this case.

{\em Case 2: 5-point construction.}
Since, by assumption, $\Psi_5 := \Psi_{5,a,b}^{y_-,x_-,z,x_+,y_+}$ is convex, it is continuous at each $w\in\{y_-,x_-,z,x_+,y_+\}$, and therefore $\Psi_5(y_-) 
=b(y_-)=T^{y_-,x_-}_{b,a}(y_-)$, $\Psi_5
(z)=T^{y_-,x_-}_{b,a}(z)=T^{x_+,y_+}_{a,b}(z)$ and $\Psi_5(y_+)
=b(y_+)=T^{x_+,y_+}_{a,b}(y_+)$. Then, since $a$ is convex and (again by assumption) $\Psi_5
(z) \geq a(z)$, we have that $\Psi_5
\geq a\geq b$ on $(x_-,x_+)$ and $\Psi_5\leq a$ on $\R\setminus(x_-,x_+)$. By applying a similar argument (together with the convexity of $b$) we also have that $\Psi_5
\geq b$ on $(y_-,y_+)$ (recall that $\Psi_5=b$ on $\R\setminus(y_-,y_+)$). It follows that $\Psi_5
\in\tilde\sS(b)$. To ease the notation, set $\psi=\Psi_5 = \Psi_{5,a,b}^{y_-,x_-,z,x_+,y_+}$, $\phi=(a-\psi)^+$, $T_-=T^{y_-,x_-}_{b,a}$ and $T_+=T^{x_+,y_+}_{a,b}$. Then
		\begin{eqnarray*}
		\int_\R\phi(w)\mu(dw)&= &\int_{\alpha_\mu}^{y_-}(a(w)-b(w))\mu(dw)+\int_{y_-}^{x_-}(a(w)-T_-(w))\mu(dw)\\
		&& \hspace{5mm} +\int_{x_+}^{ y_+}(a(w)-T_{+}(w))\mu(dw)+\int_{y_+}^{\beta_\mu}(a(w)-b(w))\mu(dw)
		\end{eqnarray*}
		and
\[		\int_\R\psi(v)\nu(dv)=\int_{\alpha_\nu}^{y_-}b(v)\nu(dv)+\int_{y_-}^{z} T_{-}(v)\nu(dv)
+\int_{z}^{ y_+}T_{+}(v)\nu(dv)+\int_{y_+}^{\beta_\nu} b(v)\nu(dv).
\] 
		It follows that
\begin{eqnarray*}
\lefteqn{\int_\R\phi(w)\mu(dw)+\int_\R\psi(v)\nu(dv) }  \\
	&=&\int_{(\alpha_\mu,x_-)\cup(x_+,\beta_\mu)}a(w)\mu(dw)+\int_{(\alpha_\nu,y_-)\cup( y_+,\beta_\nu)}b(v)(\nu-\mu)(dv)\\
		&& \hspace{20mm}+\left\{\int_{y_-}^{z}  T_{-}(v)\nu(dv)-\int_{y_-}^{x_-}T_{-}(w)\mu(dw)\right\} \\
		&& \hspace{20mm} 
        + \left\{\int_{z}^{ y_+} T_+(v)\nu(dv)-\int_{x_+}^{y_+} T_{+}(w)\mu(dw)\right\}.
		\end{eqnarray*}
		However, similarly as in Case 1, the bracketed terms in the last two lines vanish, whilst $(\nu-\mu)$ is a well-defined (non-negative) measure on $\R\setminus(y_-,y_+)$. For the latter, note that, since $\Pi_M^{y_-,x_-,z,x_+,y_+}(\mu,\nu)\neq\emptyset$, by Lemma \ref{lem:GammaPsi5} we have that $(y_-,x_-,z,x_+,y_+)\in\Gamma^{\mu,\nu}_5$. Then by \eqref{eq:Gamma5def} it follows that $\mu\leq\nu$ on $(-\alpha_\nu,y_-)\cup(y_+,\beta_\nu)$ (see also \eqref{eq:5points1}). For the former, again using that $(y_-,x_-,z,x_+,y_+)\in\Gamma^{\mu,\nu}_5$ (together with \eqref{eq:Gamma5def}), we have both $\mu|_{(y_-,x_-)}\leq_{cx} \nu|_{(y_-,z)}$ and  $\mu|_{(x_+,y_+)}\leq_{cx}\nu|_{(z,y_+)}$, and thus the bracketed terms vanish due to the linearity of $T_-$ and $T_+$. Also note that, under each $\pi\in\Pi_M^{y_-,x_-,z,x_+,y_+}(\mu,\nu)$, the mass in $(x_-,x_+)$ is mapped to $(-\alpha_\nu,y_-)\cup(y_+,\beta_\nu)$ and $\mu\lvert_{(x_-,x_+)}\leq_{cx} (\nu-\mu)\lvert_{(-\alpha_\nu,y_-)\cup(y_+,\beta_\nu)}$; see \eqref{eq:5points4} and Lemma \ref{lem:musthave5}. In particular, $(\nu - \mu)(dv)\left[I_{ \{ v \geq y_+ \} } +I_{ \{ v \leq y_- \}} \right] = \int_{w\in(x_-,x_+)} \pi_w(dv) \mu(dw) $.
        
        Finally, fix $\pi\in\Pi^{y_-,x_-,z,x_+,y_+}_M(\mu,\nu)$ and let $\sM_\pi\in M^{can}(\mu,\nu)$ be the corresponding canonical model. Also, recall the definition of the candidate stopping time $\tau^{x_-,x_+}\in\sT^{can}_{1,2}$, see \eqref{eq:tau4point}. Then
\begin{eqnarray*}
\lefteqn{\int_\R\phi(w)\mu(dw)+\int_\R\psi(v)\nu(dv)}\\ 
& = & \int_{w \in (\alpha_\mu,x_-) \cup (x_+,\beta_\mu)} a(w) \pi_w(dv) \mu(dw)  +
\int_{w \in (x_-,x_+)} \int_{v\in\R} b(v) \pi_w(dv) \mu(dw)  \\
& = & \E^{\sM_\pi}[a(Z_1) I_{ \{ \tau^{x_-,x_+}=1 \} } + b(Z_2) I_{ \{ \tau^{x_-,x_+} = 2 \} }],
\end{eqnarray*}
which finishes the proof in this case.
\subsection*{Proof of Lemma \ref{lem:couplingsDispersionAssumption}}

      To prove the claim, one needs to construct the functions $(f^L,g^L)$ and $(f^R,g^R)$ under the Dispersion Assumption, and then show that they satisfy the required properties. See, for example, Henry-Labord\`ere and Touzi \cite[Section 3.4]{HenryLabordereTouzi:16} where the authors construct the supporting functions by solving a system of coupled ODEs. In particular, under the Dispersion Assumption, define $f^L(x)=x=g^L(x)$ on $(\alpha_\mu,e^L]$, while on $(e^L,\beta_\mu)$ let $(f^L,g^L)$ be the solutions to \eqref{eq:fgode}. 
    
    Alternatively, a more general construction via potential functions is provided by Hobson and Norgilas \cite{HobsonNorgilas:22}. Fix $x\in(\alpha_\mu,\beta_\mu)$.  By Theorem \ref{thm:curtain}, the left-curtain coupling embeds $\mu_{\alpha_\mu}^x$  into $\nu$ via the shadow measure $S^\nu(\mu_{\alpha_\mu}^x)$. Then the main insight of Hobson and Norgilas \cite{HobsonNorgilas:22} is that we can recover the locations $(f^L(x),g^L(x))$ from the potential function $P_{S^\nu(\mu_{\alpha_\mu}^x)}$. In particular, by Proposition \ref{prop:potential} we have that $P_{\nu-S^\nu(\mu_{\alpha_\mu}^x)}=(P_\nu-P_{\mu_{\alpha_\mu}^x})^c$. In the case $(P_\nu-P_{\mu_{\alpha_\mu}^x})^c(x)=(P_\nu-P_{\mu_{\alpha_\mu}^x})(x)$ one sets $f^L(x)=x=g^L(x)$. On the other hand, if $(P_\nu-P_{\mu_{\alpha_\mu}^x})^c(x)<(P_\nu-P_{\mu_{\alpha_\mu}^x})(x)$, then $(P_\nu-P_{\mu_{\alpha_\mu}^x})^c$ is linear and satisfies $(P_\nu-P_{\mu_{\alpha_\mu}^x})^c<(P_\nu-P_{\mu_{\alpha_\mu}^x})$ on some open interval $I_x\ni x$. In this case $(f^L(x),g^L(x))=I_x$, i.e., the locations of the supporting functions coincide with the end-points of $I_x$. 

    Under the Dispersion Assumption, we have that $\mu\leq\nu$ (and thus also $\mu_{\alpha_\mu}^x\leq\mu\leq\nu$ for $x\leq e^L$) on $(\alpha_\mu,e^L]$. Then, for $x\in(\alpha_\mu,e^L]$, $(P_\nu-P_{\mu_{\alpha_\mu}^x})$ is convex and therefore $(P_\nu-P_{\mu_{\alpha_\mu}^x})^c=(P_\nu-P_{\mu_{\alpha_\mu}^x})$. It follows that $f^L(x)=x=g^L(x)$. On the other hand, for $x\in(e^L,\beta_\mu)$, we have that $(P_\nu-P_{\mu_{\alpha_\mu}^x})^c(x)<(P_\nu-P_{\mu_{\alpha_\mu}^x})(x)$ (this follows from the fact that $(P_\nu-P_{\mu_{\alpha_\mu}^x})$ is strictly convex and increasing on $(\alpha_\mu,e^L]$, strictly concave on $(e^L,e^R\wedge x)$ and again strictly convex on $(e^R\wedge x,\beta_\mu)$; note that $(P_\nu-P_{\mu_x})$ is strictly positive and continuous on $(\alpha_\nu,\beta_\nu)$). In this case $(f^L(x),g^L(x))=I_x$.  One can show that $f^L,g^L$ are indeed continuous and strictly monotone (since the densities $\rho,\eta$ are strictly positive on the respective intervals). The case of the right-curtain martingale coupling can be proved using symmetric arguments.
\subsection*{Proof of Lemma \ref{lem:new:FormOfShadow}}
Since (by assumption) $y<\beta_\mu$, we have that either $x<(g^L)^{-1}(\beta_\mu)$ or $(g^L)^{-1}(\beta_\mu)\leq x<y<\beta_\mu =B(x)$. Further note that $y>B(x)$ (resp., $y=B(x)$) is equivalent to $g^L(x)< f^R(y)$ (resp., $g^L(x)=f^R(y)$).

We first prove Case (i).  Note that by Corollary \ref{cor:LCandRC} we have that $S^\nu(\mu_{\alpha_\mu}^x)=\mu_{\alpha_\mu}^{f^L(x)} + \nu_{f^L(x)}^{g^L(x)}$ and $S^{\nu}(\mu_y^{\beta_\mu})= \nu_{f^R(y)}^{g^R(y)} + \mu_{g^R(y)}^{\beta_\nu}$. On the other hand, by the associativity of the shadow measure (recall  Proposition~\ref{prop:shadow_assoc}),  we have that $S^\nu(\bar\mu^y_x)=S^\nu(\mu_{\alpha_\mu}^x+\mu_y^{\beta_\mu})=S^\nu(\mu_{\alpha_\mu}^x)+S^{\nu-S^\nu(\mu_{\alpha_\mu}^x)}(\mu_y^{\beta_\mu})$, and thus we are left to prove that $S^{\nu-S^\nu(\mu_{\alpha_\mu}^x)}(\mu_y^{\beta_\mu})=S^{\nu}(\mu_y^{\beta_\mu})$. For this, observe that, since $g^L(x)\leq f^R(y)$, the supports of the shadow measures $S^\nu(\mu_{\alpha_\mu}^x)$ and $S^\nu(\mu_y^{\beta_\mu})$ are disjoint, and therefore $S^{\nu}(\mu_y^{\beta_\mu})\leq (\nu-S^\nu(\mu_{\alpha_\mu}^x))$. Then, by applying Lemma \ref{lem:specialMeasures1} with $\xi=\mu_y^{\beta_\mu}$, $\eta=(\nu-S^\nu(\mu_{\alpha_\mu}^x))$ and $\chi=\nu$, we obtain $S^{\nu-S^\nu(\mu_{\alpha_\mu}^x)}(\mu_y^{\beta_\mu})=S^{\nu}(\mu_y^{\beta_\mu})$, as claimed. Finally, $\mu^y_x\leq_{cx}(\nu-\mu)^{f^L(x)}_{\alpha_\nu}+\nu_{g^L(x)}^{f^R(y)}+(\nu-\mu)_{g^R(y)}^{\beta_\nu}$ is a direct consequence of the associativity of the shadow measure (Proposition~\ref{prop:shadow_assoc}), since $\mu^y_x\leq_{cx}(\nu-S^\nu(\bar\mu^y_x)$. This completes the proof in Case (i).

Note that, in Case (i) and with $y=B(x)$ (i.e., $g^L(x)=f^R(y)$), we have that $S^\nu(\bar{\mu}_x^y) = \mu_{\alpha_\mu}^{f^L(x)} + \nu_{f^L(x)}^{g^R(y)} + \mu_{g^R(y)}^{\beta_\nu}$. Then by taking $w=f^L(x)$ and $z=g^R(y)$, we immediately obtain the result in Case (ii) with $y=B(x)$.

We now prove Case (ii) for $y<B(x)$; note that in this case $x$ can take values in $(\alpha_\mu,\beta_\mu)$. First, by the associativity of the shadow measure we have that $S^\nu(\bar\mu^y_x)=S^\nu\left(\bar\mu_x^{B(x)}+\mu_y^{B(x)}\right)=S^\nu(\bar\mu_x^{B(x)})+S^{\nu-S^\nu(\bar\mu_x^{B(x)})}(\mu_y^{B(x)})$. Now, by the first part of the proof, $S^\nu(\bar\mu_x^{B(x)})=\mu_{\alpha_\mu}^{f^L(x)}+\nu_{f^L(x)}^{g^R(B(x))}+\mu_{g^R(B(x))}^{\beta_\mu}$ and $S^{\nu-S^\nu(\bar\mu_x^{B(x)})}(\mu_y^{B(x)})=S^{(\nu-\mu)_{\alpha_\nu}^{f^L(x)}+(\nu-\mu)_{g^R(B(x))}^{\beta_\nu}}(\mu_y^{B(x)})$. Note that $f^L(x)\leq x<y<B(x)\leq g^R(B(x))$, and therefore $(\nu-\mu)_{\alpha_\nu}^{f^L(x)}+(\nu-\mu)_{g^R(B(x))}^{\beta_\nu}$ and $\mu_y^{B(x)}$ have disjoint supports. Then by Lemma \ref{lem:specialMeasures2}, $S^{\nu-S^\nu(\bar\mu_x^{B(x)})}(\mu_y^{B(x)})=(\nu-\mu)_{w}^{f^L(x)}+(\nu-\mu)_{g^R(B(x))}^{z}$for some $w<f^L(x)\leq e^L<e^R\leq g^R(B(x))<z$ (uniqueness follows from the fact that, due to the Dispersion Assumption, $(\nu-\mu)_{\alpha_\nu}^{f^L(x)}+(\nu-\mu)_{g^R(B(x))}^{\beta_\nu}$ admits density that is strictly positive on $(\alpha_\nu,f^L(x))\cup(g^R(y),\beta_\nu)$). It follows that $S^\nu(\bar\mu^y_x)=\mu_{\alpha_\mu}^w + \nu_w^z + \mu_{z}^{\beta_\nu}$ and (again by the associativity of the shadow measure) $\mu^y_x\leq_{cx}(\nu-\mu)^{w}_{\alpha_\nu}+(\nu-\mu)_{z}^{\beta_\nu}$, which finishes the proof.

\subsection*{Proof of Lemma \ref{lem:Ccont}}
We only prove the claim regarding $C^L$ (the case for $C^R$ uses symmetric arguments). By definition, for each $x\in(e^L,\beta_\mu)$, $C^L( g^L(x)) = a(x) + \frac{g^L(x) - x}{x-f^L(x)} (a(x) - b(f^L(x)))$, and the continuity of $C^L$ on $(e^L,\beta_\nu)$ follows from the continuity of $f^L,g^L,a,b$ on $(e^L,\beta_\mu)$.

Since $a\geq b$ on $(\alpha_\mu,\beta_\mu)$ and $b$ is convex (and finite on $(\alpha_\nu,\beta_\nu)\supseteq(\alpha_\mu,\beta_\mu)$), we have that for all $x\in(e^L,\beta_\mu)$,
$$
C^L(g^L(x))\geq a(x)+(g^L(x)-x)\frac{b(x)-b(f^L(x))}{x-f^L(x)}\geq a(x)+(g^L(x)-x)b'(f^L(x));
$$ 
recall that $b'(x)$ denotes the right derivative of $b$ at $x$. 
Choose $x_1 \in (e^L,\beta_\mu)$; then for $x<x_1$, $f^L(x)>f^L(x_1)$ and $b'(f^L(x)) \geq b'(f^L(x_1))$.
Furthermore, $b$ is finite on $(\alpha_\nu,\beta_\nu)$, and therefore $b'(x)\in\R$ for all $x\in(\alpha_\nu,\beta_\nu)$. Then, since $e^L\in
(\alpha_\nu,\beta_\nu)$, 
$$
\liminf_{z\downarrow e^L}C^L(z)=\liminf_{x\downarrow e^L}C^L(g^L(x))\geq [\liminf_{x\downarrow e^L} a(x)]+b'(f^L(x_1))[\liminf_{x\downarrow e^L}(g^L(x)-x)]=a(e^L),
$$
which finishes the proof.
\subsection*{Proof of Lemma \ref{lem:a=betc}}
We only prove the first statement. Suppose that $a(e^L) > b(e^L)$. Then, 
     since $C^L( g^L(x)) = a(x) + \frac{g^L(x) - x}{x-f^L(x)} (a(x) - b(f^L(x)))$, by the continuity of $f^L$ at $e^L$, there exists $\epsilon>0$ such that for $e^L<x<e^L+\epsilon$, $a(x) > b(f^L(x))$. Then $C^L( g^L(x)) > a(x)$ on $(e^L,e^L+\epsilon)$ and $\hat{z}^L>e^L+\epsilon$. 
     It follows that if $\hat z^L=e^L$ then $a(e^L)=b(e^L)$.

Now suppose that $a(e^L) = b(e^L)$. If $c$ is a convex function and if $w>v>u$ then $c(w) \geq c(v) + (w - v) \frac{c(v)-c(u)}{v-u}$. Thus, defining $c$ by $c(w) = b(w)I_{ \{ w<e^L \}} + a(w) I_{ \{ w \geq e^L \} }$ and fixing  $x\in(e^L,(g^L)^{-1}(e^R)]$ (so that $g^L(x)\in(e^L,e^R]$), we have that $a(g^L(x)) {\geq} a(x) + (g^L(x)-x) \frac{a(x)-b(f^L(x))}{x - f^L(x)} = C^L(g^L(x))$. Hence $\hat{z}^L \leq g^L(x)$. Since $x\in(e^L,(g^L)^{-1}(e^R)]$ is arbitrary and $\lim_{x \downarrow e^L}g^L(x) = e^L$ we have that $\hat{z}^L=e^L$.

The final statement then follows from the fact that $f^L$ and $g^L$ are continuous and $g^L(e^L)=e^L$.
\subsection*{Proof of Lemma \ref{lem:PropertiesofsA}}
It follows from Lemma~\ref{lem:new:FormOfShadow}(ii) that, for each $(x,y) \in \sA^<$, $S^\nu(\bar{\mu}_x^y) = \mu_{\alpha_\mu}^{w} + \nu_w^z + \mu_z^{\beta_\mu}$ for some (unique) $w \leq e^L\wedge x < e^R\vee y \leq z$.  We set $w=w(x,y)$ and $z=z(x,y)$.

Next we prove the continuity, monotonicity and limit properties of $w$ and $z$ in $x$ and $y$. We fix $x$ and consider varying $y \in \sA^<_x$. The argument for fixed $y$ and varying $x\in\sA^<_y$ is similar.

Suppose $x <y< y' \leq B(x)\wedge\beta_\mu$. Then by the associativity property of the shadow measure (Proposition~\ref{prop:shadow_assoc})
\[  S^\nu( \bar{\mu}_x^y) = S^\nu( \bar{\mu}_x^{y'}) + S^{\nu - S^\nu(\bar{\mu}_x^{y'})}(\mu_y^{y'}). \]
Moreover, by Lemma~\ref{lem:new:FormOfShadow}(ii),
\[ S^\nu( \bar{\mu}_x^y) = \mu_{\alpha_\mu}^{w(x,y)} + \nu_{w(x,y)}^{z(x,y)}  + \mu_{z(x,y)}^{\beta_\mu} \quad\textrm{and} \quad S^\nu( \bar{\mu}_x^{y'}) = \mu_{\alpha_\mu}^{w(x,y')} + \nu_{w(x,y')}^{z(x,y')}  + \mu_{z(x,y')}^{\beta_\mu}; \]
where both $w(x,y) \leq e^L\wedge x<e^R\vee y \leq z(x,y)$ and $w(x,y') \leq e^L\wedge x<e^R\vee y' \leq z(x,y')$.
It follows that $\nu - S^\nu(\bar{\mu}_x^{y'})=(\nu-\mu)_{\alpha_\nu}^{w(x,y')}+(\nu-\mu)_{z(x,y')}^{\beta_\nu}$ and $\mu^{y'}_y$ have disjoint supports, and thus, by Lemma \ref{lem:specialMeasures2}, $S^{\nu - S^\nu(\bar{\mu}_x^{y'})}(\mu_y^{y'})=(\nu-\mu)^{w(x,y')}_{\hat w}+(\nu-\mu)^{\hat z}_{z(x,y')}$ for some $\hat w < w(x,y')$ and $\hat z>z(x,y')$. Hence
$$
\mu_{\alpha_\mu}^{w(x,y)} + \nu_{w(x,y)}^{z(x,y)}  + \mu_{z(x,y)}^{\beta_\mu}=S^\nu( \bar{\mu}_x^y)=S^\nu( \bar{\mu}_x^{y'}) + S^{\nu - S^\nu(\bar{\mu}_x^{y'})}(\mu_y^{y'})=\mu_{\alpha_\mu}^{\hat w} + \nu_{\hat w}^{\hat z}  + \mu_{\hat z}^{\beta_\mu},
$$
and therefore $w(x,y)=\hat w < w(x,y')$ and $z(x,y)=\hat z>z(x,y')$, as required.

For the right-continuity, fix $y$ and consider the limit $y' \downarrow y$. We know that
\[  S^{\nu - S^\nu(\bar{\mu}_x^y)}(\mu_y^{y'}) = (\nu-\mu)_{w(x,y)}^{w(x,y')} + (\nu-\mu)_{z(x,y')}^{z(x,y)}. \]
As $y' \downarrow y$ we see that $S^{\nu - S^\nu(\bar{\mu}_x^y)}(\mu_y^{y'})(\R)=\mu_y^{y'}(\R) \downarrow 0$ and hence $(\nu-\mu)_{w(x,y)}^{w(x,y')}(\R) \downarrow 0$ and $(\nu-\mu)_{z(x,y')}^{z(x,y)}(\R) \downarrow 0$. Since, by the Dispersion Assumption (recall Definition \ref{def:dispersion}), $(\nu-\mu)$ admits a strictly positive density on $(\alpha_\nu,e^L)\cup(e^R,\beta_\nu)$, we find that $w(x,y') \downarrow w(x,y)$ and $z(x,y') \uparrow z(x,y)$.
For the left-continuity we can let $y \uparrow y'$ and apply a similar argument. Hence $w(x, \cdot)$ and $z(x,\cdot)$ are continuous.

Finally, still fixing $x$, note that if $y \downarrow x$ then $(\nu-S^\nu(\bar\mu_x^y))(\R)={\mu}_x^y(\R)\downarrow0$. But $(\nu-S^\nu(\bar\mu_x^y))=(\nu-\mu)_{\alpha_\nu}^{w(x,y)} + (\nu-\mu)_{z(x,y)}^{\beta_\nu}$, and therefore $ (\nu-\mu)_{\alpha_\nu}^{w(x,y)}(\R)\downarrow0$ and  $(\nu-\mu)_{z(x,y)}^{\beta_\nu}(\R)\downarrow 0$. Then, again using that $(\nu-\mu)$ admits a strictly positive density on $(\alpha_\nu,e^L)\cup(e^R,\beta_\nu)$, we get that $w(x,y) \downarrow \beta_\nu$ and $z(x,y) \uparrow \beta_\nu$. 

\subsection*{Proof of Proposition \ref{prop:Lambday}}
Fix $(x,y')\in\sA^<$ such that $\Lambda(x,y')=0$, and note that $T_{a,a}^{x,y'}$ is well-defined and $T_{a,a}^{x,y'}(z(x,y'))\leq b(z(x,y'))$. Furthermore, the convexity of $a$ ensures that $a\leq T_{a,a}^{x,y'}$ on $[x,y']$ and $a\geq T_{a,a}^{x,y'}$ on $\R\setminus [x,y']$. Similarly, by the convexity of $a$ and $b$, and the fact that $a\geq b$, we have that $b\geq T^{a,a}_{x,y'}$ on $(z(x,y'),\infty)$; also, if $\Lambda(x,y')=0$, then $b\leq T^{a,a}_{x,y'}$ on $[x,z(x,y')]$.

We first prove the monotonicity of $y\to\Lambda(x,y)$ on $(x,y')$. Since $b$ is convex, $a\geq b$, and $T_{a,a}^{x,y'}(z(x,y'))\leq b(z(x,y'))$, we have that $b\geq T_{a,a}^{x,y'}$ and (more importantly) $(T_{a,a}^{x,y'}-b)$ is {non-increasing} on $(z(x,y'),\infty)$. Then since $z(x,\cdot)$ non-increasing on $(x,y')$ (recall Lemma \ref{lem:PropertiesofsA}) it follows that $y\mapsto\Lambda(x,y)=T_{a,a}^{x,y'}(z(x,y))-b(z(x,y))$ is non-decreasing on $(x,y')$. (Here, in the case $y'=\beta_\mu$, we use $T_{a,a}^{x,y'}(z)=a(\beta_\mu-)+\frac{z-\beta_\mu}{\beta_\mu-x}(a(\beta_\mu-)-a(x))$ and $\hat\Lambda(x,\beta_\mu,z)=T_{a,a}^{x,y'}(z)-b(z)$ as in \eqref{eq:Lambdabeta}.) 

We now suppose that $\Lambda(x,y')=0$. We will only prove that $\Lambda(x,y)>0$  for $y \in (y',B(x){]}$; that $\Lambda(x,y)<0$ for $y<y'$ can be proved in an identical fashion. 

Assume that $y'<B(x)$, else there is nothing to prove.
Fix $y\in(y',B(x)]$ and note that $x<z(x,y)<z(x,y')$. Then $a(y)\geq T_{a,a}^{x,y'}(y)$ (from which it follows that $T_{a,a}^{x,y}\geq T_{a,a}^{x,y'}$ on $[x,\infty)$) and $b(z(x,y))\leq T_{a,a}^{x,y'}(z(x,y))$. It follows that $\Lambda(x,y)=T_{a,a}^{x,y}(z(x,y))-b(z(x,y))\geq 0$. It is left to prove that $\Lambda(x,y)=0$ cannot happen.

Suppose that $\Lambda(x,y)=0$. Then we must have that $T_{a,a}^{x,y}(z(x,y))=T_{a,a}^{x,y'}(z(x,y))=b(z(x,y))$, and therefore $T_{a,a}^{x,y}=T^{a,a}_{x,y'}$. However, due to the convexity of $a$ (resp., convexity of $b$ and the fact that $a\geq b$) this is only possible if $a=T_{a,a}^{x,y'}$ on $[x,y]$ (resp., $b=T_{a,a}^{x,y'}$ on $[x,z(x,y')]$). It follows that $a=T_{a,a}^{x,z'}=b$ on $(x,y)\cap(x,z(x,y'))\neq\emptyset$, which contradicts Standing Assumption~\ref{sass:payoffs}. We conclude that $\Lambda(x,y)>0$.

\subsection*{Proof of Proposition \ref{prop:LambdaxB}}
For $\bar{x} \leq x \leq (g^L)^{-1}(\beta_\mu)$ we have $z(x,B(x))=B(x)$ and $\Lambda(x,B(x)) = \hat{\Lambda}(x,B(x),B(x))= a(B(x)) - b(B(x)) \geq 0$. In order to complete the proof it is sufficient to consider $x \in [x_0, \bar{x})$.

 Observe that, for $x\in[x_0,\bar x)$ we have that $x<B(x)<e^R<z(x,B(x))$, and therefore both $T_{a,a}^{x,B(x)}$ and $T_{a,b}^{B(x),z(x,B(x))}$ are well-defined. Then $\Lambda(x,B(x))=0$ if and only if $T_{a,a}^{x,B(x)}=T_{a,b}^{B(x),z(x,B(x))}$ (or equivalently, $\tilde S_{a,a}^{x,B(x)}=\tilde S_{a,b}^{B(x),z(x,B(x))}$), and $\Lambda(x,B(x))>0$ if and only if $T_{a,a}^{x,B(x)}>T_{a,b}^{B(x),z(x,B(x))}$ on $[B(x),\infty)$ (or equivalently, $\tilde S_{a,a}^{x,B(x)}>\tilde S_{a,b}^{B(x),z(x,B(x))}$).

By construction, since we are in either Case $(a)$ or Case $(b)$, we have that $\tilde S_{a,a}^{x_0,y_0}\geq\tilde S_{a,b}^{y_0,z(x_0,y_0)}$, and thus $\Lambda(x_0,y_0=B(x_0))\geq0$.  We are left to prove that $\Lambda(x,B(x)) \geq  0$ for $x\in(x_0,\bar x)$.

By the definition of $B$, we have that $w(x,B(x))=f^L(x)$ and $z(x,B(x))=g^R(B(x))$, for all $x\in(\alpha_\mu,(g^L)^{-1}(\beta_\mu))$. In particular, $w(\cdot,B(\cdot))$ is increasing (resp., decreasing) on $(\alpha_\mu,e^L)$ (resp., $(e^L,\beta_\mu)$), while $z(\cdot,B(\cdot))$ is decreasing (resp., increasing) on $(\alpha_\mu,B^{-1}(e^R))$ (resp.,  $(B^{-1}(e^R),\beta_\mu)$)

We show that $\tilde S_{a,a}^{x,B(x)} \geq \tilde S_{a,b}^{B(x),z(x,B(x))=g^R(B(x))}$ for $x\in(x_0,\bar x)$. It then follows that $\Lambda(x,B(x))\geq 0$ for $x\in(x_0,\bar x)$.  Since $B(\cdot)$ is non-decreasing, by the convexity of $a(\cdot)$ we have that $\tilde S_{a,a}^{x_0,y_0=B(x_0)}\leq\tilde S_{a,a}^{x,B(x)}$ for all $x\in[x_0,(g^L)^{-1}(\beta_\mu))$. On the other hand, on $[x_0,\bar x)$, $x\mapsto z(x,B(x))=g^R(B(x))$ is non-increasing, and therefore, using that $a\geq b$ and the convexity of $a(\cdot)$ and $b(\cdot)$, we have that, $\tilde S_{a,b}^{y_0=B(x_0),z(x_0,y_0)=g^R(B(x_0))}\geq \tilde S_{a,b}^{B(x),z(x,B(x))=g^R(B(x))}$ on $[x_0,\bar x)$. It follows that, for all $x\in[x_0,\bar x)$, $\tilde S_{a,b}^{B(x),z(x,B(x))=g^R(B(x))}\leq \tilde S_{a,b}^{y_0=B(x_0),z(x_0,y_0)=g^R(B(x_0))}\leq\tilde S_{a,a}^{x_0,y_0=B(x_0)}\leq\tilde S_{a,a}^{x,B(x)}$, and thus $\Lambda(x , B(x))\geq 0$.

\subsection*{Proof of Corollary \ref{cor:existenceContinuity}}
We have that $\Lambda(x,B(x))\geq 0$ for all $x \in \sA^\Lambda_{>}\cup\sA^\Lambda_{=}=([x_1, \tilde x)\setminus\sA^\Lambda_{<})$. Then the existence of $y^*(x)$ follows from the definition of $\tilde x$ and the intermediate value theorem, while the uniqueness follows from Proposition~\ref{prop:Lambday}. We are left to prove the continuity.

We first claim that $\Lambda(x,y)<0$ for all $x\in[x_1,\tilde x)$ and $y\in(x,y^*(x))$. Indeed, $\Lambda(x,y^*(x))\leq0$ for all $x\in[x_1,\tilde x)$, and thus the claim follows by Proposition \eqref{prop:Lambday}. 

Fix $x\in[x_1,\tilde x)$ and $y\in(x,y^*(x))$. Then $\Lambda(x,y)<0$. By the continuity of $\Lambda(\cdot,y)$, there exists a small enough $\epsilon\in(0,y\wedge\tilde x-x)$ such that $\Lambda(x',y)<0$ for all $x'\in[x,x+\epsilon)$. Recall that $B$ is non-decreasing. If $x'\in\sA^\Lambda_>\cup\sA^\Lambda_=$, then $y<y^*(x')\leq B(x')$. If $x'\in\sA^\Lambda_<$, then $y<y^*(x)\leq B(x)\leq\beta_\mu=y^*(x')$. In either case we have that $y<y^*(x')$ for all $[x,x+\epsilon)$, and therefore $\liminf_{x'\downarrow x}y^*(x')\geq y$. Since  $y\in(x,y^*(x))$ was arbitrary, we conclude that $\liminf_{x'\downarrow x}y^*(x')\geq y^*(x)$.

We now show that $\limsup_{x'\downarrow x}y^*(x')\leq y^*(x)$. Fix $x\in\sA^\Lambda_>$ so that $y^*(x) < B(x)$. Choose $y\in(y^*(x),B(x))$ so that $\Lambda(x,y)>0$. Then by the continuity of $\Lambda(\cdot,y)$, there exists a small enough $\epsilon\in(0,(y\wedge\tilde x-x))$ such that $\Lambda(x',y)>0$ for all $x'\in[x,x+\epsilon)$. Then $y^*(x')<y$ for all $x'\in[x,x+\epsilon)$ and thus $\limsup_{x' \downarrow x} y^*(x') \leq y$. But $y\in (y^*(x),B(x))$ was arbitrary, and thus $\limsup_{x' \downarrow x} y^*(x') \leq y^*(x)$. Now  suppose that $x\in\sA^\Lambda_=\cup\sA^\Lambda_<$, so that $y^*(x)=B(x)$. Then, for $x'\in(x,\tilde x)$, we have $y^*(x') \leq B(x')$, and since $B$ is continuous, $\limsup_{x' \downarrow x} y^*(x') \leq \limsup_{x' \downarrow x} B(x') = B(x) = y^*(x)$. 

We conclude that $\liminf_{x' \downarrow x} y^*(x')=\limsup_{x' \downarrow x} y^*(x')=y^*(x)$, which proves the right-continuity on $[x_1,\tilde x)$. The argument for the left-continuity on $(x_1,\tilde x)$ is similar.

\subsection*{Proof of Lemma \ref{lem:tildeX}}
Since $x< y^*(x)\leq B(x)$ for all $x\in[x_1,\tilde x)$, the claim  immediately follows if $\tilde x=\beta_\mu$. In the rest of the proof we assume that $\tilde x<\beta_\mu$.

By Corollary \ref{cor:existenceContinuity}, $x<y^*(x)$ for $x\in[x_1,\tilde x)$, and therefore $\limsup_{x \uparrow \tilde{x}}y^*(x)\geq\liminf_{x \uparrow \tilde{x}}y^*(x)\geq\tilde{x}$.
To prove that we have equalities throughout, we argue by contradiction. Suppose $y=\limsup_{x \uparrow \tilde{x}}y^*(x)>\tilde{x}$. Note that, since $\tilde x<\beta_\mu$ and $y^*\leq B$ on $[x_1,\tilde x)$, we have that $y\leq\lim_{x\uparrow\tilde x}B(x)=B(\tilde x)<\infty$. 

Let $(x_n)_{n\geq 1} \uparrow \tilde{x}$ be such that $(y^*(x_n))_{n\geq 1} \to y$. Then 
$$
0 \geq \Lambda(x_n,y^*(x_n))=T_{a,a}^{x_n, y^*(x_n)}(z(x_n,y^*(x_n))) - b(z(x_n,y^*(x_n)));
$$
note that strict inequality is possible in the case $x_n\in\sA^\Lambda_<$, since then $y^*(x_n)=B(x_n)=\beta_\mu$ and, by definition, $\Lambda(x,B(x))<0$ for all $x\in\sA^\Lambda_<$. Also, $x_n<y^*(x_n)$ for each $n\geq 1$ (resp., $\tilde x<y$), and thus $T_{a,a}^{x_n,y^*(x_n)}$ (resp.,  $T_{a,a}^{\tilde x,y}$) is well-defined. Taking limits, and using that $z(x_n,y^*(x_n)) \rightarrow z(\tilde{x},y)$ by the joint continuity of $z(\cdot,\cdot)$, we get
$$ 0 \geq T_{a,a}^{\tilde{x}, y}(z(\tilde{x},y)) - b(z(\tilde{x},y)) = \Lambda(\tilde{x},y).
$$
Since $y>\tilde{x}$, using Proposition~\ref{prop:Lambday} we have that $\Lambda(\tilde{x},y') <0$ for all $y' \in (\tilde{x},y)$. By the continuity of $\Lambda(\cdot,y')$, there exists (small enough) $\epsilon>0$ such that $\Lambda(x,y') <0$ for all $x \in [\tilde{x},\tilde{x}+\epsilon)\subset[\tilde x,\beta_\mu)$. Then (again by Proposition~\ref{prop:Lambday}) $\Lambda(x,x+) <0$ for all $x \in [\tilde{x},\tilde{x}+\epsilon)$, contradicting the definition of $\tilde{x}$.
\end{document}